\documentclass[11pt]{article}
\usepackage{fullpage}
\usepackage{float}
\usepackage{epsfig}
\usepackage{amsmath}
\usepackage{amssymb}
\usepackage{amstext}
\usepackage{amsthm}
\usepackage{xspace}
\usepackage{latexsym}
\usepackage{verbatim}
\usepackage{multirow}
\usepackage{ifthen}
\usepackage{url}
\usepackage{natbib}
\usepackage{subcaption}
\usepackage[linesnumbered,ruled]{algorithm2e}
\usepackage{mdframed}
\usepackage{dsfont}
\usepackage{breakcites}
\usepackage[utf8]{inputenc}
\definecolor{linkcol}{rgb}{0,0,0.9} 
\definecolor{citecol}{rgb}{0.8,0,0}
\renewcommand{\algorithmcfname}{ALGORITHM}

\DeclareMathAlphabet{\mathitbf}{OML}{cmm}{b}{it}
\RequirePackage[colorlinks,citecolor=red,urlcolor=blue]{hyperref}
\definecolor{gray230}{RGB}{240,240,240}
\newtheorem{theorem}{Theorem}[section]
\newtheorem{definition}{Definition}
\newtheorem{observation}{Observation}

\newtheorem{lemma}[theorem]{Lemma}

\newtheorem{example}{Example}

\renewcommand{\comment}[1]{}

\definecolor{gray230}{RGB}{255,255,255}

\SetCommentSty{mycommfont}
\newcommand{\prob}[2][]{\text{\bf Pr}\ifthenelse{\not\equal{}{#1}}{_{#1}}{}\!\left[#2\right]}
\newcommand{\expect}[2][]{\text{\bf E}\ifthenelse{\not\equal{}{#1}}{_{#1}}{}\!\left[#2\right]}

\newcommand{\argmax}{\operatorname{argmax}}

\newcommand{\agind}[1][i]{_{#1}}

\newcommand{\inversed}[1]{#1^{-1}}
\newcommand{\ironed}{\bar}

\newcommand{\differentiated}[1]{#1'}
\newcommand{\fortype}{\tilde}
\newcommand{\estimated}{\hat}
\newcommand{\sampled}{\hat}

\newcommand{\noaccents}[1]{#1}
\newcommand{\composed}[3]{#1{#2{#3}}}

\newcommand{\newindexedvar}[4][\noaccents]{%
\expandafter\newcommand\expandafter{\csname #2\endcsname}{#1{#4}}%
\expandafter\newcommand\expandafter{\csname #2s\endcsname}{#1{\boldsymbol{#4}}}%
\expandafter\newcommand\expandafter{\csname #2sm#3\endcsname}[1][#3]{#1{\boldsymbol{#4}}_{-##1}}%
\expandafter\newcommand\expandafter{\csname #2#3\endcsname}[1][#3]{#1{#4}\agind[##1]}%
\expandafter\newcommand\expandafter{\csname #2#3th\endcsname}[1][#3]{#1{#4}_{(##1)}}%
}

\newindexedvar{wal}{k}{w}
\newindexedvar[\differentiated]{margwal}{k}{\wal}
\newindexedvar{cumwal}{k}{W}

\newindexedvar[\ironed]{iwal}{k}{\wal}
\newindexedvar[\ironed]{cumiwal}{k}{\cumwal}
\newindexedvar[\composed{\differentiated}{\ironed}]{margiwal}{k}{\wal}

\newindexedvar{yal}{k}{y}
\newindexedvar{cumyal}{k}{Y}
\newindexedvar[\ironed]{iyal}{k}{\yal}
\newindexedvar[\composed{\differentiated}{\ironed}]{margiyal}{k}{\yal}

\newindexedvar{murev}{k}{P}
\newindexedvar[\estimated]{emurev}{k}{\murev}
\newindexedvar[\differentiated]{mumarg}{k}{\murev}

\newindexedvar[\ironed]{imurev}{k}{\murev}
\newindexedvar[\composed{\differentiated}{\ironed}]{imumarg}{k}{\murev}

\renewcommand{\[}{\left[}
\renewcommand{\]}{\right]}

\newindexedvar{val}{i}{v}

\newindexedvar{eff}{i}{e}

\newindexedvar{bid}{i}{b}
\newindexedvar[\estimated]{ebid}{i}{\bid}
\newindexedvar[\sampled]{sbid}{i}{\bid}

\newindexedvar{strat}{i}{s}

\newindexedvar{virt}{i}{\phi}
\newindexedvar[\inversed]{virtinv}{i}{\virt}

\newindexedvar[\ironed]{ivirt}{i}{\virt}

\newindexedvar{dist}{i}{F}

\newindexedvar{dens}{i}{f}

\newindexedvar{price}{i}{p}
\newindexedvar[\fortype]{tprice}{i}{p}

\newindexedvar{alloc}{i}{x}

\newindexedvar[\fortype]{talloc}{i}{\alloc}

\newindexedvar{util}{i}{u}
\newindexedvar[\fortype]{tutil}{i}{u}

\newcommand{\SmallUpperCase}[1]{\textsc{\MakeLowercase{#1}}}
\begin{document}

\title{QUAD: A Quality Aware Multi-Unit Double Auction Framework for IoT-Based Mobile Crowdsensing in Strategic Setting}
\author{ Vikash Kumar Singh\thanks{School of Computer Science and Engineering, VIT-AP University, Amaravati, India, Email: \url{vikash.singh@vitap.ac.in}} \href{https://orcid.org/0000-0003-4221-7622}{\includegraphics[scale=0.0048]{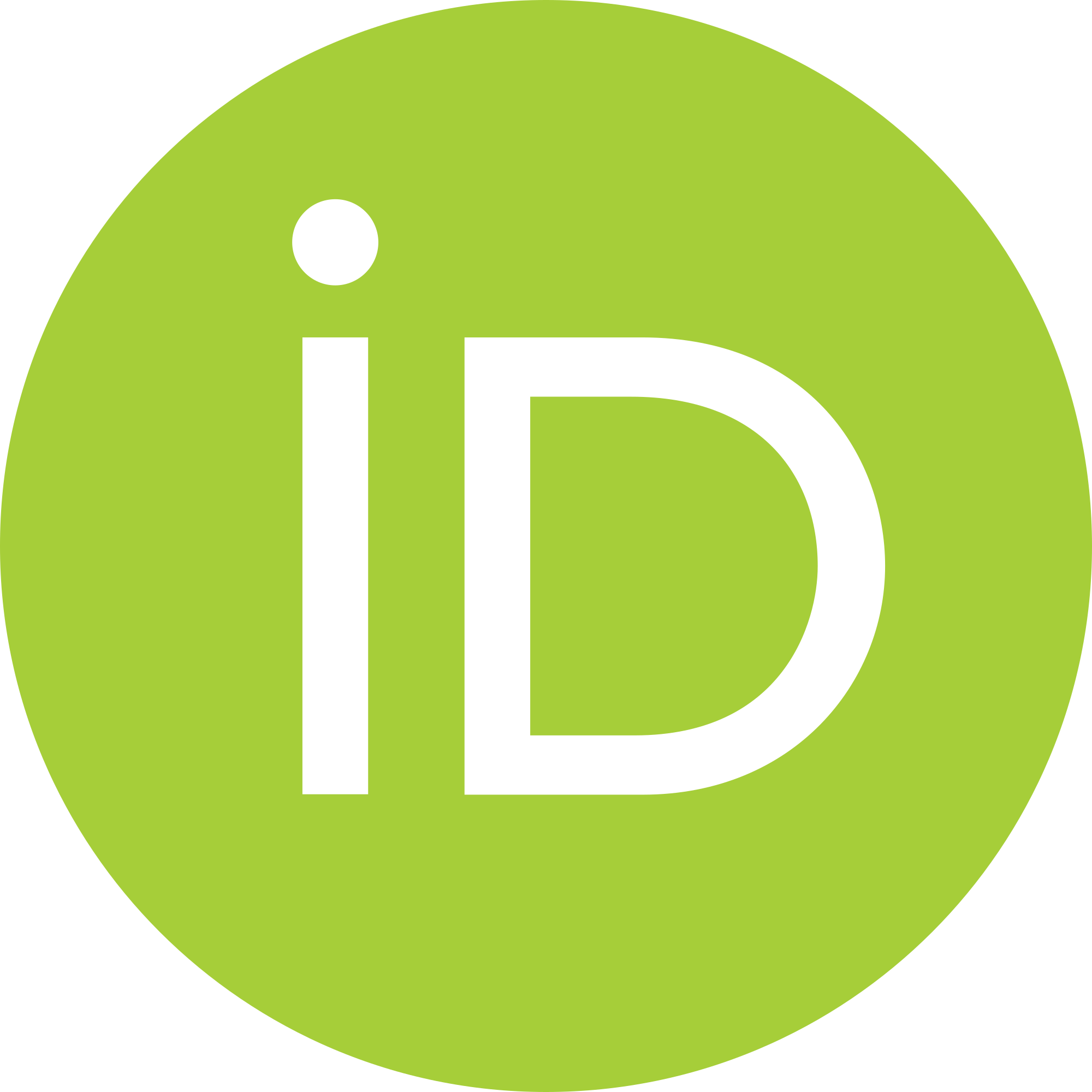}} \and Anjani Samhitha Jasti\thanks{School of Computer Science and Engineering, VIT-AP University, Amaravati, India, Email: \url{samhitha.20bcd7070@vitap.ac.in}} \href{https://orcid.org/0000-0002-1025-1466}{\includegraphics[scale=0.0048]{orcid.png}} \and Sunil Kumar Singh\thanks{School of Computer Science and Engineering, VIT-AP University, Amaravati, India, Email: \url{sunil.singh@vitap.ac.in}} \href{https://orcid.org/0000-0001-5262-6636}{\includegraphics[scale=0.0048]{orcid.png}} \and Sanket Mishra\thanks{School of Computer Science and Engineering, VIT-AP University, Amaravati, India, Email: \url{sanket.mishra@vitap.ac.in}} \href{https://orcid.org/0000-0002-3193-8160}{\includegraphics[scale=0.0048]{orcid.png}} 
}

\maketitle

\begin{abstract}
Crowdsourcing with the intelligent agents carrying smart devices is becoming increasingly popular in recent years. It has opened up meeting an extensive list of real life applications such as measuring air pollution level, road traffic information, and so on. In literature this is known as \emph{mobile crowdsourcing} or \emph{mobile crowdsensing}. In this paper, the discussed set-up consists of multiple task requesters (or task providers) and multiple IoT devices (as \emph{task executors}), where each of the task providers is having multiple homogeneous sensing tasks. Each of the task requesters report bid along with the number of homogeneous sensing tasks to the platform. On the other side, we have multiple IoT devices that reports the \emph{ask} (charge for imparting its services) and the number of sensing tasks that it can execute. The valuations of task requesters and IoT devices are \emph{private} information, and both might act \emph{strategically}. One assumption that is made in this paper is that the bids and asks of the agents (\emph{task providers} and \emph{IoT devices}) follow \emph{decreasing marginal returns} criteria. In this paper, a truthful mechanism is proposed for allocating the IoT devices to the sensing tasks carried by task requesters, that also keeps into account the quality of IoT devices. The mechanism is \emph{truthful}, \emph{budget balance}, \emph{individual rational}, \emph{computationally efficient}, and \emph{prior-free}. The simulations are carried out to measure the performance of the proposed mechanism against the benchmark mechanisms. The code and the synthetic data are available at \textcolor{blue}{\textbf{https://github.com/Samhitha-Jasti/QUAD-Implementation}}. \\

\end{abstract}



\section{Introduction}
\label{s:intro}
Over the past decades, there has been an unprecented growth of the mobile users with smartphones (mobile phones with embedded-sensors). As per the Ericsson mobility report of 2017 \cite{Heuveldop_2017} the number of worldwide mobile subscriptions is growing at around $4\%$ every year, reaching to 9.1 billion in 2022. As the mobile users (or crowd workers) are equipped with sensing devices, so the researchers thought of utilizing these mobile users for sensing and collecting data for several real world applications and then distributing it to the community or organization. For example, measuring the air pollution level across the cities \cite{Pan2017CrowdsensingAQ, ijgi10020046, 9068628}, giving information about the road traffic \cite{STANIEK2021554, 6815170}, noise pollution assessment \cite{doi:10.1177/2399808320987567, 10.1145/1791212.1791226}, information about the potholes \cite{ENIGO2016316, s20195564}, and many more \cite{Nagatani:2013:ERN:2421033.2421037, Poblet2014}. The process of completion of task(s) by the crowd workers or group of common people equipped with sensing devices in the form of an open call give rise to a new paragmatic field of study termed as \emph{mobile crowdsensing} (a.k.a. \emph{mobile crowdsourcing} or \emph{crowdsensing}) \cite{book1, fi14020049, 10.1145/2746285.2746306, Phuttharak2019ARO, s20072055}. The general framework of the mobile crowdsensing consists of three entities: (1) task requester(s) or task provider(s), (2) platform (or third party), and (3) task executors (crowd workers with smart devices). First the task provider(s) provide the sensing tasks to the platform. Once the platform receives the sensing tasks, it is supplied to the task executor(s) that are present on the other side of the crowdsensing market for execution purpose. The task executors equipped with sensing devices completes the tasks and submits the completed tasks to the platform. The platform gives back the completed tasks to the task provider(s). The task executors receives the incentive in return of their services. In mobile crowdsensing market \cite{Singh2019, s20072055, 8570744}, the two challenges that are of major concern in strategic setting\footnote{By strategic, it is meant that the participating agents will try to maximize their utility by mis-reporting their private information. Here, utility is quasi-linear utility, means the difference between the true valuation and the payment. By private it is meant that the bids and asks are only known to the respective agents and not known to others.} are:
\begin{enumerate}
\item Which set of task executors should be selected for task execution purpose?
\item What incentives is to be given to the task executors in exchange of their services?
\end{enumerate}
Most of the works in mobile crowdsensing is carried out answering the questions raised above in points 1 and 2 \cite{Mukhopadhyay2021, Singh_2020, Singh2019, 10.1145/3371425.3371459}, in strategic setting. Apart from the challenges mentioned above, another challenge that persists in mobile crowdsourcing system is to have a large group of crowd workers in crowdsensing market. But, the question is: \emph{how to drag more number of crowd workers to the mobile crowdsensing market}? One of the plausible solutions is to provide incentives to the crowd workers in exchange of their services. In past, several works have been carried out that design the schemes such that the task executors receive incentives in some terms (may be money \cite{Mukhopadhyay2021, Singh_2020, Phuttharak2019ARO, Singh2019}, or some social recognition \cite{Singh_2022}). Another challenging aspect in crowdsensing environment is to get the \emph{quality} data from the task executors \cite{Singh_2020, 10.5555/3061053.3061148, 10.5555/2832249.2832277}. For this purpose, in past, the mechanisms are proposed that along with truthfulness keeps track of \emph{quality} task executors \cite{Singh2019, Mukhopadhyay2021, Singh_2020,10.5555/2832249.2832277}. In \cite{Singh_2020} a truthful mechanism is proposed for the crowdsourcing set-up in combinatorial environment that also ensure that the task requesters receive quality data. In \cite{Mukhopadhyay2021}, a budget feasible truthful mechanism is proposed for the set-up with single task requester and multiple IoT devices (as task executors) in strategic setting. Here, the task requesters have limited budget along with the constraint that the overall budget is not available apriori but arrive in the system on incremental basis in several rounds. The set-up with multiple task requesters and multiple IoT devices with each of the task requesters is having a single task along with the budget, is discussed in \cite{Singh2019}. For this, a truthful mechanism is proposed that also take into account the quality of IoT devices along with satisfying the constraint that the total payment made to the IoT devices is within the budget of the respective task requesters.
\begin{figure*}
                \centering
                \includegraphics[scale=0.92]{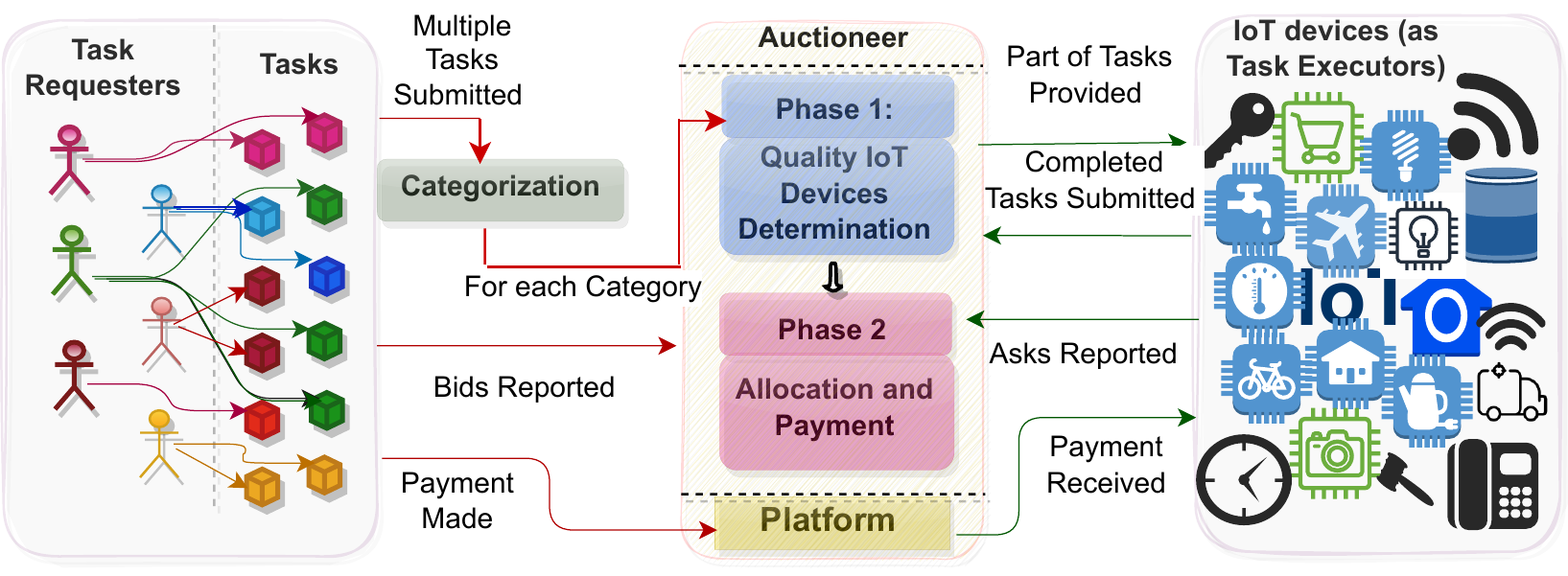}
                \caption{A Multi-unit Double Auction Framework for IoT-based Mobile Crowdsensing}
                \label{fig:2a10}
        \end{figure*}%

\indent Motivated from the above discussed scenarios, in this paper, one of the scenarios of IoT-based mobile crowdsensing is investigated using double auction framework. The detailed overview of the proposed framework is shown in Figure \ref{fig:2a10}. In this, there are multiple task requesters and multiple IoT devices. Each of the task requesters is having multiple homogeneous tasks and each IoT device is capable of performing subset of homogeneous tasks\footnote{For example, one task requester may have tasks of measuring air pollution at different locations, other task requester may have tasks of detecting road traffic through different routes in the city, and so on. On the other hand, a particular set of IoT devices will be capable of measuring the air pollution level, or some may be having capability of getting information about the road traffic and so on. The agents participating in measuring air pollution level belongs to one category and the agents participating in determination of road traffic condition belongs to other category and likewise.}. Based on the type of tasks floated by the task requesters and the type of tasks the IoT devices are capable of performing, the task requesters and IoT devices are categorized into multiple different categories. In each category, the task requesters submit the tasks along with the bids to the platform. On the other side, there are multiple task executors that reveals the number of tasks that they can perform along with the cost (or \emph{ask}) for executing the tasks. The bids and asks of the agents are termed as valuations. The valuation of the agents are \emph{private} information. As there are multiple task requesters and multiple task executors, so the problem can be modeled using double auction. It is to be noted that the proposed framework is a two-phase process. In the first phase, the set of quality IoT devices is determined. For this purpose, the platform gives the part of the available tasks to the IoT devices for execution purpose. Each of the IoT devices execute the tasks and submits the completed tasks back to the platform. Once the platform receives the completed tasks, it distributes the completed tasks to other IoT devices for grading (or ranking) purpose. The IoT devices provide the ranking over the completed tasks and submit the ranks of the tasks to the platform. Based on the ranking received, the set of quality IoT devices is selected among the available IoT devices by the platform. The output of the first phase is a set of quality IoT devices. In the second phase, the platform allocate the tasks to the quality IoT devices for execution purpose. The quality IoT devices execute the tasks and submit it back to the platform. The platform gives the completed tasks to the respective task requesters and the IoT devices get paid in return of their services. The bid and ask values reported by the task requesters and IoT devices respectively are \emph{private} information. \\
\indent In this paper, for the problem under consideration a two-fold mechanism namely \emph{\textbf{qu}ality \textbf{a}ware multi-unit \textbf{d}ouble auction mechanism} (QUAD) is proposed. In the first fold, the set of quality IoT devices is determined using the subroutine \emph{IoT-devices quality determination mechanism} motivated by \cite{NNis_Pre_2007,T.roughgarden_2016}. After that in the second fold from among the set of quality IoT devices, the subset of IoT devices are selected as winners and payment are made using \emph{allocation and pricing rule} motivated by \cite{SegalHalevi2018MUDAAT, T.roughgarden_20141}. In the upcoming subsection the contributions made in this paper are listed.         

\subsection{Contributions}
In the presence of strategic agents, the objective is to design a truthful mechanism that will do the following: (1) determine the quality IoT devices, (2) which IoT devices will execute the tasks?, and (3) what price will be paid to the winning IoT devices?
In particular, our contributions are:
\begin{itemize}
\item One of the scenarios of crowdsensing is investigated in strategic setting. The set-up consists of multiple task requesters and multiple IoT devices with each of the task requesters have multiple sensing tasks that are to be executed. The objective is to assign the floated sensing tasks to the quality IoT devices and to incetivize the IoT devices in return of their services.     

\item A two-fold \emph{truthful} mechanism is proposed namely QUAD for allocating tasks to the IoT devices. In  first fold, a very small part of the tasks from the available tasks is given to the IoT devices for execution purpose for estimating the quality of the IoT devices. It is required as it will help us to assure that the quality sensing data is provided by the IoT devices to the task providers. In second fold, the set of quality IoT devices is considered as winners and the payment is decided by utilizing the concept of mechanism design.

\item Theoretical analysis is carried out in Section \ref{sec:atppm} to show that the proposed mechanism is \emph{computationally efficient} (Lemma \ref{lemma:1a}), \emph{correct} (Lemma \ref{lemma:1b}), \emph{prior free} (Lemma \ref{lemma:1c}) \emph{truthful} (Lemma \ref{lemma:1d}), \emph{individual rational} (Lemma \ref{lemma:IR}), and \emph{budget balance} (Lemma \ref{lemma:BB}). Further probabilistic analysis is carried out to have an estimate on the number of tasks that will be executed by the quality IoT devices for any task provider (Lemma \ref{lemma:1l} and Lemma \ref{lemma:1101}).    

\item In the simulation results it is shown that QUAD is not vulnerable to manipulation and is compared with the benchmark mechanism namely \textbf{p}osted \textbf{p}rice \textbf{m}echanism (PPM) \cite{T.roughgarden_20141} and McAfee double auction (DA) \cite{RPM_The_1992, Jbre_INte_2005}. QUAD is compared with PPM on the ground of \emph{truthfulness} and \emph{budget balance}. Further, QUAD is compared with McAfee DA on the basis of satisfaction level of the agents.
\end{itemize}
\subsection{Paper Organization}
 The remainder of the paper is structured as follows. In Section \ref{sec:rw} related works in the fields of \emph{double auction}, \emph{mobile crowdsourcing}, and IoT is discussed. Section \ref{s:prelim} describes our proposed model and discuss about the game theoretic properties. The proposed mechanism, namely, QUAD is illustrated in section \ref{ref:ehpppm}. In section \ref{sec:atppm} the theoretical and probabilistic analysis of QUAD is carried out. The simulation results are depicted in section \ref{sec:ef}. The paper is concluded with the possible future directions in section \ref{se:conc}. 
\section{Related Prior Works}\label{sec:rw}
In this section, the related works in the areas of double auction, mobile crowdsensing, and Internet of Things (IoT) are discussed one by one in the given order.   

\subsection{Double Auction}
The double auction mechanism is a subroutine in QUAD; therefore, the related works in double auction is discussed in this section. In a double auction, there are multiple sellers who are ready to sell their items,  multiple buyers who are ready to buy the available items, and an auctioneer (or third-party). The double auction gives the flexibility to the two parties (buyers and sellers) to reveal their privately held type simultaneously and can act strategically \cite{RPM_The_1992,NNis_Pre_2007, MYERSON1983265}. A double auction is a mechanism that provide the platform for deciding that \emph{who will buy a particular item and at what price?} It is applicable to several application domains such as \emph{spectrum market} \cite{5062011, 7420720, Leyton-Brown7202}, \emph{Internet advertisement} \cite{DBLP:journals/corr/FeldmanG16}, \emph{emission trading market} \cite{10.1007/s10479-018-2826-y}, etc. For two-sided market, the impossibility result by \cite{MYERSON1983265} states that any mechanism that is incentive compatible (IC), budget balance (BB), and individual rational (IR) cannot maximize Gain-from-trade\footnote{It is defined as the difference between the total valuations of the buyers to the total asks of the sellers.} (or efficiency). \\
\indent In past, several works have been carried out in double auction \cite{PLOTT1990245, RPM_The_1992, 10.1007/3-540-45749-6_34}. In \cite{RPM_The_1992} the set-up is, there are multiple sellers and multiple buyers. Each seller wants to sell a single unit of good and each buyer wants to buy a single unit of good. A truthful mechanism is proposed for matching buyers to sellers. In addition to \emph{truthfulness}, McAfee's mechanism is individual rational, budget balance, and prior-free. By prior-free, it is meant that the mechanism does not make any assumption on the type of valuations the agents reveal. Some double auction mechanism makes assumption on the valuation of the agents such as agent's valuation as \emph{decreasing marginal returns} (DMR) \cite{Blumrosen2014ReallocationM, SegalHalevi2018MUDAAT}, \emph{additive valuation} \cite{doi:10.1287/moor.2021.1124, 45749, DBLP:journals/corr/FeldmanG16}, \emph{gross substitute} \cite{ijcai2018-68} or their valuation is represented by single parameter \cite{Gonen2017DYCOMAD,10.1145/1250910.1250914}. In \cite{Blumrosen2014ReallocationM} a mechanism is proposed for multi parametric agents with DMR valuation but is not asymptotically efficient. The competitive ratio achieved is $\frac{1}{48}$. In \cite{SegalHalevi2018MUDAAT} a truthful double auction mechanism is proposed that allows trading of multiple units of item per agent but with the constraint that the agents have decreasing marginal return valuation. Theoretical analysis shows that the  competitive ratio of the proposed mechanism is at least $1 - O(M\sqrt{\frac{ln~mk}{mk}})$. In \cite{ijcai2018-68} a truthful mechanism is designed for the set-up with multiple sellers and multiple buyers. Each of the sellers have multiple distinct goods and each buyer have infinite amount of money. A double auction based truthful mechanism is proposed that allow trading of multiple units of distinct goods per agents but with the constraint that the agents have gross substitute valuation.\\
\indent In this paper, we have utilized the idea presented in \cite{SegalHalevi2018MUDAAT} for allocating the tasks to the IoT devices and deciding their payment, once the quality IoT devices are determined.        

\subsection{Mobile Crowdsensing}
In order to get an overview on the current trend in mobile crowdsensing the readers can follow \cite{10.1145/3185504, 8570744, s20072055, 8733838, book1}. As discussed above, one of the major challenges in \emph{mobile crowdsensing} is \emph{how to motivate the group of workers carrying smartphones towards the crowdsourcing market?}\\
\indent For this purpose, several incentive based schemes are proposed for incentivizing the crowd workers \cite{Mukhopadhyay2021, 6848055, Yang2012MCN, 10.1145/3371425.3371459, DBLP:journals/tpds/WangGCG18}. In \cite{Mukhopadhyay2021} the set-up with single task requester and multiple task executors in mobile crowdsourcing is investigated in \emph{strategic} setting. Here, the task requester has single task and the limited budget associated with that task. However, the constraint that is preserved in the discussed set-up is that the overall budget is not available apriori but comes in incremental manner in the system in several rounds. Given the above discussed set-up a truthful mechanism is designed that ensures that the total payment made to the task executors is within the avaialble budget. In \cite{6848055} the set-up is such that the platform float the sensing tasks alongwith its location information, and on the other side of the MCS market the smartphone users provide the bids for the sensing tasks falling within their coverage area. The bids of the smartphone users are \emph{private} information. For this purpose, a \emph{reverse auction} based mechanism is proposed namely TRAC. TRAC is a \emph{two-phase mechanism}. In the first phase, a near-optimal approximation algorithm is developed for determining the winners with low computation complexity. The second phase determines the payment of the winners. \cite{DBLP:journals/tpds/WangGCG18} proposed a truthful mechanism that selects the quality crowd workers for completing the tasks. One of the contraint that is preserved in the set-up is that the quality of the crowd workers is changing frequently. In \cite{10.1145/3371425.3371459} auction is coupled with the experience model to have a fair competition among the crowd workers. In \cite{Yang2012MCN} the two separate incentive mechanisms are proposed for user-centric model and platform-centric model respectively. In case of platform-centric model, only one sensing task is floated and the total reward will be shared among the winning mobile phone users. For designing an incentive mechanism Stackelberg game is used, where the utility of the platform is maximized by calculating the unique Stackelberg equilibrium. For the user-centric model an auction-based incentive mechanism is proposed where the mobile phone users can select multiple tasks and can report their bid price to the platform. Based on the reported bid prices, the platform will decide the winners.\\
\indent  Also, in some cases the payment made to the task executors depend on the quality of the work performed by \cite{10.5555/2832249.2832277, 10.5555/3061053.3061148}. In \cite{10.5555/2832249.2832277} the prime focus is to provide the bonuses to the task executors for their exceptional work, so as to improve the overall utility of a task requester. In order to have the information about the impact of bonuses on the quality of work supplied by the task executors, the hidden Markov model is used. In \cite{10.5555/3061053.3061148}, the  output agreement mechanisms are used to have the true answers from the group of task executors. In \cite{7925505} the set-up consists of multiple task consumers and multiple smartphone users, where each of the task consumers have multiple tasks to complete. The objective is to assign the sensing tasks to the smartphone users for execution purpose. For the above discussed set-up a truthful mechanism is proposed that achieves max-min fairness. In \cite{7218592} the focus is on two-sided mobile crowdsourcing market that consists of multiple \emph{service users}, \emph{platform}, and multiple \emph{service providers}. Here, each of the service users can request for single service and each of the service providers can provide single service. The above discussed set-up is modeled using the double auction framework and a \emph{truthful} mechanism is proposed that also satisfies \emph{individual rationality}, and \emph{budget balance}. In \cite{9474925} the focus is on collaborative mobile crowdsourcing, where the tasks are acomplished by the group of IoT devices that communicate among themselves and share their operational activities.\\
\indent From the above discussed literature reviews it can be seen that the set-up discussed in this paper in IoT-based mobile crowdsensing in strategic setting is not considered. In this paper, for the discussed set-up a truthful mechanism is proposed that allocates the tasks provided by the task requesters to the quality IoT devices for execution purpose.                     
\subsection{Internet of Things (IoT)}
For detailed overview in IoT the readers can refer the following research works \cite{10.1145/3185504, ATZORI20102787, RAY2018291, 8614931, 7123563, fi14020049}. The term `Internet of Things' first came into picture in the year 1998 \cite{Ashton1999ThatO}. Later on the International Telecommunication Union (ITU) formally gave the concept of IoT in the year 2005.
\begin{table}[H]
\renewcommand{\arraystretch}{1.3}  
\caption{Several research directions in IoT}
\label{table_example}
\centering
  \begin{tabular}{|c||c|}
\hline
\textbf{Research paper} & \textbf{Research directions in IoT}\\
\hline
\cite{abraham2021ai} & Smart agriculture\\
\hline
\cite{9322470} & Privacy\\
\hline
\cite{9187421} & Smart cities\\
\hline
\cite{8891679} & Privacy and transportation\\
\hline
\cite{siarry2021fusion} & Healthcare\\
\hline
\cite{7823334} & Security and privacy\\
\hline
\cite{8766496} & Healthcare and well being\\
\hline
\cite{TALAVERA2017283} & Industries, agriculture, and environment \\
\hline
\cite{Alavi2018InternetOT} & Transportation, health, and smart cities \\
\hline
\cite{Garca2017AnalysisOS} & Security, privacy, and architecture\\
\hline
\cite{s141019582} & Quality of service, scalability, and iteroperability \\
\hline
\cite{fi7030329} & Security, privacy, and reliability.\\
\hline
\cite{6158307} & Data processing, security, and privacy \\
\hline
\cite{6740844} & Transportation, health, and smart cities\\
\hline
\cite{YAN2014120} & Quality of service (QoS), identification, and authentication \\
\hline
\cite{KHAJENASIRI2017770} & Energy and environment\\
\hline
\end{tabular}
\end{table} 
 It is estimated by the large number of companies and the research professionals that the IoT will contribute 4\%-11\% of global GDP in the year 2025 \cite{Huawei_2015}. Huawei predicts that by 2025 there will be approximatley 100 billion IoT connections \cite{Manyika_2015}. Internet and the advancement of recent technologies have been the catalyst for the research in IoT. The IoT have many application areas but not limited to healthcare \cite{Alavi2018InternetOT}, agriculture \cite{Qiu2013FrameworkAC}, environment \cite{TALAVERA2017283} etc. Table \ref{table_example} depicts the various application areas of IoT.\\   
\indent In the upcoming section the problem of mobile crowdsensing is formulated using double auction framework and is discussed in detailed manner.

\section{Preliminaries}\label{s:prelim}
First the model and notations are discussed. After that game theoretic properties are discussed.
\subsection{Model and Notation}
In this section, the problem is formulated by utilizing the concept of double auction. There are \emph{m} task requesters and \emph{n} IoT devices, such that $m < n$. In this model it is considered that task requesters and IoT devices are heterogenous by nature. By heterogeneity it is meant that the task requesters may float different types of tasks (such as \emph{measuring air pollution}, \emph{road condition}, and so on) and IoT devices may vary in terms of type of tasks they are capable to perform. However, a particular task requester will be endowed with similar type of tasks and particular IoT device will be capable of performing one type of task. Based on the type of tasks the task requesters are floating and the type of tasks the IoT devices can perform, the task requesters and IoT devices are categorized into different categories. Let us say, we have $k$ different categories and is given as $w = \{w_1, w_2, \ldots, w_k\}$, where $w_i$ represents $i^{th}$ category.\\
\indent In any $w_i$ category, say, we have $m_i$ task requesters and $n_i$ IoT devices. The set of task requesters is given as $r^i = \{r_1^i, r_2^i, \ldots, r_{m_i}^i\}$, where $r_j^i$ represents $j^{th}$ task requester in $w_i$ category. The set of IoT devices is given as $\mathcal{I}^i = \{\mathcal{I}_1^i, \mathcal{I}_2^i, \ldots, \mathcal{I}_{n_i}^i\}$, where $\mathcal{I}_k^i$ represents $k^{th}$ IoT device in $i^{th}$ category. The set of task requesters and IoT devices in all the $k$ categories is given as $r = \{r^1, r^2, \ldots, r^k\}$ and $\mathcal{I} = \{\mathcal{I}^1, \mathcal{I}^2, \ldots, \mathcal{I}^k\}$ respectively. The task requesters are endowed with multiple similar tasks  and the bids (maximum price he is ready to pay). Let us say that any $k^{th}$ task requester is endowed with at most $\mathcal{Q}_k^i$ number of similar tasks. On the other side of MCS market, the IoT devices report the number of tasks that they can execute along with the asks (charge for imparting its services). Any $j^{th}$ IoT device can execute at most $\mathcal{Q}_j^i$ number of similar tasks. In our case, the task requesters and IoT devices taken together will be termed as agents and both may behave \emph{strategically}. The bids and asks of the agents will be termed as valuation as and when required.\\
\indent In our set-up the discussed problem is studied as a two-phase process. In the first phase, among the available IoT devices, the set of quality IoT devices is determined for executing the tasks. For this purpose, firstly, a part of available tasks is given to the IoT devices by the platform for execution. On receiving the tasks the IoT devices execute the tasks and submits the completed tasks back to the platform. Now, each of the completed tasks is given to the peers (other \emph{IoT devices}) for the review purpose. In any category $w_i$, if $j^{th}$ IoT device prefers $k^{th}$ IoT device over $l^{th}$ IoT device then it is represented as $\mathcal{I}_k^i \boldsymbol{\succ}_j^i \mathcal{I}_l^i$. For all the IoT devices in $i^{th}$ category the rank profile is represented as $\boldsymbol{\succ}^i = \{\boldsymbol{\succ}_1^i, \boldsymbol{\succ}_2^i, \ldots, \boldsymbol{\succ}_{n_i}^i\}$. For all the $k$ categories, the rank profile is given as $\boldsymbol{\succ} = \{\boldsymbol{\succ}^1, \boldsymbol{\succ}^2, \ldots, \boldsymbol{\succ}^k\}$. After reviewing the completed tasks, the review report is submitted to the platform and the quality IoT devices are determined. Once the quality IoT devices are determined, next, the below mentioned challenges are handled:
\begin{itemize}
\item Which quality IoT devices should be selected for tasks execution purpose?
\item What will be the payment made to the selected quality IoT devices?
\end{itemize}
\indent In the second fold, the above two issues are resolved. In any category $w_i$, each agent \emph{k} has a valuation function $\boldsymbol{\nu}_k^i$ that returns, for every integer \emph{f} ($0 < f \leq \mathcal{Q}_k^i-1$), the agent's value for owing $f$ units. It is to be noted that the valuation for owing \emph{zero} unit is 0. In our discussed setup, all the agents posses DMR valuation function. For any agent \emph{j} in $i^{th}$ category, by DMR valuation function $\boldsymbol{\nu}_j^i$ we mean that $\boldsymbol{\nu}_j^i(f) - \boldsymbol{\nu}_j^i(f-1) \geq \boldsymbol{\nu}_j^i(f+1) - \boldsymbol{\nu}_j^i(f)$. It means that, marginal utility for an agent from having one more task is weakely-decreasing in his current number of tasks. 

\begin{table}[H]
\renewcommand{\arraystretch}{1.3}  
\caption{Notations used}
\label{table_example}
\centering
  \begin{tabular}{c||c}
\hline
\textbf{Symbols} & \textbf{Description}\\
\hline
$m$ & Number of task requesters\\
\hline
$n$ & Number of IoT devices\\
\hline
$w$ & $w = \{w_1, w_2, \ldots, w_k\}$: Set of $k$ different categories\\
\hline
$m_i$ & Number of task requesters in $w_i$ category\\
\hline
$n_i$ & Number of IoT devices in $w_i$ category\\
\hline
$\mathcal{I}$ & $\mathcal{I} = \{\mathcal{I}^1, \mathcal{I}^2, \ldots, \mathcal{I}^k\}$: Set of IoT devices in all the $k$ categories\\
\hline
$\mathcal{I}^i$ & $\mathcal{I}^i = \{\mathcal{I}_1^i, \mathcal{I}_2^i, \ldots, \mathcal{I}_{n_i}^i\}$: Set of IoT devices in $w_i$ category\\
\hline
$\mathcal{I}_k^i$ & $k^{th}$ IoT device in category $w_i$.\\
\hline
$r$ & $r = \{r^1, r^2, \ldots, r^k\}$: Set of task requesters in all the $k$ categories\\
\hline
$r^i$ & $r^i = \{r_1^i, r_2^i, \ldots, r_{m_i}^i\}$: Set of task requesters in $w_i$ category\\
\hline
$r_l^i$ & $l^{th}$ task requester in category $w_i$.\\
\hline
$\boldsymbol{\succ}^i$ & $\boldsymbol{\succ}^i = \{\boldsymbol{\succ}_1^i, \boldsymbol{\succ}_2^i, \ldots, \boldsymbol{\succ}_{n_i}^i\}$: Rank profile of all the IoT devices in $i^{th}$ category.\\
\hline
$\boldsymbol{\succ}_k^i$ & Rank list of $k^{th}$ IoT device in $i^{th}$ category.\\ 
\hline
$\mathcal{Q}_k^i$ & Maximum number of tasks held by any $k^{th}$ agent in $i^{th}$ category\\
\hline
$\boldsymbol{\nu}_j^i$ & Valuation function of agent $j$ in $i^{th}$ category \\
\hline
$\boldsymbol{\nu}_{j,f}^i$ & Agent $j$ marginal value for having $f^{th}$ task in $i^{th}$ category \\
\hline
$p$ & Equilibrium price \\
\hline
$u_j^i (f,~p)$ & Utility of task requester $j$ from buying $f$ units in $i^{th}$ category at price $p$ \\
\hline
$z_j^i (f,~p)$ & Utility of IoT device $j$ from selling $f$ units in $i^{th}$ category at price $p$ \\
\hline
$\boldsymbol{d}_k^i (p)$ & Demand of any task requester $k$ at a given price $p$ in $i^{th}$ category \\
\hline
$\boldsymbol{s}_j^i (p)$ & Supply of any IoT device $j$ at a given price $p$ in $i^{th}$ category \\
\hline
$\boldsymbol{d}^R$ & Total demand at a given price $p$ in right crowdsensing arena in $i^{th}$ category \\
\hline
$\boldsymbol{s}^R$ & Total supply at a given price $p$ in left crowdsensing arena in $i^{th}$ category \\
\hline
\end{tabular}
\end{table}
The reason behind restricting the valuation function to DMR is that, in our discussed set-up the equilibrium price vector\footnote{By \emph{equilibrium price} it is meant that, a price at which the number of tasks floated by the task requesters is equal to the number of tasks executed by the IoT devices. In other words, we say it as supply equals demand.} exists only when we have a DMR valuation. In any category $w_i$, if any agent $k$ have $\mathcal{Q}_k^i$ homogeneous tasks, then it is represented as $\mathcal{Q}_k^i$ single unit virtual agents. The value of virtual-agent $f$ of agent \emph{k} is the agent's marginal value for having the $f^{th}$ unit and is given as $\boldsymbol{\nu}_{k, f}^i = \boldsymbol{\nu}_k^i(f) - \boldsymbol{\nu}_{k}^i (f-1)$ for $f \in \{1, \ldots, \mathcal{Q}_k^i\}$. The ties between the marginal values of the virtual agents are broken randomly. In any category $w_i$, given an equilibrium price $p$ for each task, the utility of task requester $k$ resulted by buying $f$ units of completed task is:
     
\begin{equation}                   
u_k^i(f,p) =
  \begin{cases}
  \boldsymbol{\nu}_k^i(f) - f \cdot p, & \textit{if agent k receives f executed tasks} \\
   0,        & \textit{otherwise}
  \end{cases}
  \end{equation} 
Similarly, the utility of IoT device $j$ by supplying sensed data for $f$ tasks is:
\begin{equation}                   
\boldsymbol{z}_j^{i}(f,p) =
  \begin{cases}
  f \cdot p - (\boldsymbol{\nu}_j^i(\mathcal{Q}_j^i) - \boldsymbol{\nu}_j^i(\mathcal{Q}_j^i-f)), & \textit{if agent j provides sensed data for f tasks} \\
   0,        & \textit{otherwise}
  \end{cases}
  \end{equation}   

\noindent In any category $w_i$, the demand of any task requester $k$ at an equilibrium price $p$ is the set of tasks that maximizes the value $u_k^i(f,p)$ and is given as:
\begin{equation}
\boldsymbol{d}_k^i(p) = \argmax\limits_{f \in [0, \mathcal{Q}_k^i]} u_k^i(f,p)
\end{equation}
Here, if the valuation of any $k^{th}$ task requester is following DMR, then $\boldsymbol{d}_k^i(p)$ is just the number of virtual task requesters whose bid value is greater than equilibrium price $p$ $i.e.$ $\boldsymbol{\nu}_{k,f}^i > p$. The total demand at an equilibrium price $p$ is the sum of the demands of all the task requesters and is given as $\boldsymbol{d}^i = \sum\limits_{k=1}^{m_i} \boldsymbol{d}_k^i(p)$. Similarly, the supply from any $j^{th}$ IoT device at an equilibrium price $p$ is the set of tasks that maximizes the value $z_j^i(f, p)$ and is given as 
\begin{equation}
\boldsymbol{s}_j^i(p) = \argmax\limits_{f \in [0, \mathcal{Q}_j^i]} z_j^i(f,p)
\end{equation}
Here, if the valuation of any $j^{th}$ IoT device is following DMR, then $\boldsymbol{s}_j^i(p)$ is just the number of virtual IoT devices whose ask value is less than equilibrium price $p$ $i.e.$ $\boldsymbol{\nu}_{j,f}^i < p$. The total supply at an equilibrium price $p$ is the sum of the supply from all the IoT devices and is given as $\boldsymbol{s}^i = \sum\limits_{j=1}^{n_i} \boldsymbol{s}_j^i(p)$. Let us try to understand DMR valuation with the help of an example given below.
\begin{mdframed}[backgroundcolor=gray230]
\vspace*{-2mm}
\begin{example}
\label{Example:1}
\emph{In any category $w_i$, the bid value reported by any agent $k$ for 3 units of a task is 10. For $1^{st}$ unit the value is say $5$ $i.e.$ $\boldsymbol{\nu}_{1,1}^i = 5$. Similarly, for $2^{nd}$ unit and $3^{rd}$ unit the value is given as $\boldsymbol{\nu}_{1,2}^i = 4$ and $\boldsymbol{\nu}_{1,3}^i = 1$ respectively. Here, the valuation function $\boldsymbol{\nu}_k^i$ of agent $k$ is DMR, if the equilibrium price is set as $3$ then the demand of $k$ will be $2$ units, as $\boldsymbol{\nu}_{1,1}^i$ and $\boldsymbol{\nu}_{1,2}^i$ values are above equilibrium price. However, if the bid value configuration is such that $\boldsymbol{\nu}_{1,1}^i = 2$, $\boldsymbol{\nu}_{1,2}^i = 5$, and $\boldsymbol{\nu}_{1,3} = 3$ then the demand will be 0 even if the marginal values reported by the virtual agents are above the equilibrium price. The reason is that the valuation function here is not following DMR criteria. Similar argument can be given for IoT devices.}
\end{example}
\end{mdframed}
      
\subsection{Game Theoretic Properties}
In this section, the five game theoretic properties that will be utilized in this paper are discussed.   

  \begin{definition}[\textbf{Prior-free}]
  \label{def:1}
A mechanism is prior-free, if it does not use any statistical information on the valuation of the agents.

  \end{definition}
  \begin{definition}[\textbf{Truthful or Dominant Strategy Incentive Compatible (DSIC)}]
  \label{def:2}
A mechanism is truthful, if no matter how other agents are bidding, no agent j can improve his utility by mis-reporting his valuation $i.e.$ $\hat{u}_j^i(f,p) \leq u_j^i(f,p)$ in case of task requester and $\hat{z}_j^i(f,p) \leq z_j^i(f,p)$ in case of task executor.      
  \end{definition}
    
  \begin{definition}[\textbf{Individual Rationality (IR)}]
  \label{def:IR}
A mechanism is IR, if no winning task executor is paid less than his ask value and no winning task provider  pays more than his bid value. It means that the utility of the agents should be at least 0.     
  \end{definition}

  \begin{definition}[\textbf{Strongly Budget-Balanced (SBB)} ]
  \label{def:3}
A mechanism is SBB, if the total payment of the agents is exactly 0. It means that the mechanism neither incurs surplus nor incurs deficit.        
  \end{definition}

  \begin{definition}[\textbf{Weakly Budget-Balanced (WBB)}]
  \label{def:4}
A mechanism is WBB, if the total payment of the participating agents is positive. It means that the mechanism incurs a surplus.        
  \end{definition}

\section{Proposed DSIC Mechanism}\label{ref:ehpppm}
In this section, a \emph{truthful} mechanism namely QUAD motivated by \cite{NNis_Pre_2007,T.roughgarden_2016, T.roughgarden_20141, SegalHalevi2018MUDAAT} is presented and discussed. QUAD determines the quality of the IoT devices along with deciding which of the quality IoT devices will be hired for task execution? and what will be their payment? The QUAD consists of three components: 
\begin{itemize}
\item \textcolor{blue}{k-category procedure} $-$ The reason for developing \emph{k-category procedure} is to process all the \emph{k} different categories of the task requesters and IoT devices present in MCS.
\item \textcolor{blue}{IoT-devices quality determination mechanism} (motivated by \cite{NNis_Pre_2007,T.roughgarden_2016})$-$ To hire quality IoT devices.
\item \textcolor{blue}{Allocation and pricing rule} (motivated by \cite{SegalHalevi2018MUDAAT, T.roughgarden_20141}). The allocation and pricing rule consists of three subroutine, namely 
\begin{enumerate}
    \item Splitting and equilibrium price determination, 
    \item Demand and supply calculation, and 
    \item Winner determination and payment
\end{enumerate}
\end{itemize}

In the upcoming subsections, the components of the proposed mechanism are discussed.

\subsection{k-category Procedure}\label{subsub:main}
In \emph{k-category procedure}, for each category, in lines 2-5: (a) call to subroutine IoT-QDBC is made $-$ for selecting the quality IoT devices from the available ones, (b) call to \emph{split and equilibrium price determination} subroutine is made $-$ to distribute the agents into two different MCS arenas and determine the equilibrium price in both the arenas, (c) call to \emph{demand and supply calculation} is made $-$ to determine the demand and supply of the agents in the respective MCS arenas by using the equilibrium prices of the opposite MCS arena, and (d) call to \emph{winner determination and payment} is made $-$ to decide the winning agents and their payment. For each category $w_i$, in line 6, the winning IoT devices and winning task requesters are captured in $\mathcal{I}^{w(i)}$ and $r^{w(i)}$ respectively. In line 7 the equilibrium price  of right mobile crowdsourcing arena is stored in $p_f$ and after that $p_R$ is reset to $\phi$. Similarly, the equilibrium price of left crowdsourcing arena can also be calculated and stored.  In line 9, the list of winning task requesters $i.e.$ $r^w$ and list of winning IoT devices $i.e.$ $\mathcal{I}^w$ for all the categories are returned.

 \IncMargin{0.2em}
\begin{algorithm}[H]\label{algo:0}
\DontPrintSemicolon
\SetNoFillComment
    \SetKwInOut{Output}{Output}
\caption{k-category procedure ($\mathcal{I}$, $r$, $\boldsymbol{\succ}$, $w$)}

      \Output{ $\mathcal{I}^w$ $\leftarrow$ $\phi$, $r^w$ $\leftarrow$ $\phi$, $p_f \leftarrow \phi$}
	\ForEach{$w_i$ $\in$ $w$}
	{
 $\tilde{\mathcal{Q}}^i$ $\leftarrow$ IoT-QDBC ($\mathcal{I}^i$, $\boldsymbol{\succ}^i$) \tcp*{\textcolor{blue}{Call to IoT-QDBC subroutine is made and the quality IoT devices are held in $\tilde{\mathcal{Q}}^i$.}}
	($\mathcal{I}_L^i$, $r_L^i$, $p_R$)$\leftarrow$ Splitting and equilibrium price determination ($\tilde{\mathcal{Q}}^i$, $r$) \tcp*{\textcolor{blue}{Call to Splitting and equilibrium price determination subroutine is made.}}
	($\boldsymbol{d}^L$, $\boldsymbol{s}^L$) $\leftarrow$ Demand and supply calculation ($\mathcal{I}_L^i$, $r_L^i$, $p_R$) \tcp*{\textcolor{blue}{Call to Demand and supply calculation subroutine is made.}}
	($\mathcal{I}^{w(i)}$, $r^{w(i)}$) $\leftarrow$ Winner determination and payment ($\boldsymbol{d}^L$, $\boldsymbol{s}^L$, $\mathcal{I}_L^i$, $r_L^i$, $p_R$) \tcp*{\textcolor{blue}{Call to Winner determination and payment subroutine is made.}}
	$\mathcal{I}^{w}$ $\leftarrow$ $\mathcal{I}^{w}$ $\cup$ $\mathcal{I}^{w(i)}$; $r^{w}$ $\leftarrow$ $r^{w}$ $\cup$ $r^{w(i)}$\\
	$p_f \leftarrow p_f \cup p_R$; $p_R \leftarrow \phi$ 
        }
   \Return $\mathcal{I}^w$, $r^w$, $p_f$\\
\end{algorithm}
\IncMargin{0.2em}

\subsection{IoT-devices Quality Determination Mechanism}
\label{sub:QIoTBC}
In this, for hiring the quality IoT devices among the available IoT devices, a subroutine namely \emph{\textbf{IoT}-devices \textbf{q}uality \textbf{d}etermination using \textbf{b}orda \textbf{c}ount} (IoT-QDBC) motivated by \cite{NNis_Pre_2007,T.roughgarden_2016} is proposed. 

\subsubsection{\textbf{Outline of IoT-QDBC}}
In this section, the main idea of the IoT-QDBC is illustrated below.
 \begin{mdframed}[backgroundcolor=gray230]
\begin{center}\textbf{IoT-QDBC}\end{center}
Fix a category $w_i$:
\begin{enumerate}
\item Each time $\gamma$ IoT devices are picked up randomly that provide full rank list over the $\beta$ other IoT devices that are chosen randomly from the $(n-\gamma)$ IoT devices.
\item Now, based on the rank lists of the IoT devices, each IoT device $\mathcal{I}_j^i$ gets $\beta$ points for each first preference, $\beta-1$ points for each second preference, and so on, with 1 point for each last preference.
\item The process iterates until all the IoT devices are not ranked.
\end{enumerate}
\end{mdframed}

\subsubsection{\textbf{Detailing of IoT-QDBC}}
The detailing of IoT-QDBC is depicted in below listing. In line 1 of Alg. \ref{algo:1}, the lists $g$ and $s$ are initialized to $\phi$. In line 2, the copies of the set of IoT devices in $i^{th}$ category are maintained in $\boldsymbol{\mathcal{Z}}$ and $\boldsymbol{\mathcal{J}}$. The \emph{while} loop in line 3-16 determines the quality IoT devices. In line 4, $\gamma$ IoT devices are picked up randomly from the list of IoT devices $\boldsymbol{\mathcal{Z}}$ using $R\SmallUpperCase{SELECT}$ method and held in $g$. $s$ holds randomly selected $\beta$ IoT devices from the list of IoT devices remaining after picking $\gamma$ IoT devices in line 4. In line 6, each of the IoT devices in $g$ provides the full rank list over the IoT devices in $s$. Lines 7-13 determines the total point obtained by each of the IoT devices $\mathcal{I}_k^i \in s$ and held in $c$. In line 14, for each iteration of \emph{while} loop the IoT device having the maximum point is kept in $\mathcal{Q}^i$. The set of IoT devices that are ranked in the current iteration are removed from $\boldsymbol{\mathcal{J}}$ in line 15. The process iterates until $\boldsymbol{\mathcal{J}}$ becomes empty. Line 17 returns the set of quality IoT devices obtained in $w_i$ category.        

\IncMargin{0.2em}
\begin{algorithm}[!htbp]
\DontPrintSemicolon
\SetNoFillComment
    \SetKwInOut{Output}{Output}
\caption{IoT-QDBC ($\mathcal{I}^i$, $\boldsymbol{\succ}^i$)}
\label{algo:1}
      \Output{ $\mathcal{Q}^i$ $\leftarrow$ $\phi$}
$g \leftarrow \phi$, $s \leftarrow \phi$ \tcp*{\textcolor{blue}{$g$ and $s$ lists are initialized to $\phi$.}}
$\boldsymbol{\mathcal{Z}} = \boldsymbol{\mathcal{J}}= \mathcal{I}^i$ \tcp*{\textcolor{blue}{Set of IoT devices in $i^{th}$ category are held in $\boldsymbol{\mathcal{Z}}$ and $\boldsymbol{\mathcal{J}}$.}}
\While{$\boldsymbol{\mathcal{J}} \neq \phi$}
{
  $g$ $\leftarrow$ $R\SmallUpperCase{SELECT}(\boldsymbol{\mathcal{Z}},~\gamma)$ \tcp*{\textcolor{blue}{$\gamma$ IoT devices are picked up from $\boldsymbol{\mathcal{Z}}$ and held in $g$.}}
  $s$ $\leftarrow$ $R\SmallUpperCase{SELECT}(\boldsymbol{\mathcal{J}} \setminus g,~\beta)$ \tcp*{\textcolor{blue}{$\beta$ IoT devices are picked up from $\boldsymbol{\mathcal{J}} \setminus g$ and held in $s$.}}
Each $\mathcal{I}_k^i \in g$ provide a full rank list $i.e.$ $\boldsymbol{\succ}_k^i$ over all the IoT devices in $s$.\\
	\ForEach{$\mathcal{I}_k^i$ $\in$ $s$}
	{
$c_k \leftarrow 0$ \tcp*{\textcolor{blue}{$c_k$ holds the point of each IoT device and is initialized to 0.}}
         \ForEach{$\mathcal{I}_j^i$ $\in$ $g$}
{
$c_k \leftarrow c_k + (\beta - \ell)$ \tcp*{\textcolor{blue}{where $\ell = 0$ for first preference, $\ell = 1$ for second preference, likewise $\ell = (\beta-1)$ for last preference.}}
}
$c \leftarrow c \cup \{c_k\}$ \tcp*{\textcolor{blue}{$c$ holds the point received by each IoT device in $s$.}}
}
$\mathcal{Q}^i \leftarrow \mathcal{Q}^i \cup \argmax\limits_{\mathcal{I}_k^i \in s} \{c\}$ \tcp*{\textcolor{blue}{Selects the IoT device with maximum points and store it in $\mathcal{Q}^i$.}}
$\boldsymbol{\mathcal{J}} \leftarrow \boldsymbol{\mathcal{J}} \setminus s$ \tcp*{\textcolor{blue}{Removes the set of IoT devices from $\boldsymbol{\mathcal{J}}$ that are already ranked.}}
}
       
   \Return $\mathcal{Q}^i$
\end{algorithm}
\IncMargin{0.2em}

\begin{example}
\emph{
Let us understand IoT-QDBC algorithm with the help of an example for $w_2$ category. In our running example, there are 3 task requesters and 9 IoT devices. Here, $\beta$ and $\gamma$ values are taken as 3. Following line 4 of Alg. \ref{algo:1}, 3 IoT devices are picked up randomly for providing the ranking. Next, the remaining 3 IoT devices are chosen that are to be ranked. The full rank lists of the IoT devices for the first iteration are depicted in Figure \ref{fig:1a}. Now, following line 7-13 of Alg. \ref{algo:1} the points gained by each of the IoT devices will be calculated. As IoT device $\mathcal{I}_3^2$ is ranked first by two of the IoT devices ($\mathcal{I}_2^2$ and $\mathcal{I}_6^2$), and ranked second by one of the IoT device ($\mathcal{I}_4^2$), due to this reason it receives 3 points for each first rank and 2 points for second rank. So, IoT device $\mathcal{I}_3^2$ gains a total of 8 points as shown in Figure \ref{fig:2a1}. In the similar fashion, the IoT devices $\mathcal{I}_1^2$ and $\mathcal{I}_5^2$ gains a total of 7 points each depicted in Figure \ref{fig:2a1}. So, from the first iteration $\mathcal{I}_3^2$ is placed in the list of quality IoT devices.\\
\indent In the next iteration, again 3 IoT devices will be picked up randomly for providing the ranking over the other 3 IoT devices as shown in Figure \ref{fig:1b1}. As the IoT device $\mathcal{I}_2^2$ is ranked first by two IoT devices ($\mathcal{I}_1^2$ and $\mathcal{I}_8^2$), and ranked second by one IoT device ($\mathcal{I}_7^2$) due to this it receives 3 points for each first rank and 2 points for a second rank. So, IoT device $\mathcal{I}_2^2$ also gains a total point of 8 as shown in Figure \ref{fig:2b1}. In the similar fashion, the IoT devices $\mathcal{I}_4^2$ and $\mathcal{I}_6^2$ gains a total point of 6 and 4 respectively. So, in the second iteration $\mathcal{I}_2^2$ is placed in the list of quality IoT devices.}

\begin{figure}[H]
\begin{subfigure}[b]{0.33\textwidth}
                \centering
                \includegraphics[scale=0.90]{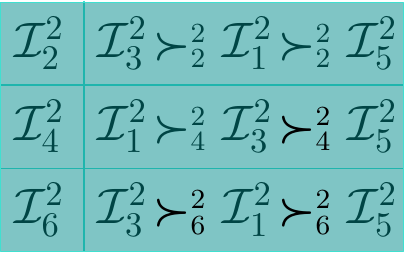}
                \subcaption{Rank list of $\mathcal{I}_2^2$, $\mathcal{I}_4^2$, $\mathcal{I}_6^2$}
                \label{fig:1a}
        \end{subfigure}%
        \begin{subfigure}[b]{0.33\textwidth}
                \centering
                \includegraphics[scale=0.90]{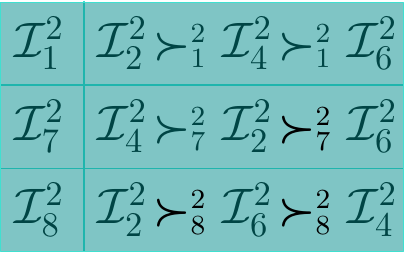}
                \subcaption{Rank list of $\mathcal{I}_1^2$, $\mathcal{I}_7^2$, $\mathcal{I}_8^2$}
                \label{fig:1b1}
        \end{subfigure}%
\begin{subfigure}[b]{0.33\textwidth}
                \centering
                \includegraphics[scale=0.90]{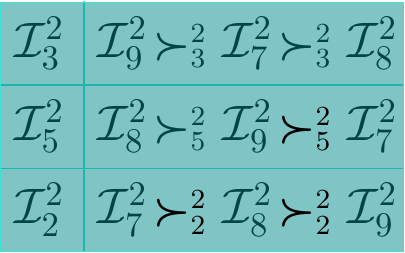}
        \subcaption{Rank list of $\mathcal{I}_3^2$, $\mathcal{I}_5^2$, $\mathcal{I}_2^2$}\label{fig:1c1}
        \end{subfigure}
        \par\bigskip
        \begin{subfigure}[b]{0.33\textwidth}
                \centering
                \includegraphics[scale=0.90]{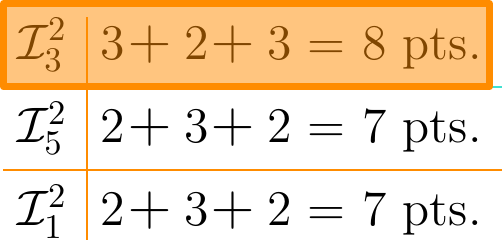}
                \subcaption{Point calculation of $\mathcal{I}_1^2$, $\mathcal{I}_3^2$, $\mathcal{I}_5^2$ }
                \label{fig:2a1}
        \end{subfigure}%
        \begin{subfigure}[b]{0.33\textwidth}
                \centering
                \includegraphics[scale=0.90]{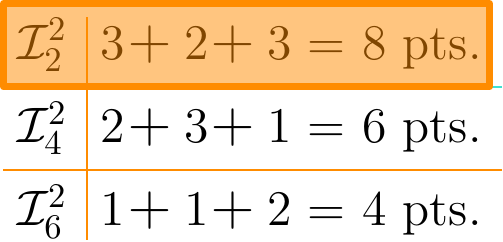}
                \subcaption{Point calculation of $\mathcal{I}_2^2$, $\mathcal{I}_4^2$, $\mathcal{I}_6^2$}
                \label{fig:2b1}
        \end{subfigure}%
\begin{subfigure}[b]{0.33\textwidth}
                \centering
                \includegraphics[scale=0.90]{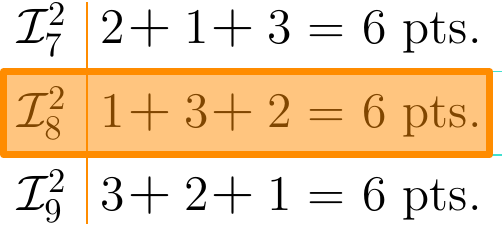}
        \subcaption{Point calculation $\mathcal{I}_7^2$, $\mathcal{I}_8^2$, $\mathcal{I}_9^2$}\label{fig:2c1}
        \end{subfigure}
        \caption{Detailed Illustration of IoT-QDBC mechanism.}
\end{figure}
\indent \emph{In the last iteration, again 3 IoT devices are picked up randomly and rankings are provided. In this iteration, IoT device $\mathcal{I}_8^2$ is ranked first by one IoT device ($\mathcal{I}_5^2$), ranked second by one IoT device ($\mathcal{I}_2^2$), and ranked third by the IoT device $\mathcal{I}_3^2$. Due to this it receives 3 points for a first rank, 2 points for a second rank, and 1 point for a third rank. So, IoT device $\mathcal{I}_8^2$ gains a total points of 6. In the similar fashion, the IoT devices $\mathcal{I}_7^2$ and $\mathcal{I}_9^2$ gains a total points of 6 each. So, in the this iteration, we can select any of the IoT devices as the quality IoT device, say, $\mathcal{I}_8^2$ is considered in the list of quality IoT devices.  Alg. \ref{algo:1} returns $\mathcal{I}_3^2$, $\mathcal{I}_2^2$, and $\mathcal{I}_8^2$ as the list of quality IoT devices for the running example.}
\end{example}

\subsection{Allocation and Pricing Rule}
Once the quality IoT devices are determined in the first phase of the proposed framework, now the goal is to tackle the issues raised in points 1 and 3 in Section \ref{s:intro}. For this purpose, a subroutine of QUAD namely \emph{allocation and pricing rule} motivated by \cite{SegalHalevi2018MUDAAT, T.roughgarden_20141} is proposed. The allocation and pricing rule consists of three subroutines: 
\begin{itemize}
    \item Splitting and equilibrium price determination, 
    \item Demand and supply calculation, and 
    \item Winner determination and payment
\end{itemize}
In the upcoming subsections, each of the subroutines is discussed in a detailed manner one-by-one.
\subsubsection{Splitting and Equilibrium Price Determination}
Using line 1 of Algorithm \ref{algo:truthful1} the set of task requesters and the IoT devices are divided into two arenas called \emph{\textbf{l}eft \textbf{m}obile \textbf{c}rowdsourcing \textbf{a}rena} (LMCA) and \emph{\textbf{r}ight \textbf{m}obile \textbf{c}rowdsourcing \textbf{a}rena} (RMCA). 
 
\begin{algorithm}[H]
\DontPrintSemicolon
\SetNoFillComment
    \SetKwInOut{Input}{Input}
    \SetKwInOut{Output}{Output}
    \tcc{\textcolor{blue}{Split mobile crowdsensing arena}}
    Divide the mobile crowdsensing arena into two sub-arenas namely \textbf{l}eft \textbf{m}obile \textbf{c}rowdsensing \textbf{a}rena (LMCA) and \textbf{r}ight \textbf{m}obile \textbf{c}rowdsensing \textbf{a}rena (RMCA).\\
    With probability $1/2$ place the task requesters and IoT devices into LMCA and RMCA independently. The task requesters and IoT devices in LMCA are held in $r_L^i$ and $\mathcal{I}_L^i$ respectively. For RMCA, $r_R^i$ and $\mathcal{I}_R^i$ used to hold task requesters and IoT devices respectively.  \\
    \tcc{\textcolor{blue}{Equilibrium price determination in RMCA}}
    $p \leftarrow 0$, $\boldsymbol{d}^R = \infty$, $\boldsymbol{s}^R = 0$ \tcp*{\textcolor{blue}{The variables are initialized to 0 and $\infty$.}}
    \While{$\boldsymbol{d}^R \neq \boldsymbol{s}^R$}
        {
              $p \leftarrow p + \epsilon$ \tcp*{\textcolor{blue}{Each time the price $p$ is incremented by $\epsilon$.}}
              \ForEach{$r_j^i \in r_R^i$}
              {
               $\boldsymbol{d}_j^i(p) = \argmax\limits_{f \in [0, \mathcal{Q}_j^i]} u_j^i(f,p)$ \tcp*{\textcolor{blue}{Determines the maximum demand of task requester $r_j^i$ at price $p$ and hold the demand in $\boldsymbol{d}_j^i(p)$.}} 
        }
        \ForEach{$\mathcal{I}_k^i \in \mathcal{I}_R^i$}
              {
$\boldsymbol{s}_k^i(p) = \argmax_{f \in [0, \mathcal{Q}_k^i]} z_k^i(f,p)$ \tcp*{\textcolor{blue}{Determines the maximum supply of IoT device $\mathcal{I}_k^i$ at price $p$ and hold the supply in $\boldsymbol{s}_k^i(p)$.}}
        }
        $\boldsymbol{d}^R = \sum\limits_{j=1}^{n_i} \boldsymbol{d}_j^i(p)$ \tcp*{\textcolor{blue}{Total demand in $i^{th}$ category at price $p$ is calculated for all the $n_i$ task requesters and is stored in $\boldsymbol{d}^R$.} }
        $\boldsymbol{s}^R = \sum\limits_{k=1}^{m_i} \boldsymbol{s}_k^i(p)$ \tcp*{\textcolor{blue}{Total supply  in $i^{th}$ category at price $p$ is calculated for all the $m_i$ IoT devices and is stored in $\boldsymbol{s}^R$.} }
        }   
        $p_R \leftarrow p$ \tcp*{\textcolor{blue}{The equilibrium price of RMCA is stored in $p_R$.}}   
        \Return $\mathcal{I}_L^i$, $r_L^i$, $p_R$  \tcp*{\textcolor{blue}{Returns the list of IoT devices, list of task requesters in LMCA and equilibrium price of RMCA.}}    
    \caption{Splitting and equilibrium price determination ($\mathcal{I}^i$, $r^i$)}
    \label{algo:truthful1}
\end{algorithm}
\noindent Each time a task requester or an IoT device is placed to LMCA/RMCA with probability $1/2$, independent of others, as depicted in line 2. Here, for determining the equilibrium price we are considering RMCA. However, in the similar fashion one can determine the equilibrium price of LMCA. In line 3, initially, the equilibrium price $p$ and supply ${\boldsymbol{s}}^R$ is initialized to 0, and the demand ${\boldsymbol{d}}^R$ at price $p = 0$ is set to $\infty$. The \emph{while} loop in line $4-14$ takes care about determining the demand and supply of each agent and terminates once the total demand becomes equal to total supply. For each iteration of \emph{while} loop the price $p$ is incremented by $\epsilon$. (some small constant value)Line 6-8 calculates the demand of each task requester present in RMCA at price $p$. $r_R^i$ captures the set of task requesters present in RMCA. Line 9-11 calculates the supply of each IoT device present in RMCA at price $p$. $\mathcal{I}_R^i$ captures the set of IoT devices present in RMCA. Line 12 and 13 holds the total demand and total supply of agents in $\boldsymbol{d}^R$ and $\boldsymbol{s}^R$ respectively. In line 15, the equilibrium price $p$ is held in $p_R$. Line 16 returns the list of IoT devices and list of task requesters in LMCA, and equilibrium price $p_R$ of RMCA.
\subsubsection{Demand and Supply Calculation}
In this section, the demand and supply of task requesters and IoT devices respectively are determined in LMCA by using the equilibrium price of RMCA (determined using Algorithm \ref{algo:truthful1}). In line $1-7$, for each task requester $r_j^i \in r_L^i$, the demand is calculated and the task requesters whose demands are positive at price $p_R$ is held in $\tilde{r}_L^i$. In line 5, the total demand in LMCA is calculated and stored in $\boldsymbol{d}^L$. In line $8-14$, for each IoT device $\mathcal{I}_k^i \in \mathcal{I}_L^i$, the supply is calculated and the IoT devices whose supply are positive at price $p_R$ is held in $\tilde{\mathcal{I}}_L^i$. In line 12, the total supply in LMCA is calculated and stored in $\boldsymbol{s}^L$. Line 15 returns total supply and total demand in LMCA at $p_R$.
 
  \begin{algorithm}[H]
\DontPrintSemicolon
\SetNoFillComment
    \SetKwInOut{Input}{Input}
    \SetKwInOut{Output}{Output}
    \Output{$\boldsymbol{d}^L \leftarrow 0$, $\boldsymbol{s}^L \leftarrow 0$}
  \tcc{\textcolor{blue}{Demand of task requesters in LMCA at price $p_R$}}
   \ForEach{$r_j^i \in r_L^i$}
              {
              $\boldsymbol{d}_j^i(p_R) = \argmax\limits_{f \in [0, \mathcal{Q}_j^i]} u_j^i(f,p_R)$ \tcp*{\textcolor{blue}{Calculating the demand of $r_j^i$ at equilibrium price $p_R$.}}            
            \If{$\boldsymbol{d}_j^i(p_R) > 0$}
            {
            $\tilde{r}_L^i \leftarrow \tilde{r}_L^i \cup \{r_j^i\}$ \tcp*{\textcolor{blue}{Each time $\tilde{r}_L^i$ holds the task requester $\tilde{r}_j^i$ if criteria in line 3 is satisfied.}}
            $\boldsymbol{d}^L = \boldsymbol{d}^L + \boldsymbol{d}_j^i(f,p_R)$ \tcp*{\textcolor{blue}{Total demand in LMCA is calculated and held in $\boldsymbol{d}^L$.}} 
           }
           }
           \tcc{\textcolor{blue}{Supply of IoT devices in LMCA at price $p_R$}}
              \ForEach{$\mathcal{I}_k^i \in \mathcal{I}_L^i$}
              {
              $\boldsymbol{s}_k^i(p_R) = \argmax\limits_{f \in [0, \mathcal{Q}_k^i]} z_k^i(f,p_R)$ \tcp*{\textcolor{blue}{Calculating the supply of $\mathcal{I}_j^i$ at equilibrium price $p_R$.}}            
            \If{$\boldsymbol{s}_k^i(p_R) > 0$}
            {
            $\tilde{\mathcal{I}}_L^i \leftarrow \tilde{\mathcal{I}}_L^i \cup \{\mathcal{I}_k^i\}$ \tcp*{\textcolor{blue}{Each time $\tilde{\mathcal{I}}_L^i$ holds the task requester $\tilde{\mathcal{I}}_k^i$ if criteria in line 10 is satisfied.}}
            $\boldsymbol{s}^L = \boldsymbol{s}^L + \boldsymbol{s}_j^i(f,p_R)$ \tcp*{\textcolor{blue}{Total supply in LMCA is calculated and held in $\boldsymbol{d}^L$.}} 
           }
           }
           \Return $\boldsymbol{d}^L,\boldsymbol{s}^L$ \tcp*{\textcolor{blue}{Returns, total demand and supply from LMCA.}}
    \caption{Demand and supply calculation ($\mathcal{I}_L^i$, $r_L^i$, $p_R$)}
\label{algo:truthful2_algo}
\end{algorithm}
 
\subsubsection{Winner Determination and Payment}
In this, the winners and their payment are determined in LMCA. Similarly, we can determine the set of winners and their payment in RMCA by doing notational modifications in Algorithm \ref{algo:truthful3}. On determining the demand and supply of the task requesters and task executors respectively in LMCA based on $p_R$, the three cases may arise: (1) $\boldsymbol{d}_L = \boldsymbol{s}_L$, (2) $\boldsymbol{d}_L > \boldsymbol{s}_L$, and (3) $\boldsymbol{d}_L < \boldsymbol{s}_L$. In line 1-12, the case with $\boldsymbol{d}_L = \boldsymbol{s}_L$ is tackled and the set of winning task requesters $r^{w(i)}$, winning IoT devices $I^{w(i)}$, and their payment are determined. 

\begin{algorithm}[!htbp]
\DontPrintSemicolon
\SetNoFillComment
    \SetKwInOut{Input}{Input}
    \SetKwInOut{Output}{Output}
\If{$\boldsymbol{d}^L == \boldsymbol{s}^L$}
   {
       \ForEach{$r_j^i$ $\in$ $r_L^i$}
       {
        \ForEach{$r_{j,l}^i \in r_j^i$}
        {
           $r^{w(i)} \leftarrow r^{w(i)} \cup \{r_{j,l}^i\}$ \tcp*{\textcolor{blue}{$r^{w(i)}$ holds the winning virtual task requesters (v.r's).}}
       } 
       }
        \ForEach{$\mathcal{I}_k^i$ $\in$ $\mathcal{I}_L^i$}
       {
        \ForEach{$\mathcal{I}_{k,v}^i \in \mathcal{I}_k^i$}
        {
           $\mathcal{I}^{w(i)} \leftarrow \mathcal{I}^{w(i)} \cup \{\mathcal{I}_{k,v}^i\}$ \tcp*{\textcolor{blue}{$\mathcal{I}^{w(i)}$ holds the winning virtual IoT devices (v.I's).}}
       } 
       }
   }    \ElseIf{$\boldsymbol{d}^L > \boldsymbol{s}^L$}
   {
   $\tilde{r}_L^i$ $\leftarrow$ Sort ($r_L^i$) \tcp*{\textcolor{blue}{Sort v.r's in decreasing order of marginal valuation.}}
   \While{$|r^{w(i)}| \neq \boldsymbol{s}^L$}
   {
       $r^*$ $\leftarrow$ Pick ($\tilde{r}_L^i$) \tcp*{\textcolor{blue}{Picks v.r's from sorted ordering of virtual task requesters.}}
       \ForEach{$r_{j,l}^{i}$ $\in$ $r^*$}
       {
           \While{$\boldsymbol{\nu}_{j,l}^i < p_R$ and $|r^{w(i)}| \neq \boldsymbol{s}^L$}
           {
           $r^{w(i)} \leftarrow r^{w(i)} \cup \{r_{j,l}^{i}\}$ \tcp*{\textcolor{blue}{$r^{w(i)}$ holds the winning v.r's in $i^{th}$ category.}}
           }
       }
   }
   \ForEach{$\mathcal{I}_{k,v}^i$ $\in$ $\mathcal{I}_L^i$}
       {
           $\mathcal{I}^{w(i)}$ $\leftarrow$  $\mathcal{I}^{w(i)}$ $\cup$ $\{\mathcal{I}_{k,v}^i\}$ \tcp*{\textcolor{blue}{$\mathcal{I}^{w(i)}$ holds the winning v.I's in $i^{th}$ category.}}
       }  
   }
    \Else
    {
 $\tilde{\mathcal{I}}_L^i$ $\leftarrow$ Sort ($\mathcal{I}_L^i$) \tcp*{\textcolor{blue}{Sort v.I's in increasing order of marginal valuation.}}
   \While{$|\mathcal{I}^{w(i)}| \neq \boldsymbol{d}^L$}
   {
       $\mathcal{I}^*$ $\leftarrow$ Pick ($\tilde{\mathcal{I}}_L^i$) \tcp*{\textcolor{blue}{Picks v.I's from sorted ordering of virtual IoT devices.}}
       \ForEach{$\mathcal{I}_{k,v}^{i}$ $\in$ $\mathcal{I}^*$}
       {
           \While{$\boldsymbol{\nu}_{k,v}^i > p_R$ and $|\mathcal{I}^{w(i)}| \neq \boldsymbol{d}^L$}
           {
           $\mathcal{I}^{w(i)} \leftarrow \mathcal{I}^{w(i)} \cup \{\mathcal{I}_{k,v}^{i}\}$ \tcp*{\textcolor{blue}{$\mathcal{I}^{w(i)}$ holds the winning v.I's in $i^{th}$ category.}}
           }
       }
   }
   \ForEach{$r_{j,l}^{i}$ $\in$ $r_L^i$}
       {
           $r^{w(i)}$ $\leftarrow$  $r^{w(i)}$ $\cup$ $\{r_{j,l}^i\}$ \tcp*{\textcolor{blue}{$r^{w(i)}$ holds the winning v.r's in $i^{th}$ category.}}
       }  
   }
    \Return $r^{w(i)}$, $I^{w(i)}$, $p_R$
    \caption{Winner determination and payment ($\boldsymbol{d}^L$, $\boldsymbol{s}^L$, $\mathcal{I}_L^i$, $r_L^i$, $p_R$)}
    \label{algo:truthful3}
\end{algorithm} 
Using line 3-5 the set of winning virtual task requesters are placed into $r^{w(i)}$. Line 7-11 determines the set of winning virtual IoT devices and are placed into $\mathcal{I}^{w(i)}$. In line 13-25, the case with $\boldsymbol{d}_L > \boldsymbol{s}_L$ is considered. In line 14 the virtual task requesters are sorted in decreasing order based on their valuation. The \emph{while} loop in line 15-22 terminates when the set of selected virtual task requesters becomes equal to the total supply in $w_i$ category. Using line 16, each time a task requester is picked up from the sorted ordering, sequentially. Now for each of the virtual task requesters for the picked up task requester the stopping condition in line 18 is checked. 
If it is satisfied then the virtual task requester is placed in the winning set otherwise not. The for loop in line 23-25 iterates through all the virtual IoT devices in LMCA and place it into $\mathcal{I}^{w(i)}$. Further, the case with  $\boldsymbol{d}_L < \boldsymbol{s}_L$ is considered in line 27-40. In line 28 the virtual IoT devices are sorted in increasing order based on their valuation. The \emph{while} loop in line 29-36 terminates when the set of selected virtual IoT devices becomes equal to the total demand in $w_i$ category. Using line 30, each time an IoT device is picked up from the sorted ordering, sequentially. Now for each of the virtual IoT devices for the picked up IoT device the stopping condition in line 32 is checked. If it is satisfied then the virtual IoT device is placed in the winning set otherwise not. The for loop in line 37-39 iterates through all the virtual task requesters in LMCA and place it into $r^{w(i)}$. Finally in line 41 the winning task requesters, winning IoT devices, and price of the agents in $i^{th}$ category are returned.
\begin{example}
\label{ex:2}
\emph{Let us understand the \emph{allocation and pricing rule} subroutine of QUAD with the help of an example. The category of the task requesters and IoT devices is considered as $w_2$. The $\epsilon$ value (increase in price $p$ in each iteration) is taken as 3. For the running example it is assumed that the agents are already distributed into two different mobile crowdsourcing arenas $i.e.$ LMCA and RMCA. Let us consider each of the MCS arenas one-by-one.}

\begin{itemize}
\item \emph{\textbf{LMCA:} We have 3 task requesters and 3 IoT devices as shown in Figure \ref{fig:1a}. For the given set-up first goal is to determine the equilibrium price $p$. So, let us say the initial price $p$ is set to 0. At this price, all the task requesters wants that their tasks get executed ($i.e.$ $\boldsymbol{d}^L$ is 6) and on the other hand no IoT device wants to execute the tasks ($i.e.$ $\boldsymbol{s}^L$ is 0).}
\begin{figure}[H]
\begin{subfigure}[b]{0.35\textwidth}
                \includegraphics[scale=0.68]{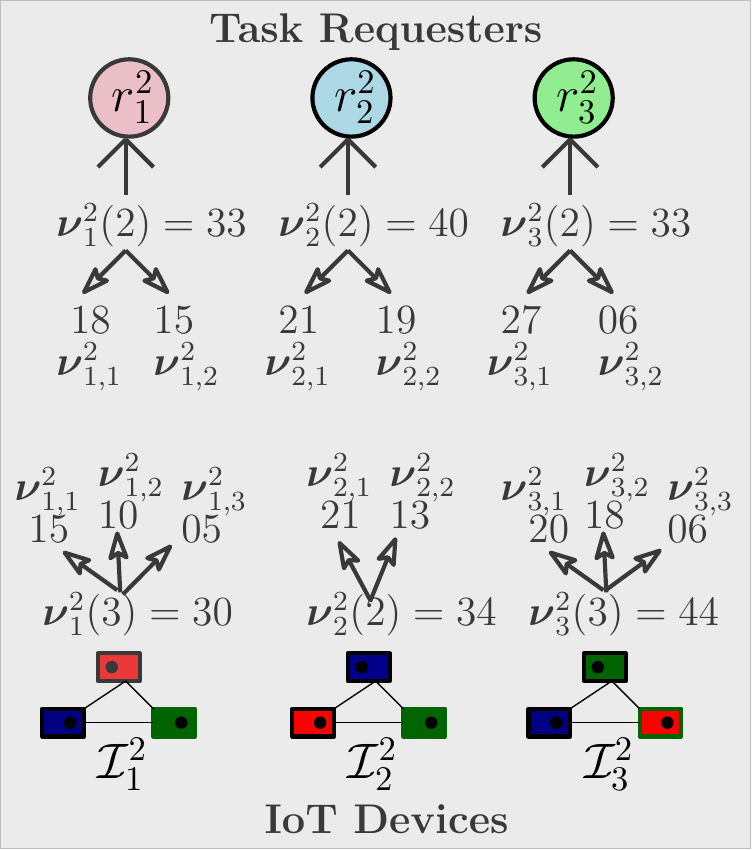}
                \subcaption{Initial set-up}
                \label{fig:1a}
        \end{subfigure}%
        \begin{subfigure}[b]{0.35\textwidth}
                \includegraphics[scale=0.68]{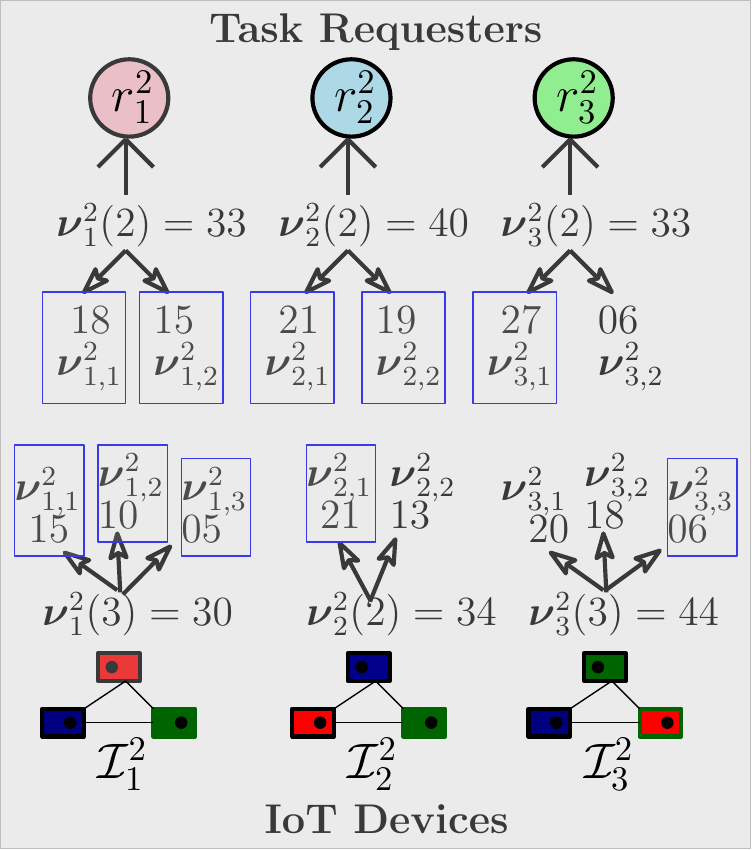}
                 \subcaption{Equilibrium Price Determination}
                \label{fig:1b}
        \end{subfigure}%
\begin{subfigure}[b]{0.37\textwidth}
                \includegraphics[scale=0.68]{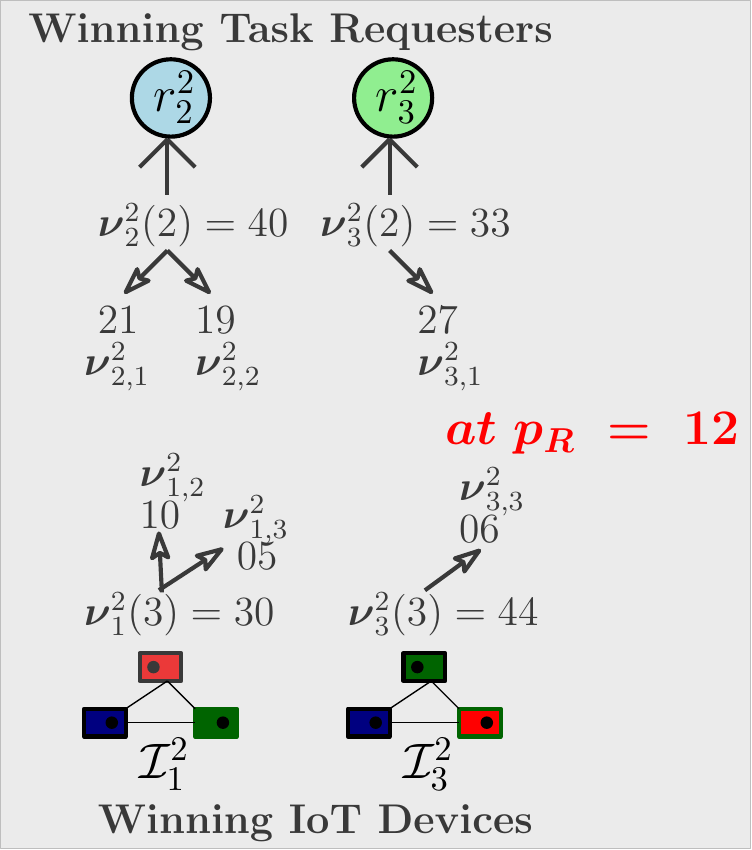}
        \subcaption{Winner and Payment Determination}
        \label{fig:1c}
        \end{subfigure}
        \caption{Detailed Illustration of Allocation and Pricing Rule in LMCA.}
\end{figure}

\emph{ Line 4 of Algorithm \ref{algo:truthful1} is satisfied and $p$ value is increased by 3. Now, at price $p = 3$ the demand from the task requesters $i.e.$ $\boldsymbol{d}^L$ is 6 and the supply of the IoT devices $i.e.$ $\boldsymbol{s}^L$ is 0. Again the condition in line 4 of the Algorithm \ref{algo:truthful1} is satisfied, so $p$ value is increased by 3 again and now $p$ is $6$. At $p = 6$, the $\boldsymbol{d}^L$ of task requesters is 6 and $\boldsymbol{s}^L$ of IoT devices is 1, so demand is not equal to supply still. Again, the price value is increased by 3 and now at $p = 9$, the $\boldsymbol{d}^L$ of task requesters is 5 and $\boldsymbol{s}^L$ of the IoT devices is 1. Similarly, we keep on increasing the price and at $p = 15$ the $\boldsymbol{s}^L$ becomes equal to $\boldsymbol{d}^L$. So the equilibrium price in the LMCA is 15. In Figure \ref{fig:1b}, at $p = 15$, the set of selected task requesters and IoT devices are shown in square box.}

\item \emph{\textbf{RMCA:} In RMCA, we have 2 task requesters and 3 IoT devices as shown in Figure \ref{fig:2a}. Let us first determine the equilibrium price $p$ in RMCA. So, let us say the initial price $p$ is set to 0. At this price, all the task requesters wants that their tasks get executed ($i.e.$ $\boldsymbol{d}^R$ is 4) and on the other hand no IoT device wants to execute the tasks ($i.e.$ $\boldsymbol{s}^R$ is 0). Line 4 of Algorithm \ref{algo:truthful1} is satisfied and $p$ value is increased by 3. Now, at price $p = 3$ the demand from the task requesters $i.e.$ $\boldsymbol{d}^R$ is 4 and the supply of the IoT devices $i.e.$ $\boldsymbol{s}^R$ is 0. Again the condition in line 4 of the Algorithm \ref{algo:truthful1} is satisfied, so $p$ value is increased by 3 again and now $p$ is $6$. At $p = 6$, the $\boldsymbol{d}^R$ of task requesters is 4 and $\boldsymbol{s}^R$ of IoT devices is 1, so demand is not equal to supply still. Again, the price value is increased by 3 and now at $p = 9$, the $\boldsymbol{d}^R$ of task requesters is 4 and $\boldsymbol{s}^R$ of the IoT devices is 2. Similarly, we keep on increasing the price and at $p = 12$ the $\boldsymbol{s}^R = 3$ becomes equal to $\boldsymbol{d}^R = 3$. So the equilibrium price in the RMCA is 12. In Figure \ref{fig:2b}, at $p = 12$, the set of selected task requesters and IoT devices are shown in square box. }

\begin{figure}[!htbp]
\begin{subfigure}[b]{0.35\textwidth}
                \includegraphics[scale=0.68]{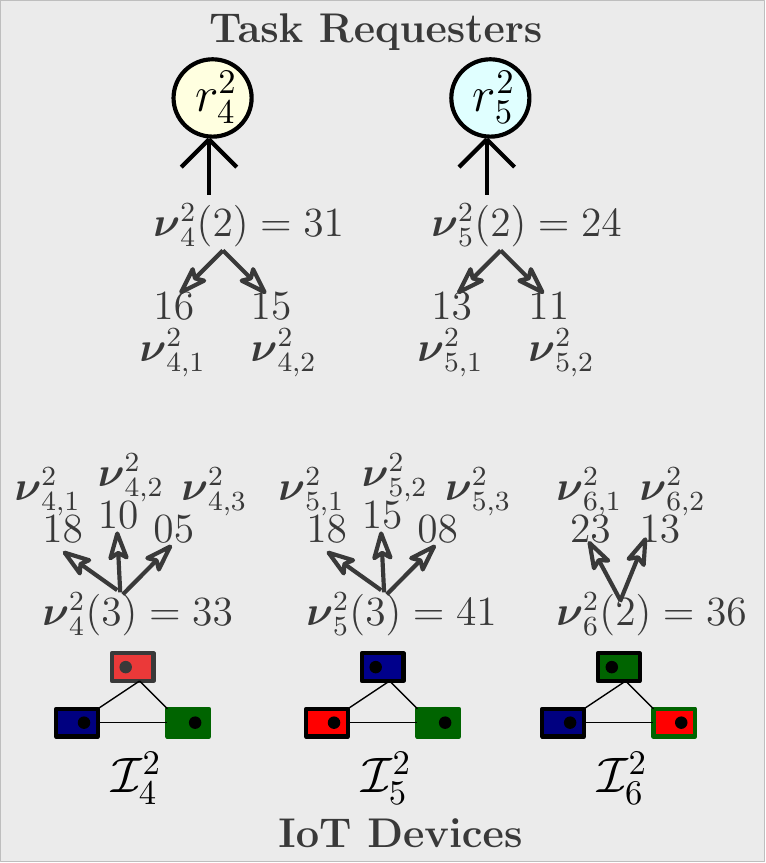}
                \subcaption{Initial set-up}
                \label{fig:2a}
        \end{subfigure}%
        \begin{subfigure}[b]{0.35\textwidth}
                \includegraphics[scale=0.68]{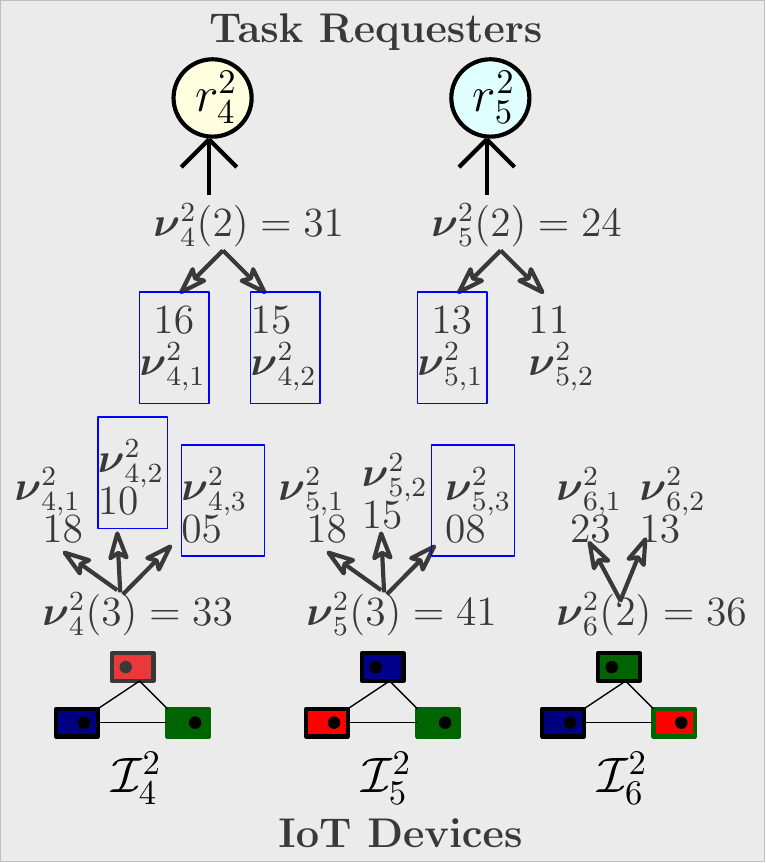}
                \subcaption{Equilibrium Price Determination}
                \label{fig:2b}
        \end{subfigure}%
\begin{subfigure}[b]{0.37\textwidth}
                \includegraphics[scale=0.68]{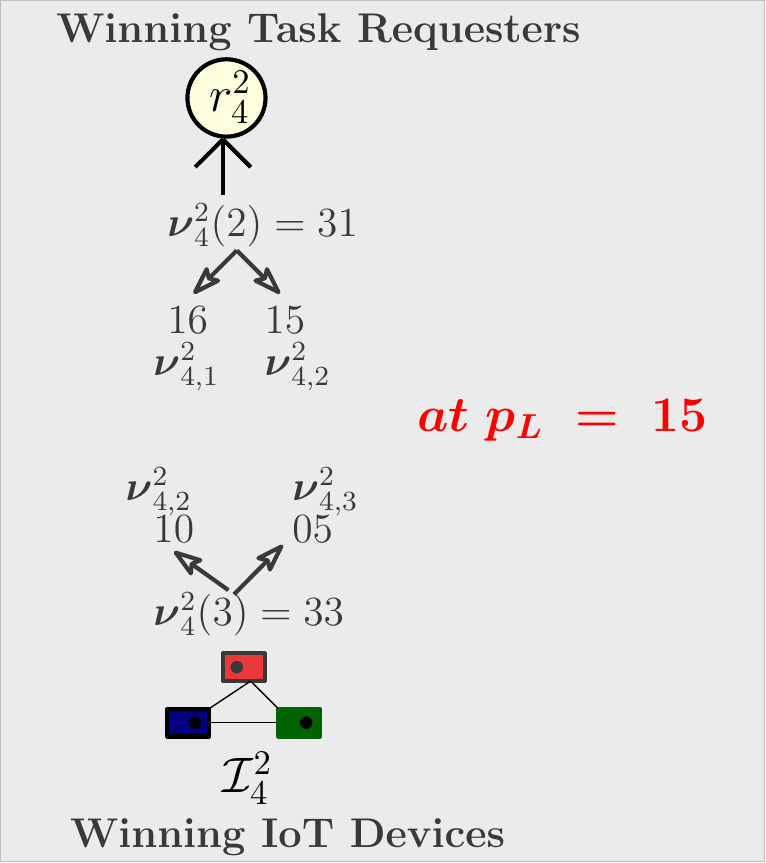}
        \subcaption{Winner and Payment Determination}\label{fig:2c}
        \end{subfigure}
        \caption{Detailed Illustration of Allocation and Pricing Rule in RMCA.}
\end{figure}

\end{itemize}
\emph{From the above discussion, we get the equilibrium price in LMCA as $p_L = 15$ and the equilibrium price in RMCA as $p_R = 12$. Once the equilibrium price in each of the markets is determined after that Algorithm \ref{algo:truthful2_algo} will be applied to our running example. Following Algorithm \ref{algo:truthful2_algo}, in RMCA, at $p_L = 15$, $\boldsymbol{d}^R$ is $3$ and $\boldsymbol{s}^R$ is $5$. On the other hand, in LMCA, at $p_R = 12$, $\boldsymbol{d}^L$ is $5$ and $\boldsymbol{s}^L$ is $3$. Now applying Algorithm \ref{algo:truthful3} to LMCA and RMCA one-by-one. For LMCA, line 13-26 of Algorithm \ref{algo:truthful3} will be activated, as $\boldsymbol{d}^L > \boldsymbol{s}^L$. Following line 14, the virtual task requesters are sorted in decreasing order of their valuation. After that using line 15-22 the winning virtual task requesters are decided and following line 23-25 winning virtual IoT devices are decided as shown in Figure \ref{fig:1b}. The price at which the trading took place is $p_R=12$. For RMCA, line 27-40 of Algorithm \ref{algo:truthful3} will be activated, as $\boldsymbol{d}^R < \boldsymbol{s}^R$. Following line 28, the virtual task requesters are sorted in increasing order of their valuation. After that using line 29-36 the winning virtual IoT devices are decided and following line 37-39 winning virtual task requesters are decided as shown in Figure \ref{fig:2c}. The price at which the trading took place is $p_L=15$.}                  
\end{example}

\section{Analysis of QUAD}\label{sec:atppm}
In this section, the analysis of QUAD is carried out. In Lemma \ref{lemma:1a} it is shown that the QUAD runs in polynomial time. In Lemma \ref{lemma:1b} it is proved that QUAD is correct. In order to prove that QUAD is correct, it is shown that each of the subroutines of QUAD is correct. In Lemma \ref{lemma:1c} it is proved that QUAD is prior-free. It means that QUAD is not using any statistical information on the valuations of the agents. It is shown in Lemma \ref{lemma:1d} that QUAD is truthful, it means that in QUAD, the agents cannot gain by misreporting their true valuation. In Lemma \ref{lemma:IR} it is shown that QUAD is IR. The reason behind proving that QUAD is IR to show that all the participating agents have non-negative utility. Lemma \ref{lemma:BB} proves that QUAD is weakly budget balanced. The reason behind proving QUAD as WBB is to show that the agents pay some amount to the platform. 
         
\begin{lemma}
\label{lemma:1a}
QUAD is computationally efficient. 
\end{lemma}
\begin{proof}
The running time of QUAD is the sum of the running time of Algorithm \ref{algo:0}, Algorithm \ref{algo:1}, Algorithm \ref{algo:truthful1}, Algorithm \ref{algo:truthful2_algo}, and Algorithm \ref{algo:truthful3}. As Algorithm \ref{algo:0} will iterate for $k$ times, so let us first determine the time taken by each of the subroutines for single iteration of Algorithm \ref{algo:0}. After that the running time for $k$ different categories will be decided.\\
\indent In Algorithm \ref{algo:1} lines 1 and 2 will take O(1) time. The \emph{while} loop in line 3-16 will terminate only when all the IoT devices are ranked. So, say the \emph{while} loop iterates for $\lfloor \frac{n_i}{\beta} \rfloor$. For each iteration of \emph{while} loop, line 4 and line 5 will take $O(\gamma)$ and $O(\beta)$ respectively. Line 6 is bounded above by $O(\gamma\beta)$. Line 7-13 will take $O(\gamma\beta)$. In line 14, the IoT device with maximum point is selected and will take $O(\beta)$ time. So, line 3-16 is bounded above by $\frac{n_i}{\beta} \cdot O(\gamma\beta) = n_i\gamma$. So, the time taken by Algorithm \ref{algo:1} is $O(n_i\gamma)$. If $\gamma$ is a function of $n_i$, then it will be $O(n_i^2)$.\\
\indent In Algorithm \ref{algo:truthful1}, line 1 and 2 will take $O(m_i+n_i)$ time. Line 3 will take constant time. The execution of Lines 5-13 depends on the condition in line 4. For each iteration of \emph{while} loop, line 5 will take $O(1)$. Line $6-8$ is bounded above by $O(m_i\mathcal{Q}_j^i)$ and line $9-11$ is bounded above by $O(n_i\mathcal{Q}_k^i)$. In the worst case, line 12 will sum the demands of all the $m_i$ task requesters at equilibrium price $p$ and will take $O(m_i)$ time. Similarly, line 13 will take $O(n_i)$ time. So, for each iteration of \emph{while} loop, line $5-13$ is bounded above by $O(1)$ + $O(n_i\mathcal{Q}_j^i)$ + $O(m_i\mathcal{Q}_k^i)$ + $O(m_i)$ + $O(n_i)$ = $O(n_i^2)$, if $\mathcal{Q}_j^i$, $\mathcal{Q}_k^i$, and $m_i$ is a function of $n_i$. If the number of iteration of \emph{while} loop in line 4-14 is a function of $n_i$, then line 4-14 is bounded above by $O(n_i^3)$. Line 15 does assignment so it will take constant time. The return statement in line 16 will take $O(1)$. So, Algorithm \ref{algo:truthful1} takes $O(n_i^3)$.\\  
\indent In Algorithm \ref{algo:truthful2_algo} for each iteration of \emph{for} loop in line 1-7, line 2 will take $O(\mathcal{Q}_j^i)$. Line 3-6 will take $O(1)$. Now, if the \emph{for} loop iterates for the number of times that is a function of $m_i$, then line 1-7 of Algorithm \ref{algo:truthful2_algo}  is bounded above by $O(m_i^2)$, if $\mathcal{Q}_j^i$ is a function of $m_i$. Similar argument can be given for line 8-14 and is bounded above by $O(n_i^2)$. The return statement in line 15 will take constant time. So, Algorithm \ref{algo:truthful2_algo} takes $O(m_i^2)$ + $O(n_i^2)$ time.\\
\indent Now, let us determine the time taken by Algorithm \ref{algo:truthful3} for each iteration of Algorithm \ref{algo:0}. From the construction of Algorithm \ref{algo:truthful3} it can be seen that at a time either line $1-11$ will be executed or line $12-24$ will be executed or line $26-38$ will be executed. Let us say line $1-12$ gets executed. Line 1 will take constant time. Line $2-6$ takes $O(m_i^2)$ in worst case. Similarly, line $7-11$ can take $O(n_i^2)$ in worst case. So line $1-12$ of Algorithm \ref{algo:truthful3} is bounded above by $O(n_i^2)$ + $O(m_i^2)$. Line $12-24$ executes when the demand of the task requesters is more than the supply from the IoT devices. In line 13 sorting takes place and will take $O(m_i\lg m_i)$, where $\lg m_i$ stands for $\log_{10} m_i$. Line $14-21$ will take $O(m_i^3)$ in worst case. Line 22-24 will take $O(n_i)$. Similarly,  line $25-38$ executes when the demand of the task requesters is less than the supply from the IoT devices. In line 26 sorting takes place and will take $O(n_i\lg n_i)$. Line $27-34$ will take $O(n_i^3)$ in worst case. Line 35-37 will take $O(m_i)$. So the running time of Algorithm \ref{algo:truthful3} is $O(m_i^2) + O(m_i \lg m_i) + O(m_i^3) + O(n_i \lg n_i) + O(n_i^3) = O(m_i^3) + O(n_i^3)$.\\
\indent So, the running time of QUAD for any category $w_i$ is $ O(n_i^2) + O(n_i^3) + O(n_i^2) + O(m_i^2) + O(n_i^3) + O(m_i^3)$ = $O(n_i^3) + O(m_i^3)$. So, for $k$ different categories the QUAD will be bounded above by $kn_i^3 + km_i^3 = k(n_i^3 + m_i^3) $ $i.e.$ $O(n_i^3 + m_i^3)$, if $k$ is constant. Hence, QUAD is computationally efficient.                                 
\end{proof}

\begin{lemma}
\label{lemma:1b}
QUAD works correctly. 
\end{lemma}
\begin{proof}
We prove the correctness of QUAD using loop invariant technique \cite{Coreman_2009}. In order to prove that QUAD is correct, it is to be shown that each of the algorithms acting as subroutine are also correct. Below the correctness of each of the subroutines is discussed one-by-one:\\

\noindent \textbf{Proof of correctness of Algorithm \ref{algo:1} (IoT-QDBC):} In order to prove that IoT-QDBC is correct, the following loop invariant is considered:
\begin{mdframed}[backgroundcolor=gray230]
\textbf{Loop invariant:} In each iteration of \emph{while} loop of lines 3-16, a quality IoT device in $i^{th}$ category is added into the output array $\mathcal{Q}^i$.
\end{mdframed}
\begin{itemize}
\item \textbf{Initialization:} We can start by showing that loop invariant hold before the first iteration of \emph{while} loop, when $\mathcal{Q}^i = \phi$. The output array $\mathcal{Q}^i$ have no quality IoT devices before the first iteration. So, loop invariant holds.  
\item \textbf{Maintenance:} For the loop invariant to be true, it is to be shown that before any $k^{th}$ iteration of the \emph{while} loop and after $k^{th}$ iteration of the \emph{while} loop the loop invariant holds. Before $k^{th}$ iteration, $i.e.$ till $(k-1)^{th}$ iteration there will be $(k-1)$ quality IoT devices in an array $\mathcal{Q}^i$. After $k^{th}$ iteration, the number of quality IoT devices will be $\sum\limits_{i=1}^{k} 1 = k$ in $\mathcal{Q}^i$. So, loop invariant holds.
\item \textbf{Termination:} From the construction of IoT-QDBC, it is clear that the \emph{while} loop will terminate only when there exists no IoT devices whose executed tasks are to be ranked. It means that once \emph{while} loop terminates $\mathcal{Q}^i$ contains all the quality IoT devices.
\end{itemize}
Hence, IoT-QDBC is correct.\\

\noindent \textbf{Proof of correctness of Algorithm \ref{algo:truthful1} (Splitting and Equilibrium Price Determination):} In order to prove that Algorithm \ref{algo:truthful1} is correct, we use the following loop invariant:
\begin{mdframed}[backgroundcolor=gray230]
\textbf{Loop invariant:} At the start of \emph{while} loop in lines 4-14, the price $p$ is increased by $\epsilon$. So, with the increase in price $p$ the demand $\boldsymbol{d}^R$ of the task requesters will be decreasing or remains same and supply $\boldsymbol{s}^R$ of the IoT devices will be increasing or remains same.  
\end{mdframed}   
\begin{itemize}
\item \textbf{Initialization:} Before the $1^{st}$ iteration of \emph{while} loop, the price $p$ is \emph{zero}. In this case, the demand $\boldsymbol{d}^R$ of the task requesters will be very high (say $\infty$) and supply $\boldsymbol{s}^R$ of the IoT devices will be very low (say 0). So, it trivially satisfies the loop invariant.  
\item \textbf{Maintenance:} In the $k^{th}$ iteration of \emph{while} loop, the price $p$ will be some non-zero value. Depending on the price value at $k^{th}$ iteration, the demand of task requester may remain same or may decrease and the supply of IoT devices may remain same or may increase. After $k^{th}$ iteration the price $p$ will be $\epsilon$ greater than it was before $k^{th}$ iteration. Now, with the increase in price the demand may decrease or remains same and supply may increase or remains same. So, the loop invariant may hold.
\item \textbf{Termination:} From the construction of Algorithm \ref{algo:truthful1}, at termination, the demand of task requesters and supply by the IoT devices becomes equal at the equilibrium price. So, the algorithm provide the equilibrium price on termination.
\end{itemize}
Hence, Algorithm \ref{algo:truthful1} is correct.\\

\noindent \textbf{Proof of correctness of Algorithm \ref{algo:truthful2_algo} (Demand and Supply Calculation):} 

\begin{mdframed}[backgroundcolor=gray230]
\textbf{Loop invariant:} In each iteration of for loops in line 1-7 and in line 8-14 the $\boldsymbol{d}^L$ and $\boldsymbol{s}^L$ respectively will be either zero or positive.  
\end{mdframed}  
\begin{itemize}
\item \textbf{Initialization:} Before the $1^{st}$ iteration of \emph{for} loops in line 1-7 and line 8-14, the $\boldsymbol{d}^L$ and $\boldsymbol{s}^L$ values respectively are 0. So, it trivially satisfies the loop invariant.  
\item \textbf{Maintenance:} Before any $k^{th}$ iteration of \emph{for} loops, in line 1-7 and line 8-14, the $\boldsymbol{d}^L$ and $\boldsymbol{s}^L$ values respectively will be either zero or some positive value. After the $k^{th}$ iteration of \emph{for} loops, in line 1-7 and line 8-14 also, the $\boldsymbol{d}^L$ and $\boldsymbol{s}^L$ values respectively will be either zero or some positive value. So, the loop invariant may hold.

\item \textbf{Termination:} From the construction of Algorithm \ref{algo:truthful1}, at termination the demand of task requesters and supply by the IoT devices will be non-negative. So, on termination also the loop invariant holds for Algorithm \ref{algo:truthful2_algo}.
\end{itemize}
Hence, Algorithm \ref{algo:truthful2_algo} is correct.\\

\noindent \textbf{Proof of correctness of Algorithm \ref{algo:truthful3} (Winner Determination and Payment):} The case with $\boldsymbol{d}^L = \boldsymbol{s}^L$ is trivial. In order to prove correctness of Algorithm \ref{algo:truthful3}, the non-trivial cases are considered $i.e.$ $\boldsymbol{d}^L > \boldsymbol{s}^L$ or $\boldsymbol{d}^L < \boldsymbol{s}^L$. Let us $\boldsymbol{d}^L > \boldsymbol{s}^L$ case, similarly $\boldsymbol{d}^L < \boldsymbol{s}^L$ can be handled.  
\begin{mdframed}[backgroundcolor=gray230]
\textbf{Loop invariant:} In each iteration of \emph{while} loop in line 14-21 the task requesters are added in the winning set $r^{w(i)}$.  
\end{mdframed}   

\begin{itemize}
\item \textbf{Initialization:} We can start by showing that loop invariant holds prior to first iteration of \emph{while} loop in line 14-21, when $r^{w(i)} = \phi$. The winners list $r^{w(i)}$ have no task requesters.  So, loop invariant holds.  
\item \textbf{Maintenance:} For the loop invariant to be true, it is to be shown that before any $k^{th}$ iteration of the \emph{while} loop in line 14-21 and after any $k^{th}$ iteration of the \emph{while} loop in line 14-21 the loop invariant holds. Before $k^{th}$ iteration, $i.e.$ till $(k-1)^{th}$ iteration there will be $(k-1)$ winning task requesters in $r^{w(i)}$. After $k^{th}$ iteration, one more task requester will be added to $r^{w(i)}$ leads to $k$ task requesters. So, loop invariant holds.
\item \textbf{Termination:} From the construction of Algorithm \ref{algo:truthful3}, it is clear that the \emph{while} loop in line 14-21 will terminate only when $r^{w(i)}$ becomes equal to $\boldsymbol{s}^L$. It means that once \emph{while} loop terminates $r^{w(i)}$ contains all the possible winning task requesters at equilibrium price. On termination the loop invariant holds for Algorithm \ref{algo:truthful3}.
\end{itemize}
Hence, Algorithm \ref{algo:truthful3} is correct.\\

\noindent As Algorithm \ref{algo:1}, Algorithm \ref{algo:truthful1}, Algorithm \ref{algo:truthful2_algo}, and Algorithm \ref{algo:truthful3} are correct for any category $w_i$ and so as the QUAD. So, QUAD will be correct for all the $k$ categories in the system. Hence QUAD is correct.   
\end{proof}

\begin{lemma}
\label{lemma:1c}
QUAD is prior-free. 
\end{lemma}
\begin{proof} 
From the construction of QUAD it can be seen that it is not using any statistical information on valuations of the agents. So, QUAD is prior-free. 
\end{proof}

\begin{observation}
\emph{QUAD mechanism consists of two components: IoT-QDBC mechanism, and allocation and pricing rule. Here, IoT-QDBC mechanism is independent of the bid values of the IoT devices and will not influence the truthfulness of QUAD.}
\end{observation}
\begin{lemma}
\label{lemma:1d}
QUAD is truthful or DSIC. 
\end{lemma}
\begin{proof}
Fix category $w_i$ and task requester $r_j^i$. To prove that QUAD is truthful, the two cases are considered: (1) \textbf{early exit}, and (2) \textbf{late exit}. In the first case, it may happen that task requester $r_j^i$ bids less than his true valuation and can come out of the market early, even if he can participate in further iterations. In late exit, it may happen that the task requester $r_j^i$ can bid more than his true valuation and can stay for longer time in the market. The proof is carried out considering the task requesters. However, in the similar fashion the truthfulness of the QUAD can be proved for the IoT devices. Let us discuss the two cases one-by-one. 
\begin{itemize}
\item [-] \textbf{Early exit:} Let us suppose that a task requester $r_j^i$ misreports his bid value for buying $t$ units of completed tasks such that $\hat{\boldsymbol{\nu}}_j^i(t) < \boldsymbol{\nu}_j^i(t)$. Now, in this situation the two scenarios can happen. One scenario could be that the $i^{th}$ iteration is the last iteration of the mechanism $i.e.$ demand and supply becomes equal and the mechanism terminates. In this case, the task requester $r_j^i$ wins and his utility will be same as his utility when he would have reported truthfully $i.e.$ $\hat{u}_j^i(t,p) = \boldsymbol{\nu}_j^i(t) - p \cdot t = u_j^i(t,p)$. Another scenario could be the mechanism can go further for $(i+1)^{th}$ iteration and many more. Now, in this scenario as the task requester has misreported (lowered his bid value) and exited from the market, so he is a loser and his utility will be 0. 
\item [-] \textbf{Late exit:} Let us suppose that a task requester $r_j^i$ misreports his bid value for buying $t$ units of completed tasks such that $\hat{\boldsymbol{\nu}}_j^i(t) > \boldsymbol{\nu}_j^i(t)$. Now, in this situation the two scenarios can happen. One scenario could be that the $i^{th}$ iteration is the last iteration of the mechanism and after that the mechanism terminates. In this case, the task requester $r_j^i$ wins and his utility will be same as his utility when he would have reported truthfully $i.e.$ $\hat{u}_j^i(t,p) = \boldsymbol{\nu}_j^i(t) - p \cdot t = u_j^i(t,p)$. Another scenario could be the mechanism can go further for $(i+1)^{th}$ iteration and many more. Now, in this scenario again the two cases can happen. The task requester $r_j^i$ can participate in further iterations by misreporting his bid value $i.e.$ $\hat{\boldsymbol{\nu}}_j^i(t) > \boldsymbol{\nu}_j^i(t)$ and he wins. In this case, his utility will be $\hat{u}_j^i(t,p) = \boldsymbol{\nu}_j^i(t) - p \cdot t < u_j^i(t,p)$. Another case could be with the increase in his bid value the task requester $r_j^i$ made his presence for some additional round but dropped from the market in the later rounds. In this case he loses and his utility will be 0.      

\end{itemize}
So, from above discussion it can be concluded that the task requester cannot gain by misreporting his bid value. In the similar line it can be shown that the IoT devices cannot gain by misreporting its bid value. Hence, QUAD is truthful for $w_i$ category and so is for all the $k$ categories.  
\end{proof}
\begin{lemma}
\label{lemma:IR}
QUAD satisfies individual rationality. 
\end{lemma}
\begin{proof}
 Fix a category $w_i$. In QUAD each time a price $p$ is set and then it is asked from the agents that \emph{whether you are ready to trade at price $p$ or not?} Now, if for any task requester $r_j^i$ the bid value is more than or equal to the current price then his answer will be \textbf{YES}, otherwise his answer will be \textbf{NO}. It is obvious that the task requester $r_j^i$ will reply \textbf{YES} only when $\boldsymbol{\nu}_j^i(t) > p$. In such case $r_j^i$ has to pay less than his valuation and achieves positive utility. Else if $\boldsymbol{\nu}_j^i(t) < p$ then $r_j^i$ will drop from the crowdsensing market and in that case his utility will be \emph{zero}.  So, the task requester will be having a non negative utility. Similar argument can be given for the IoT devices. As the utility of the participating agents are non negative, so from the definition of individual rationality (see Definition \ref{def:IR}) it can be said that QUAD is individually rational for $w_i$ category. Hence, QUAD is individually rational for all the $k$ categories.  
\end{proof}
\begin{lemma}
\label{lemma:BB}
The allocation resulted by QUAD is weakly budget balanced. 
\end{lemma}
\begin{proof}
Fix a category $w_i$. Let us say some virtual task requesters of $r_{j}^i$ is present in the winning set. Now, from the construction of QUAD the trading fee paid by any task requester $r_j^i$ is the utility that would have been achieved by the virtual task requesters that are not in the winning set because of the presence of virtual task requesters of task requester $r_j^i$ in the winning set. In Lemma \ref{lemma:IR} it is already shown that the agents in the winning set is having a non-zero utility, so it can be said that the trading fees paid by any task requester $r_j^i$ will be a non-zero value. So, the platform will always receive non zero fees from the task requesters. Similar argument can be given for the IoT devices. QUAD is weakly budget balanced for $w_i$ category. Hence, QUAD is weakly budget balance for all the $k$ different categories.         
\end{proof}

\subsection{Probabilistic Analysis of QUAD}
In this section, the probabilistic analysis of QUAD is carried out. Lemma \ref{lemma:1l} provide an estimate that, on an average how many tasks are executed for any task requester? It is shown that for any task requester $r_j^i$ the expected number of his endowed tasks will be executed by the quality IoT devices will be $\boldsymbol{\Lambda}_i \cdot \bigg(\frac{1}{\lg \boldsymbol{\Lambda}_i}\bigg)$. Here, $\boldsymbol{\Lambda}_i$ is the total number of tasks carried by any task requester $r_j^i$. Further in Lemma \ref{lemma:1101}, it is estimated that at least $0.9\lg \boldsymbol{\Lambda}_i$ tasks of any task requester $r_j^i$ will be considered out of $\boldsymbol{\Lambda}_i$ tasks is at most $\frac{10}{9\lg \boldsymbol{\Lambda}_i}$.    
 
\begin{lemma}
\label{lemma:1l}
In category $w_i$, for any task requester $r_j^i$ the expected number of his endowed tasks executed by the quality IoT devices out of $\boldsymbol{\Lambda}_i$ is given as:
\begin{equation*}
E[\boldsymbol{\mathcal{Z}}] =  \boldsymbol{\Lambda}_i \cdot \bigg(\frac{1}{\lg \boldsymbol{\Lambda}_i}\bigg)
\end{equation*}   
where, $\boldsymbol{\Lambda}_i$ is the number of tasks held by task requester $r_j^i$. $\boldsymbol{\mathcal{Z}}$ is the indicator random variable measuring the number of tasks executed by the quality IoT devices out of $\boldsymbol{\Lambda}_i$ tasks.  
\end{lemma}

\begin{proof}
Fix category $w_i$ and the task requester $r_j^i$. In this lemma, we are trying to show that, in expectation, how many tasks of the task requester $r_j^i$ gets executed by the quality IoT devices out of $\boldsymbol{\Lambda}_i$? It is represented as $E[\boldsymbol{\mathcal{Z}}]$. Before moving forward, it is important to decide the sample space associated with the task requester $r_j^i$ and is given as:
\begin{center}
S = $\{i^{th}$ task of $r_j^i$ gets executed, $i^{th}$ task of $r_j^i$ does not get executed$\}$
\end{center} 
Now, let us determine the probability that $i^{th}$ task of $r_j^i$ gets executed and is given as $\boldsymbol{q}$. So, the probability that $i^{th}$ task of $r_j^i$ do not get executed is given as $1 - \boldsymbol{q}$. Let us define an indicator random variable $\boldsymbol{\mathcal{Z}}_i$, associated with $i^{th}$ task of $r_j^i$ gets executed, which is the event $R$. The variable $\boldsymbol{\mathcal{Z}}_i$ captures that the $i^{th}$ task of $r_j^i$ gets executed or not. If it gets executed then $\boldsymbol{\mathcal{Z}}_i$ is 1 otherwise 0. We write

\begin{equation}
\label{equ:123}
\boldsymbol{\mathcal{Z}}_i = I\{R\} 
\end{equation}
\begin{equation*}
\hspace*{32mm} = 
\begin{cases}
1, & {if~R~happens},\\
0, & {otherwise} 
\end{cases} 
\end{equation*}
The expected value of the random variable $\boldsymbol{\mathcal{Z}}_i$ gives us the expected value that the $i^{th}$ task of $r_j^i$ gets executed. Taking expectation both side of equation \ref{equ:123}, we get
\begin{equation*}
E\[\boldsymbol{\mathcal{Z}}_i\] = E\[I\{R\}\] 
\end{equation*}
\begin{equation*}
\hspace*{26mm} =  1 \cdot \boldsymbol{q} + 0 \cdot (1-\boldsymbol{q})   
\end{equation*}
\begin{equation*}
\hspace*{0mm} =  \boldsymbol{q}   
\end{equation*}
Now, let $\boldsymbol{\mathcal{Z}}$ be the random variable denoting the total number of tasks executed by the quality IoT devices out of $\boldsymbol{\Lambda}_i$ tasks. It is formulated as
\begin{equation}
\label{equ:1234}
\boldsymbol{\mathcal{Z}} = \sum_{i=1}^{\boldsymbol{\Lambda}_i} \boldsymbol{\mathcal{Z}}_i
\end{equation}
 As our aim is to compute the expected number of tasks executed by the quality IoT devices out of $\boldsymbol{\Lambda}_i$ tasks and can be obtained by taking expectation both side of the equation \ref{equ:1234}. So, we get
 \begin{equation}
 \label{equ:12345}
E[\boldsymbol{\mathcal{Z}}] = E\bigg[\sum_{i=1}^{\boldsymbol{\Lambda}_i} \boldsymbol{\mathcal{Z}}_i\bigg]
\end{equation}
By linearity of expectation and then substituting the value of $E[\boldsymbol{\mathcal{Z}}_i]$, we get
 \begin{equation}
 \label{equ:123456}
E[\boldsymbol{\mathcal{Z}}] = \sum_{i=1}^{\boldsymbol{\Lambda}_i} E[\boldsymbol{\mathcal{Z}}_i] = \sum_{i=1}^{\boldsymbol{\Lambda}_i} \boldsymbol{q}
\end{equation}     
 \begin{equation}
 \label{equ:1234567}
\hspace*{-11mm} = \boldsymbol{\Lambda}_i \cdot \boldsymbol{q}
\end{equation}   
Now, if the probability that $i^{th}$ task of $r_j^i$ task requester will get executed is taken as $\frac{1}{\lg \boldsymbol{\Lambda}_i}$ then the value $E[\boldsymbol{\mathcal{Z}}]$ boils down to $\boldsymbol{\Lambda}_i \cdot \bigg(\frac{1}{\lg \boldsymbol{\Lambda}_i}\bigg)$. It can be written as:
\begin{equation*}
 E[\boldsymbol{\mathcal{Z}}] = \boldsymbol{\Lambda}_i \cdot \bigg(\frac{1}{\lg \boldsymbol{\Lambda}_i}\bigg)
\end{equation*}
Hence proved.
\end{proof}

\begin{observation}
If we consider the value of $\boldsymbol{\Lambda}_i$ as $100$ then $E[\boldsymbol{\mathcal{Z}}] = \frac{\boldsymbol{\Lambda}_i}{\lg \boldsymbol{\Lambda}_i} = \frac{100}{\lg 100} = \frac{100}{2} = 50$. It means that, for task requester $r_j^i$ on an average 50 of his tasks will be executed by the IoT devices out of his 100 sensing tasks.  
\end{observation}

\begin{lemma}
\label{lemma:1101}
The probability that at least $0.9 \boldsymbol{\Lambda}_i$ tasks of any task requester $r_j^i$ will be executed out of $\boldsymbol{\Lambda}_i$ by the quality IoT devices is less than or equals to $\bigg(\frac{10}{9\lg \boldsymbol{\Lambda}_i}\bigg)$. More formally, we can write

\begin{equation*}
P\bigg\{\boldsymbol{\mathcal{Z}} \geq 0.9 \boldsymbol{\Lambda}_i\bigg\} \leq \bigg(\frac{10}{9\lg \boldsymbol{\Lambda}_i}\bigg)
\end{equation*}

\begin{proof}
Fix a category $w_i$ and a task requester $r_j^i$. Here, we are trying to prove that at least $0.9 \boldsymbol{\Lambda}_i$ tasks of task requester $r_j^i$ will be executed out of $\boldsymbol{\Lambda}_i$ tasks is bounded above by $\bigg(\frac{10}{9\lg \boldsymbol{\Lambda}_i}\bigg)$ quantity. For this purpose we will take the help of the random variable $\boldsymbol{\mathcal{Z}}$ defined in Lemma \ref{lemma:1l}. Following Lemma \ref{lemma:1l}, we can write

\begin{equation}
\label{equ:98}
I = 
\begin{cases}
1, & {if~\boldsymbol{\mathcal{Z}} \geq 0.9\boldsymbol{\Lambda}_i},\\
0, & {otherwise} 
\end{cases} 
\end{equation}
Now, considering the case we are interested in from the above case representation form $i.e.$ $\boldsymbol{\mathcal{Z}} \geq 0.9\boldsymbol{\Lambda}_i$. We can rewrite it as $\frac{\boldsymbol{Z}}{0.9\boldsymbol{\Lambda}_i} \geq 1$. From equation \ref{equ:98}, $I$ is same as 1, so we get
\begin{equation}
\label{equ:99}
 \frac{\boldsymbol{\mathcal{Z}}}{0.9\boldsymbol{\Lambda}_i} \geq I
\end{equation}   
Taking expectation both side in equation \ref{equ:99}, we get
\begin{equation*}
E \bigg[\frac{\boldsymbol{\mathcal{Z}}}{0.9\boldsymbol{\Lambda}_i}\bigg] \geq E[I]
\end{equation*}
or
\begin{equation}
\label{equ:100}
E[I] \leq E \bigg[\frac{\boldsymbol{\mathcal{Z}}}{0.9\boldsymbol{\Lambda}_i}\bigg] = \frac{1}{0.9\boldsymbol{\Lambda}_i}E[\boldsymbol{\mathcal{Z}}]  
\end{equation} 

\begin{equation}
\label{equ:101}
\hspace*{12mm}= \frac{10}{9 \boldsymbol{\Lambda}_i} \cdot E [\boldsymbol{\mathcal{Z}}] 
\end{equation}   
From the definition of expectation, equation \ref{equ:101} can be rewritten as
\begin{equation}
\label{equ:101}
P\bigg\{\boldsymbol{\mathcal{Z}} \geq 0.9\boldsymbol{\Lambda}_i\bigg\} \cdot 1 \leq \bigg(\frac{10}{9 \boldsymbol{\Lambda}_i}\bigg) \cdot E [\boldsymbol{\mathcal{Z}}] 
\end{equation}
Substituting the value of $E [\boldsymbol{\mathcal{Z}}]$ in equation \ref{equ:101} from Lemma \ref{lemma:1l}, we get
\begin{equation*}
P\bigg\{\boldsymbol{\mathcal{Z}} \geq 0.9\boldsymbol{\Lambda}_i\bigg\} \leq \bigg(\frac{10}{9 \boldsymbol{\Lambda}_i}\bigg) \cdot \bigg(\frac{\boldsymbol{\Lambda}_i}{\lg \boldsymbol{\Lambda}_i}\bigg)
\end{equation*}

\begin{equation*}
\hspace*{12mm} =  \frac{10}{9\lg \boldsymbol{\Lambda}_i}
\end{equation*}

Hence proved.
\end{proof}  
\end{lemma}

\begin{observation}
If we consider the value of $\boldsymbol{\Lambda}_i$ as 1000, then $P\{\boldsymbol{\mathcal{Z}} \geq 900\} \leq  \frac{10}{9\lg 1000} = \frac{10}{27} = 0.37$. So, the probability that at least 900 tasks of $r_j^i$ will be executed out of 1000 will be less than or equal to $0.37$, which is quite low.
\end{observation}

\section{Experiments and Results}\label{sec:ef}
To complement our theoretical analysis, the QUAD is compared with the benchmark mechanisms namely PPM \cite{T.roughgarden_20141} and McAfee double auction (DA) \cite{RPM_The_1992, Jbre_INte_2005}. In order to measure the performance of QUAD in terms of \emph{truthfulness} and \emph{budget balanced} properties, it is compared with PPM on the following evaluation metrics: (1) Utility of agents, and (2) Utility of platform. 
The evaluation metric mentioned in point 1 above will be helpful in showing that QUAD is not vulnerable to manipulation whereas in case of PPM the participating agents gain by misreporting their true values. On the other hand the metric mentioned in point 2 will help us to show that QUAD is weakly budget balanced. It means that platform will have non-zero utility. Another direction of performance measurement of QUAD is in terms of satisfaction level of the IoT devices. Higher the incentives given to the IoT devices more satisfied and motivated the IoT devices will be in MCS market. In order to measure the performance of QUAD on the basis of satisfaction level, it is compared with McAfee DA on the ground of total charge. Higher the value of charges paid to the IoT devices, higher will be the satisfaction level of the IoT devices. The unit of bid values of the task requesters and IoT devices is taken as \$. The experiments are carried out in Java.
\subsection{Simulation Set-up}
For the simulation purpose, the task requesters and IoT devices are considered from 5 different categories. In each category, the number of task requesters, the number of IoT devices, and the bid value ranges of the task requesters and the IoT devices are considered as depicted in Table \ref{table_example3} below. In Table \ref{table_example3} the symbol $\mu$ represents the mean and symbol $\sigma$ represents the standard deviation of the bid values generated using normal distribution. In the experiments, the valuation of the agents are distributed among their respective virtual agents in such a way that it follows \emph{decreasing marginal return}. In our case, the experiment is executed for all the 5 categories for 100 times and the required values are plotted by taking the average over these 100 times.  In order to strengthen our claim the simulation is carried out for two different probability distributions, namely \emph{random distribution} (RanD) and \emph{normal distribution} (NanD). 

\begin{table}[H]
\renewcommand{\arraystretch}{0.99}  
\caption{Data set utilized for Simulation}
\label{table_example3}
\centering
  \begin{tabular}{c||c|c}
\hline
\textbf{Parameters} & \textbf{Values} & \textbf{Description}\\
\hline
$m_i$ & $\{5,~10,~15,~20,~25,~30\}$ & Number of task requesters in each category.\\
\hline
$n_i$ & $\{15,~30,~45,~60,~75,~90\}$ & Number of IoT devices in each category.\\
\hline
$\boldsymbol{\nu}^r$ & $[8,~30]$ & Bid value range of task requesters for RanD.\\
\hline
$\boldsymbol{\nu}^\mathcal{I}$ & $[5,~25]$ & Bid value range of IoT devices for RanD.\\
\hline
$\boldsymbol{\nu}^r$ (for NanD) & [$\mu = 15$, ~$\sigma = 4$] & Bid value determination of task requesters for NanD.\\
\hline
$\boldsymbol{\nu}^\mathcal{I}$ (for NanD) & [$\mu = 16$, ~$\sigma = 5$] & Bid value determination of IoT devices for NanD.\\

\hline
\end{tabular}
\end{table}
For comparing the QUAD and PPM on the ground of truthfulness, it is considered that $50\%$ of the agents misreport their true bid values in case of PPM and is represented as PPM-D in the simulation results. By misreporting the bid value it is meant that the agents will report the increased bid price so as to stay for longer duration of time in the MCS market.     
\subsection{Results and Comparison}
\label{ref:rc}
In this section, the discussion will mainly circumvent around the evaluation metrics depicted above. Figure \ref{fig:sim1} compares the two mechanisms namely QUAD and PPM on the ground of truthfulness of the agents for RanD and NanD cases. From the construction of QUAD and PPM, the utility of the participating agents in case of QUAD can be more or less than the utility of the agents in case of PPM as depicted in Figure \ref{fig:simA} and \ref{fig:simB}. However, in case of PPM, when $50\%$ of the agents misreport their true bid values then the utility of the agents is more as compared to the case when all the agents reports truthfully. It is natural as PPM is vulnerable to manipulation, so the participating agents gain by misreporting their true bid values. 
\begin{figure}[H]
\begin{subfigure}[b]{0.49\textwidth}
                \includegraphics[scale=0.65]{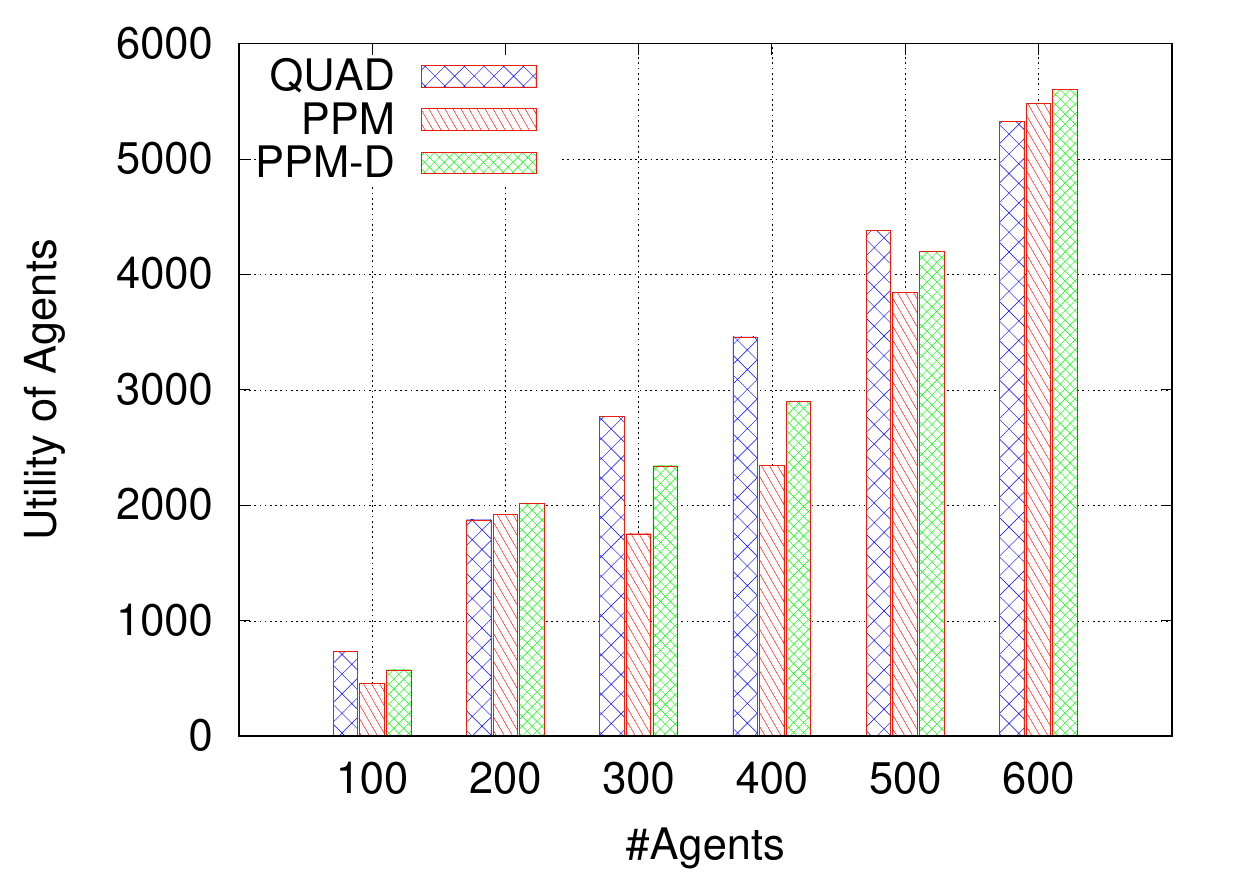}
                \subcaption{Utility of Agents (RanD)}
                \label{fig:simA}
        \end{subfigure}%
        \begin{subfigure}[b]{0.49\textwidth}
                \includegraphics[scale=0.65]{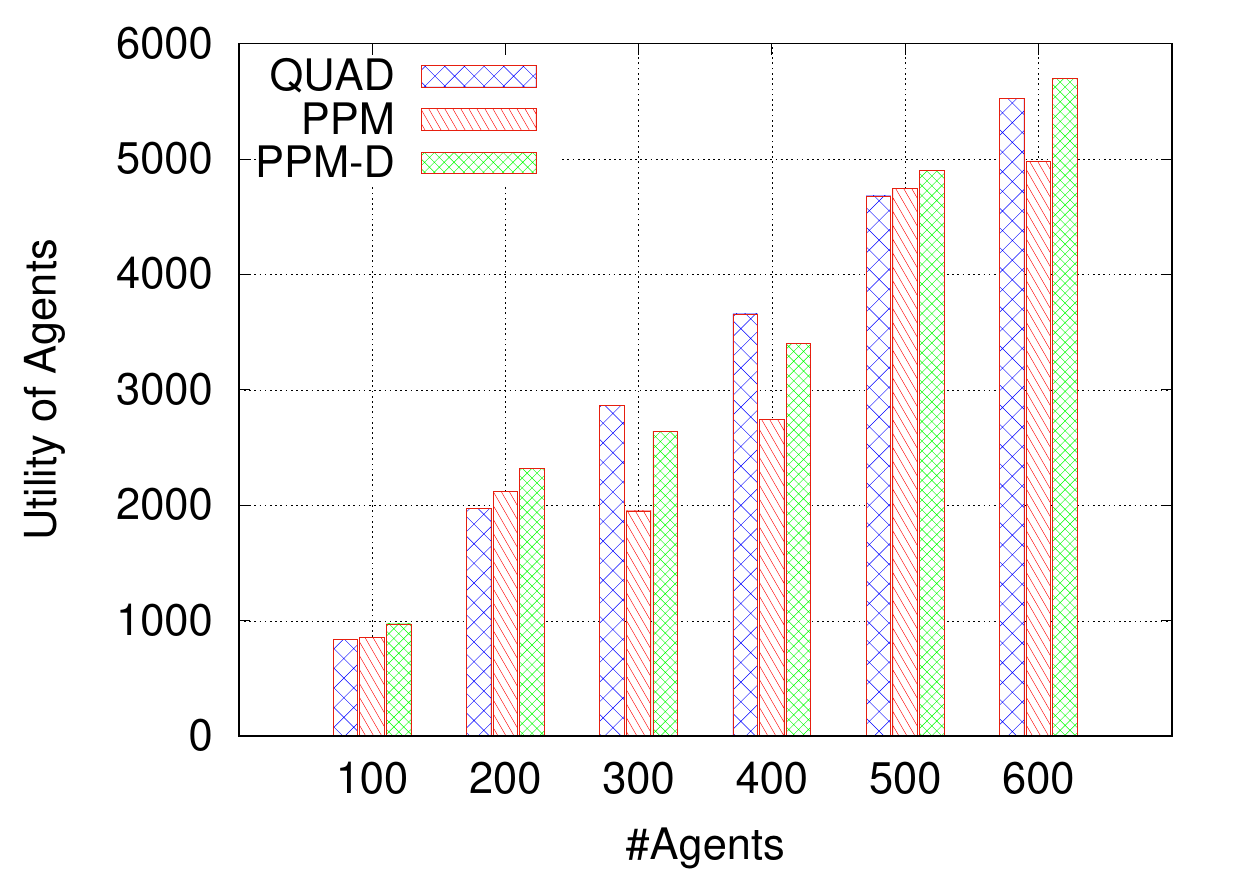}
        \subcaption{Utility of Agents (NanD)}\label{fig:simB}
        \end{subfigure}
        \caption{Comparison of Utility of Agents for RanD and NanD cases}
        \label{fig:sim1}
\end{figure}
 \indent  In Figure \ref{fig:sim1}, in some cases it can be seen that the utility of the agents in case of PPM with $50\%$ deviation is even more than the utility of the agents in case of QUAD. On the other hand, as QUAD works on the principle of Vickery auction \cite{WVic_Fin_1961} and so is not vulnerable to manipulation. It means that the participating agents can maximize their utility by only reporting truthfully in case of QUAD and is seen in Figure \ref{fig:simA} and \ref{fig:simB}.\\
 \indent Considering the second evaluation metric $i.e.$ utility of platform. In Figure \ref{fig:simC} and \ref{fig:simD} it can be seen that the utility of platform in case of QUAD could be more or less as compared to the utility of platform in case of PPM. However, for both the mechanisms $i.e.$ QUAD and PPM the utility of the platform will be non-negative irrespective of the type of distribution followed by the bid values of the agents. In such case it can be inferred from Figure \ref{fig:simC} and Figure \ref{fig:simD} that the two mechanisms is weakly budget balanced (see Definition \ref{def:4}). From the above discussion it can be inferred that the QUAD is truthful and both the mechanism satisfies one of the economic properties called budget balance.           
\begin{figure}[H]
\begin{subfigure}[b]{0.49\textwidth}
                \includegraphics[scale=0.65]{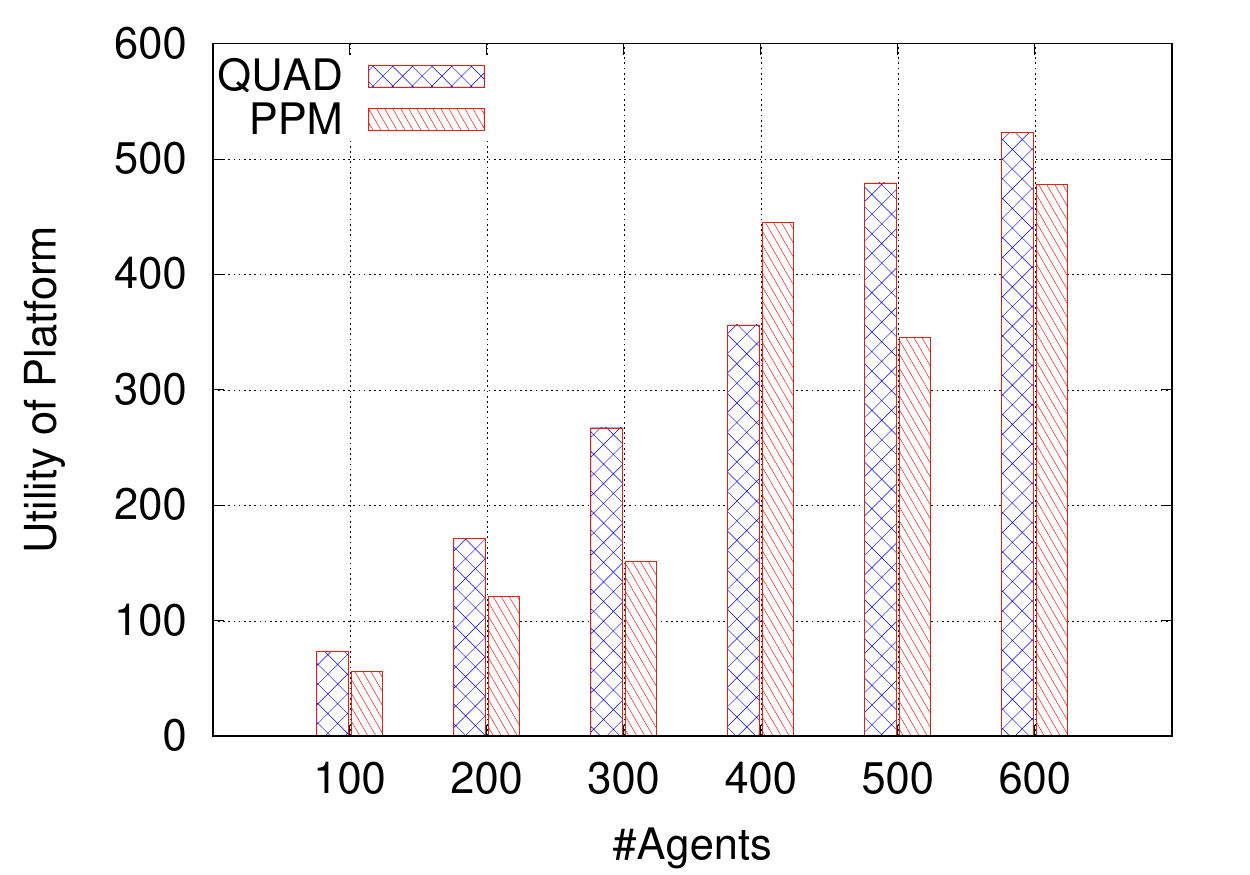}
                \subcaption{Utility of Platform (RanD)}
                \label{fig:simC}
        \end{subfigure}%
        \begin{subfigure}[b]{0.49\textwidth}
                \includegraphics[scale=0.65]{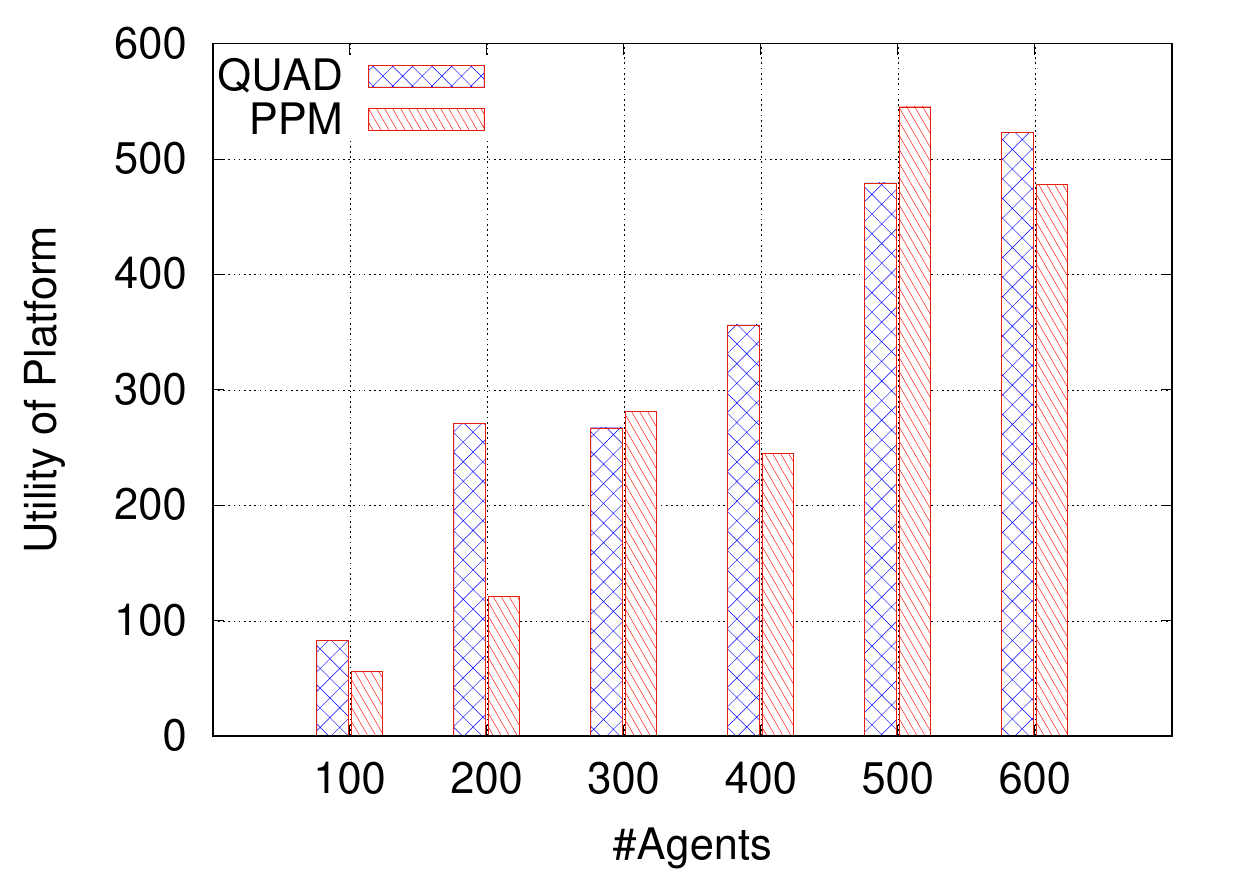}
        \subcaption{Utility of Platform (NanD)}
        \label{fig:simD}
        \end{subfigure}
        \caption{Comparison of Utility of Platform for RanD and NanD cases}
        \label{fig:sim22}
\end{figure}
   
\begin{figure}[H]
\begin{subfigure}[b]{0.49\textwidth}
                \includegraphics[scale=0.38]{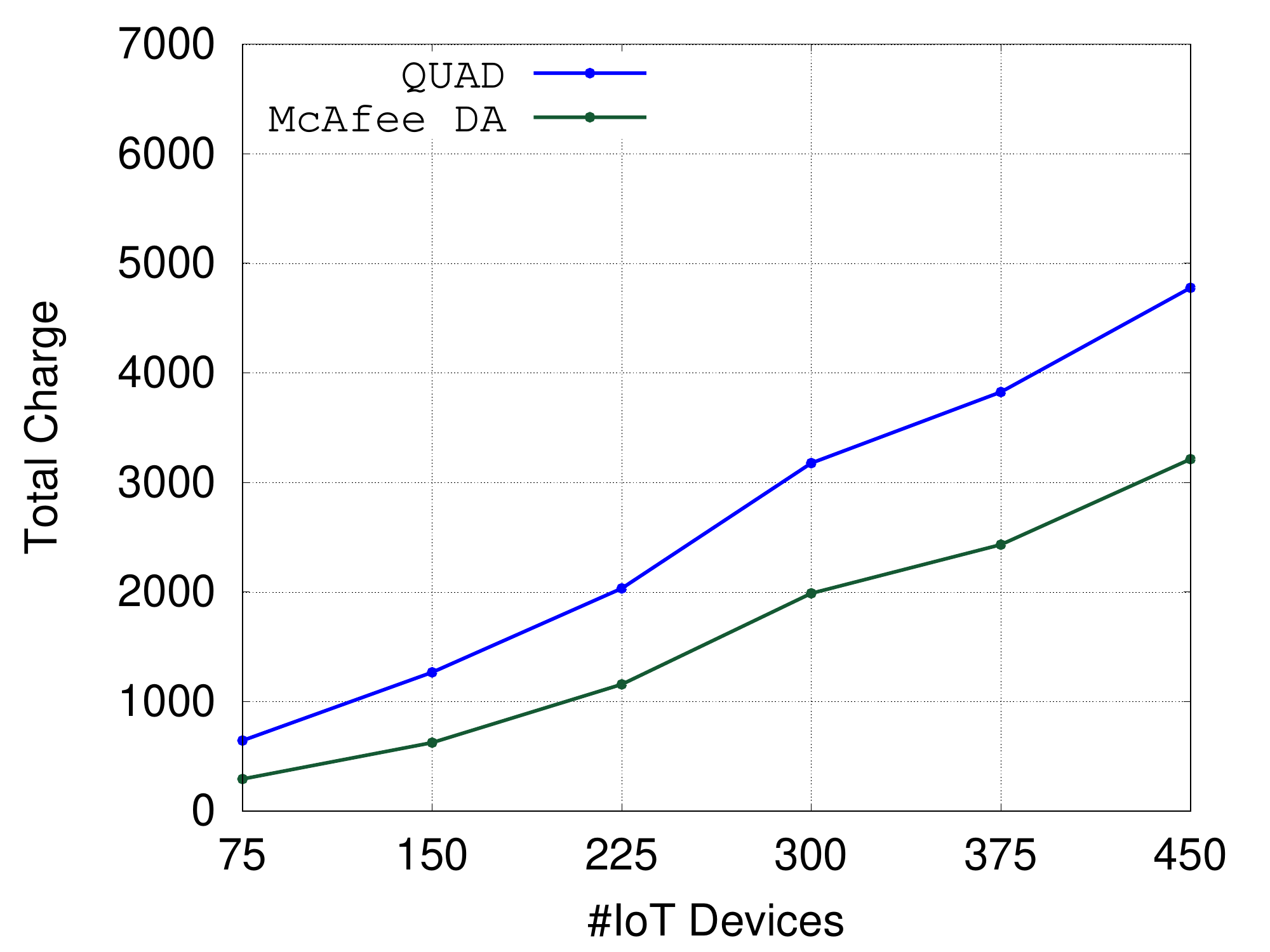}
                \subcaption{Total Charge Received by IoT Devices (RanD)}
                \label{fig:simE2}
        \end{subfigure}%
        \begin{subfigure}[b]{0.49\textwidth}
                \includegraphics[scale=0.38]{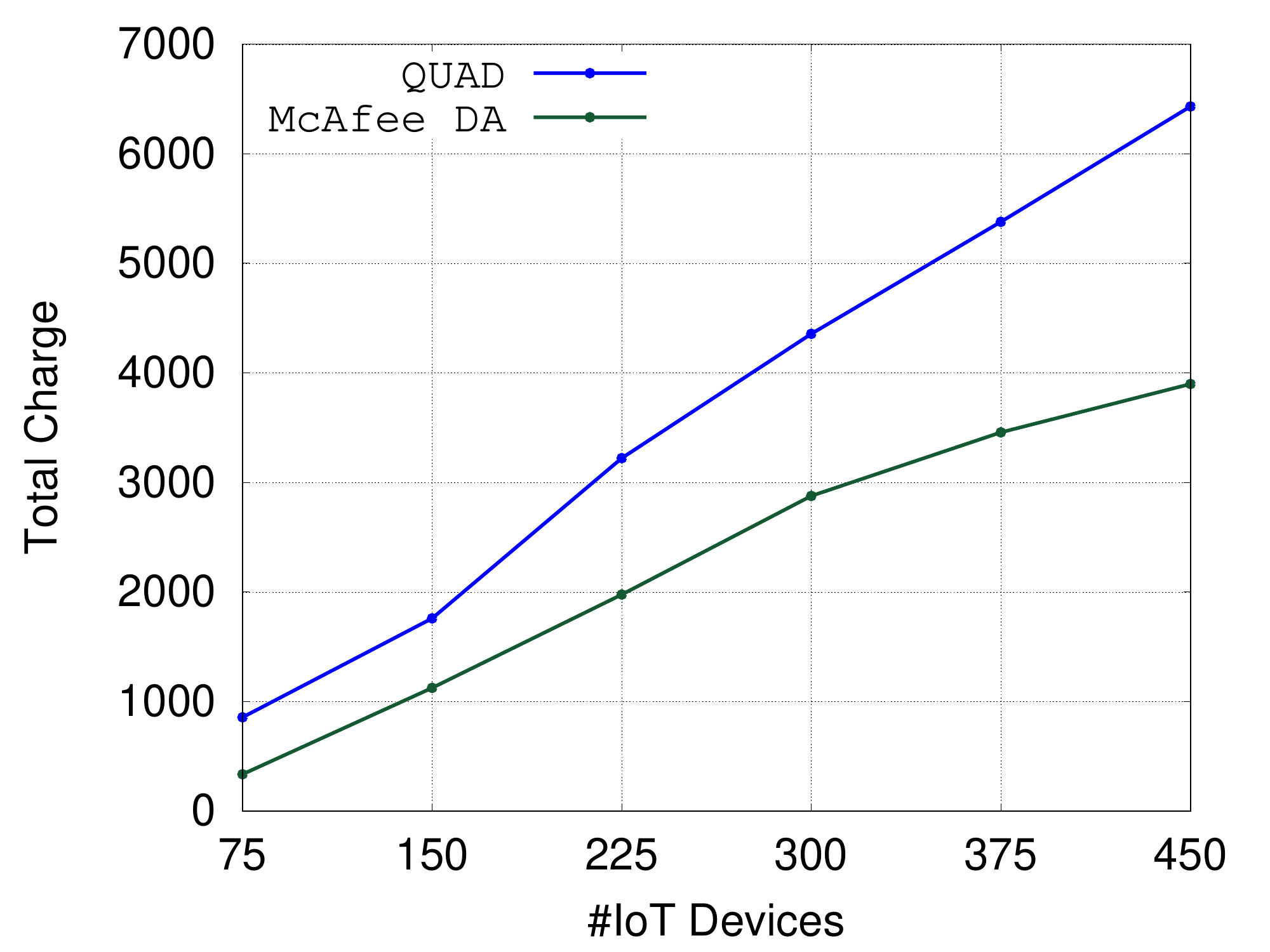}
        \subcaption{Total Charge Received by IoT Devices (NanD)}
        \label{fig:simF2}
        \end{subfigure}
        \caption{Comparison of Charge received by IoT Devices for RanD and NanD cases}
        \label{fig:sim33}
\end{figure}
\indent Figure \ref{fig:simE2} and \ref{fig:simF2} compares QUAD and McAfee double auction (DA) on the basis of total charges made to the IoT devices for both RanD and NanD. It can be seen that the total charge of the IoT devices is high in case of QUAD a compared to McAfee DA in both RanD and NanD. As it is obvious because in case of posted price ascending auctions the selling price is generally high as compared to the sealed bid auction and so the IoT devices are paid high. Due to this reason the satisfaction level of winning IoT devices in case QUAD will be more than the satisfaction level of the winning task requesters in case of McAfee DA. So, if it is seen from the perspective of IoT devices then QUAD will be more viable mechanism as compared to McAfee DA.
\begin{figure}[H]
\begin{subfigure}[b]{0.49\textwidth}
                \includegraphics[scale=0.38]{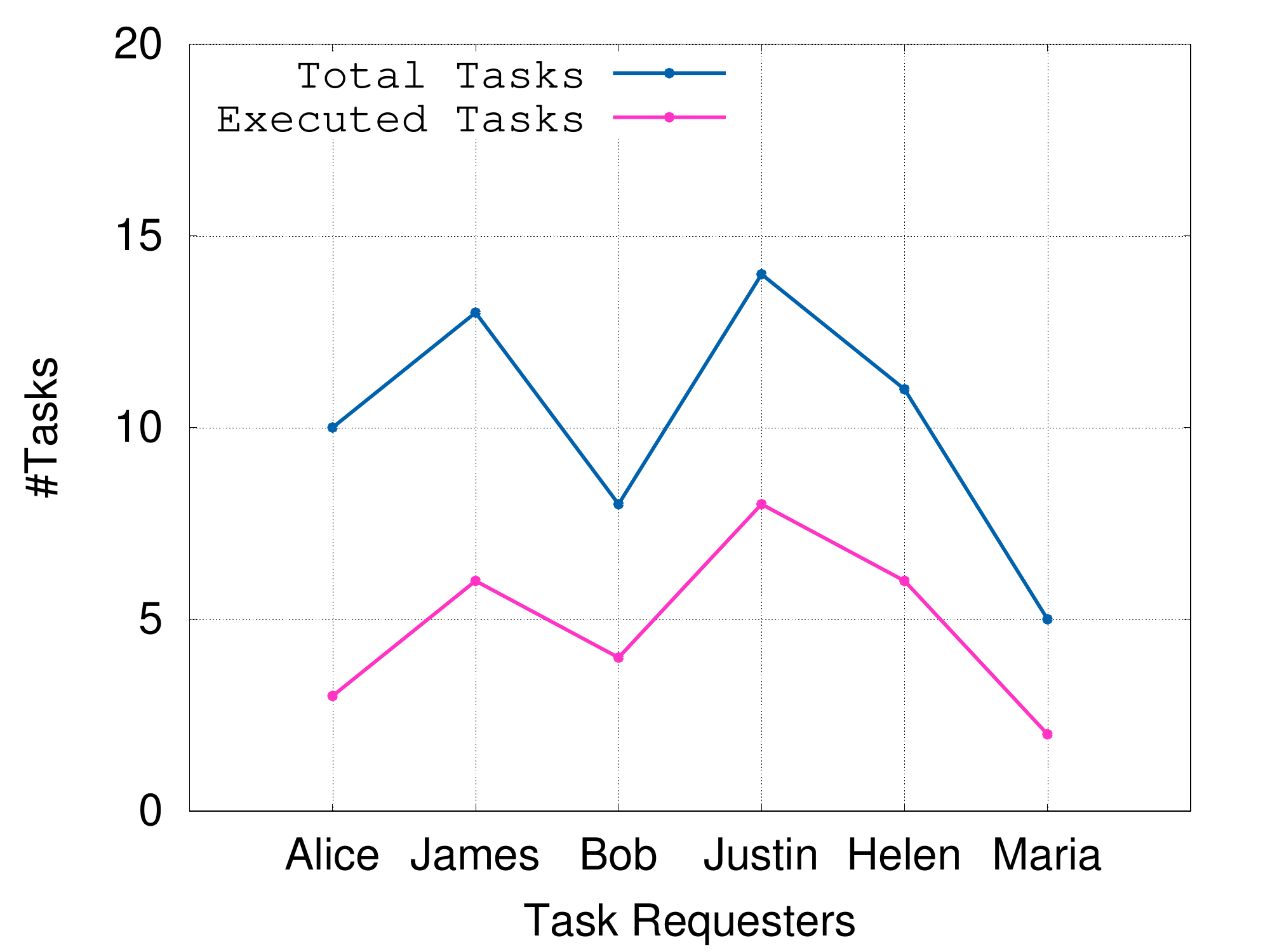}
                \subcaption{$\#$Tasks Executed by IoT Devices (RanD)}
                \label{fig:simE}
        \end{subfigure}%
        \begin{subfigure}[b]{0.49\textwidth}
                \includegraphics[scale=0.38]{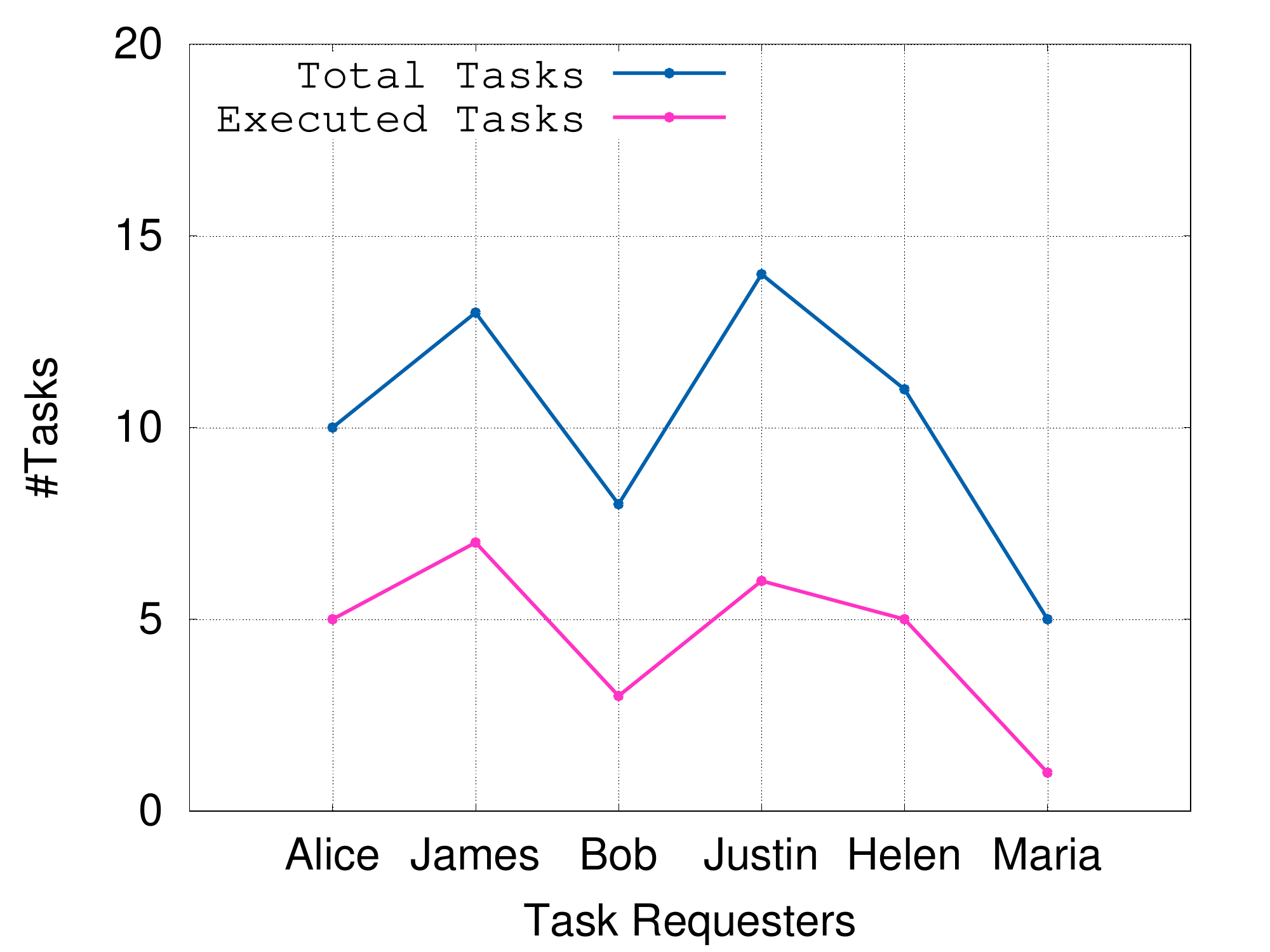}
        \subcaption{$\#$Tasks Executed by IoT Devices (NanD)}
        \label{fig:simF}
        \end{subfigure}
        \caption{Comparison of $\#$Tasks Executed by IoT Devices for RanD and NanD cases}
        \label{fig:sim33}
\end{figure}

\indent The simulation results in Figure \ref{fig:simE} and \ref{fig:simF} shows that in case of QUAD the average number of tasks executed by the quality IoT devices is almost half the total number of available tasks. The x-axis of the graphs shown in Figure \ref{fig:simE} and \ref{fig:simF} represents the task requesters and y-axis represents the number of tasks. For all the 5 task requesters the number of tasks that got executed is almost half the number of tasks endowed by the respective task requesters irrespective of the distribution of data. The similar results we obtained from the probabilistic analysis carried out in Lemma \ref{lemma:1l}.

\section{Conclusion and Future Directions}
\label{se:conc}
In this paper, for the discussed scenario in IoT-based mobile crowdsensing, a truthful mechanism is proposed that along with satisfying other economic properties such as \emph{individual rationality}, and \emph{budget balance} keeps into account the quality of IoT devices. By analysis it is shown that QUAD is \emph{computationally efficient}, \emph{truthful}, \emph{individual rational}, and \emph{prior free}. Further the probabilistic analysis is carried out to have an estimate on the number of tasks of any task requester got executed by the quality IoT devices among the available ones. The simulation results compares QUAD with the benchmark mechanisms on the ground of truthfulness, budget balance, and satisfaction level of the IoT devices.\\
\indent In future, one of the challenges could be to design a \emph{truthful} mechanism for the set-up discussed in this paper for general valuation (not DMR). Another direction could be, considering the set-up with multiple task providers and multiple IoT devices, where the task providers will provide the heterogeneous tasks. Here, the challenge will be to design a \emph{truthful} mechanism for the cases when the participating agents have DMR valuation or general valuation.

\bibliographystyle{acmsmall}
\bibliography{phd}

\begin{thebibliography}{}

\bibitem[\protect\citeauthoryear{Abraham, Dash, Rodrigues, Acharya, and
  Pani}{Abraham et~al\mbox{.}}{2021}]{abraham2021ai}
{\sc Abraham, A.}, {\sc Dash, S.}, {\sc Rodrigues, J.}, {\sc Acharya, B.}, {\sc
  and} {\sc Pani, S.} 2021.
\newblock {\em AI, Edge and IoT-based Smart Agriculture}.
\newblock Intelligent Data-Centric Systems. Elsevier Science, Barcelona, Spain.

\bibitem[\protect\citeauthoryear{Abualsaud, Elfouly, Khattab, Yaacoub, Ismail,
  Ahmed, and Guizani}{Abualsaud et~al\mbox{.}}{2019}]{8570744}
{\sc Abualsaud, K.}, {\sc Elfouly, T.~M.}, {\sc Khattab, T.}, {\sc Yaacoub,
  E.}, {\sc Ismail, L.~S.}, {\sc Ahmed, M.~H.}, {\sc and} {\sc Guizani, M.}
  2019.
\newblock A survey on mobile crowd-sensing and its applications in the {I}o{T}
  era.
\newblock {\em IEEE Access\/}~{\em 7}, 3855--3881.

\bibitem[\protect\citeauthoryear{Al-Fuqaha, Guizani, Mohammadi, Aledhari, and
  Ayyash}{Al-Fuqaha et~al\mbox{.}}{2015}]{7123563}
{\sc Al-Fuqaha, A.}, {\sc Guizani, M.}, {\sc Mohammadi, M.}, {\sc Aledhari,
  M.}, {\sc and} {\sc Ayyash, M.} 2015.
\newblock Internet of {T}hings: A survey on enabling technologies, protocols,
  and applications.
\newblock {\em IEEE Communications Surveys Tutorials\/}~{\em 17,\/}~4,
  2347--2376.

\bibitem[\protect\citeauthoryear{Alavi, Jiao, Buttlar, and Lajnef}{Alavi
  et~al\mbox{.}}{2018}]{Alavi2018InternetOT}
{\sc Alavi, A.~H.}, {\sc Jiao, P.}, {\sc Buttlar, W.~G.}, {\sc and} {\sc
  Lajnef, N.} 2018.
\newblock Internet of {T}hings-enabled smart cities: State-of-the-art and
  future trends.
\newblock {\em Measurement, Elsevier\/}~{\em 129}, 589--606.

\bibitem[\protect\citeauthoryear{Ang, Seng, and Ngharamike}{Ang
  et~al\mbox{.}}{2022}]{fi14020049}
{\sc Ang, K. L.~M.}, {\sc Seng, J. K.~P.}, {\sc and} {\sc Ngharamike, E.} 2022.
\newblock Towards crowdsourcing internet of things (crowd-iot): Architectures,
  security and applications.
\newblock {\em Future Internet, \emph{MDPI}\/}~{\em 14,\/}~2.

\bibitem[\protect\citeauthoryear{Ashton}{Ashton}{1999}]{Ashton1999ThatO}
{\sc Ashton, K.} 1999.
\newblock That ‘internet of things’ thing.

\bibitem[\protect\citeauthoryear{Atzori, Iera, and Morabito}{Atzori
  et~al\mbox{.}}{2010}]{ATZORI20102787}
{\sc Atzori, L.}, {\sc Iera, A.}, {\sc and} {\sc Morabito, G.} 2010.
\newblock The internet of things: A survey.
\newblock {\em Computer Networks, \emph{Elsevier},\/}~{\em 54,\/}~15,
  2787--2805.

\bibitem[\protect\citeauthoryear{Aubry, Silverston, Lahmadi, and Festor}{Aubry
  et~al\mbox{.}}{2014}]{6815170}
{\sc Aubry, E.}, {\sc Silverston, T.}, {\sc Lahmadi, A.}, {\sc and} {\sc
  Festor, O.} 2014.
\newblock Crowdout: A mobile crowdsourcing service for road safety in digital
  cities.
\newblock In {\em 2014 IEEE International Conference on Pervasive Computing and
  Communication Workshops (PERCOM WORKSHOPS)}. Budapest, Hungary, 86--91.

\bibitem[\protect\citeauthoryear{Baccour, Erbad, Mohamed, Hamdi, and
  Guizani}{Baccour et~al\mbox{.}}{2020}]{9322470}
{\sc Baccour, E.}, {\sc Erbad, A.}, {\sc Mohamed, A.}, {\sc Hamdi, M.}, {\sc
  and} {\sc Guizani, M.} 2020.
\newblock Distprivacy: Privacy-aware distributed deep neural networks in iot
  surveillance systems.
\newblock In {\em GLOBECOM 2020 - 2020 IEEE Global Communications Conference}.
  Taipei, Taiwan, 1--6.

\bibitem[\protect\citeauthoryear{Blumrosen and Dobzinski}{Blumrosen and
  Dobzinski}{2014}]{Blumrosen2014ReallocationM}
{\sc Blumrosen, L.} {\sc and} {\sc Dobzinski, S.} 2014.
\newblock Reallocation mechanisms.
\newblock In {\em Proceedings of the Fifteenth ACM Conference on Economics and
  Computation}. EC '14. Association for Computing Machinery, New York, NY, USA,
  617.

\bibitem[\protect\citeauthoryear{Bredin and Parkes}{Bredin and
  Parkes}{2005}]{Jbre_INte_2005}
{\sc Bredin, J.} {\sc and} {\sc Parkes, D.~C.} 2005.
\newblock Models for truthful online double auctions.
\newblock In {\em $21^{st}$ International Conference on Uncertainty in
  Artificial Intelligence (UAI)}. AUAI press, 50--59.

\bibitem[\protect\citeauthoryear{Cormen, Leiserson, Rivest, and Stein}{Cormen
  et~al\mbox{.}}{2009}]{Coreman_2009}
{\sc Cormen, T.~H.}, {\sc Leiserson, C.~E.}, {\sc Rivest, R.~L.}, {\sc and}
  {\sc Stein, C.} 2009.
\newblock {\em Introduction to algorithms}.
\newblock MIT press.

\bibitem[\protect\citeauthoryear{Dasari, Kantarci, Pouryazdan, Foschini, and
  Girolami}{Dasari et~al\mbox{.}}{2020}]{s20072055}
{\sc Dasari, V.~S.}, {\sc Kantarci, B.}, {\sc Pouryazdan, M.}, {\sc Foschini,
  L.}, {\sc and} {\sc Girolami, M.} 2020.
\newblock Game theory in mobile crowdsensing: A comprehensive survey.
\newblock {\em Sensors\/}~{\em 20,\/}~7.

\bibitem[\protect\citeauthoryear{Deshmukh, Goldberg, Hartline, and
  Karlin}{Deshmukh et~al\mbox{.}}{2002}]{10.1007/3-540-45749-6_34}
{\sc Deshmukh, K.}, {\sc Goldberg, A.~V.}, {\sc Hartline, J.~D.}, {\sc and}
  {\sc Karlin, A.~R.} 2002.
\newblock Truthful and competitive double auctions.
\newblock In {\em Algorithms --- ESA 2002}, {R.~M{\"o}hring} {and} {R.~Raman},
  Eds. Springer Berlin Heidelberg, Berlin, Heidelberg, 361--373.

\bibitem[\protect\citeauthoryear{Enigo, Kumar, Vijay, and Prabu}{Enigo
  et~al\mbox{.}}{2016}]{ENIGO2016316}
{\sc Enigo, V.~F.}, {\sc Kumar, T.~V.}, {\sc Vijay, S.}, {\sc and} {\sc Prabu,
  K.} 2016.
\newblock Crowdsourcing based online petitioning system for pothole detection
  using android platform.
\newblock {\em Procedia Computer Science\/}~{\em 87}, 316--321.
\newblock Fourth International Conference on Recent Trends in Computer Science
  \& Engineering (ICRTCSE 2016).

\bibitem[\protect\citeauthoryear{Feldman and Gonen}{Feldman and
  Gonen}{2016}]{DBLP:journals/corr/FeldmanG16}
{\sc Feldman, M.} {\sc and} {\sc Gonen, R.} 2016.
\newblock Markets with strategic multi-minded mediators.
\newblock {\em CoRR\/}~{\em abs/1603.08717}.

\bibitem[\protect\citeauthoryear{Feng, Zhu, Zhang, Ni, and Vasilakos}{Feng
  et~al\mbox{.}}{2014}]{6848055}
{\sc Feng, Z.}, {\sc Zhu, Y.}, {\sc Zhang, Q.}, {\sc Ni, L.~M.}, {\sc and} {\sc
  Vasilakos, A.~V.} 2014.
\newblock {TRAC}: Truthful auction for location-aware collaborative sensing in
  mobile crowdsourcing.
\newblock In {\em IEEE INFOCOM 2014 - IEEE Conference on Computer
  Communications}. Toronto, ON, Canada, 1231--1239.

\bibitem[\protect\citeauthoryear{Garc{\'i}a, Marin, Prieto, and
  Nieto}{Garc{\'i}a et~al\mbox{.}}{2017}]{Garca2017AnalysisOS}
{\sc Garc{\'i}a, P. A.~G.}, {\sc Marin, C.}, {\sc Prieto, J.~D.}, {\sc and}
  {\sc Nieto, Y.~V.} 2017.
\newblock Analysis of security mechanisms based on clusters iot environments.
\newblock {\em Int. J. Interact. Multim. Artif. Intell.\/}~{\em 4}, 55--60.

\bibitem[\protect\citeauthoryear{Goel, Leonardi, Mirrokni, Nikzad, and
  Paes-Leme}{Goel et~al\mbox{.}}{2016}]{45749}
{\sc Goel, G.}, {\sc Leonardi, S.}, {\sc Mirrokni, V.}, {\sc Nikzad, A.}, {\sc
  and} {\sc Paes-Leme, R.} 2016.
\newblock Reservation exchange markets for internet advertising.
\newblock {\em LIPIcs\/}~{\em 55}, 142:1--142:13.

\bibitem[\protect\citeauthoryear{Gonen, Gonen, and Pavlov}{Gonen
  et~al\mbox{.}}{2007}]{10.1145/1250910.1250914}
{\sc Gonen, M.}, {\sc Gonen, R.}, {\sc and} {\sc Pavlov, E.} 2007.
\newblock Generalized trade reduction mechanisms.
\newblock In {\em Proceedings of the 8th ACM Conference on Electronic
  Commerce}. EC '07. Association for Computing Machinery, New York, NY, USA,
  20–29.

\bibitem[\protect\citeauthoryear{Gonen and Egri}{Gonen and
  Egri}{2017}]{Gonen2017DYCOMAD}
{\sc Gonen, R.} {\sc and} {\sc Egri, O.} 2017.
\newblock Dycom: A dynamic truthful budget balanced double-sided combinatorial
  market.
\newblock In {\em AAMAS}.

\bibitem[\protect\citeauthoryear{Hamrouni, Alelyani, Ghazzai, and
  Massoud}{Hamrouni et~al\mbox{.}}{2021}]{9474925}
{\sc Hamrouni, A.}, {\sc Alelyani, T.}, {\sc Ghazzai, H.}, {\sc and} {\sc
  Massoud, Y.} 2021.
\newblock Toward collaborative mobile crowdsourcing.
\newblock {\em IEEE Internet of Things Magazine\/}~{\em 4,\/}~2, 88--94.

\bibitem[\protect\citeauthoryear{Heuveldop}{Heuveldop}{2017}]{Heuveldop_2017}
{\sc Heuveldop, N.} 2017.
\newblock Ericsson mobility report.

\bibitem[\protect\citeauthoryear{Hirai and Sato}{Hirai and
  Sato}{2021}]{doi:10.1287/moor.2021.1124}
{\sc Hirai, H.} {\sc and} {\sc Sato, R.} 2021.
\newblock Polyhedral clinching auctions for two-sided markets.
\newblock {\em Mathematics of Operations Research\/}~{\em 0,\/}~0, null.

\bibitem[\protect\citeauthoryear{Hou, Pei, and Sun}{Hou
  et~al\mbox{.}}{2019}]{book1}
{\sc Hou, F.}, {\sc Pei, Y.}, {\sc and} {\sc Sun, J.} 2019.
\newblock {\em Mobile Crowd Sensing: Incentive Mechanism Design}.
\newblock Springer, Cham.

\bibitem[\protect\citeauthoryear{Huang, Xin, Sun, and Yang}{Huang
  et~al\mbox{.}}{2017}]{7925505}
{\sc Huang, H.}, {\sc Xin, Y.}, {\sc Sun, Y.-E.}, {\sc and} {\sc Yang, W.}
  2017.
\newblock A truthful double auction mechanism for crowdsensing systems with
  max-min fairness.
\newblock In {\em 2017 IEEE Wireless Communications and Networking Conference
  (WCNC)}. 1--6.

\bibitem[\protect\citeauthoryear{Index}{Index}{2015}]{Huawei_2015}
{\sc Index, G.~C.} 2015.
\newblock Huawei technologies coo., ltd.,.

\bibitem[\protect\citeauthoryear{Khajenasiri, Estebsari, Verhelst, and
  Gielen}{Khajenasiri et~al\mbox{.}}{2017}]{KHAJENASIRI2017770}
{\sc Khajenasiri, I.}, {\sc Estebsari, A.}, {\sc Verhelst, M.}, {\sc and} {\sc
  Gielen, G.} 2017.
\newblock A review on internet of things solutions for intelligent energy
  control in buildings for smart city applications.
\newblock {\em Energy Procedia\/}~{\em 111}, 770--779.
\newblock 8th International Conference on Sustainability in Energy and
  Buildings, SEB-16, 11-13 September 2016, Turin, Italy.

\bibitem[\protect\citeauthoryear{Kong, Liu, Jedari, Li, Wan, and Xia}{Kong
  et~al\mbox{.}}{2019}]{8733838}
{\sc Kong, X.}, {\sc Liu, X.}, {\sc Jedari, B.}, {\sc Li, M.}, {\sc Wan, L.},
  {\sc and} {\sc Xia, F.} 2019.
\newblock Mobile crowdsourcing in smart cities: Technologies, applications, and
  future challenges.
\newblock {\em IEEE Internet of Things Journal\/}~{\em 6,\/}~5, 8095--8113.

\bibitem[\protect\citeauthoryear{Leyton-Brown, Milgrom, and Segal}{Leyton-Brown
  et~al\mbox{.}}{2017}]{Leyton-Brown7202}
{\sc Leyton-Brown, K.}, {\sc Milgrom, P.}, {\sc and} {\sc Segal, I.} 2017.
\newblock Economics and computer science of a radio spectrum reallocation.
\newblock {\em Proceedings of the National Academy of Sciences\/}~{\em
  114,\/}~28, 7202--7209.

\bibitem[\protect\citeauthoryear{Li, Wang, Zheng, and Franklin}{Li
  et~al\mbox{.}}{2016}]{7420720}
{\sc Li, G.}, {\sc Wang, J.}, {\sc Zheng, Y.}, {\sc and} {\sc Franklin, M.~J.}
  2016.
\newblock Crowdsourced data management: A survey.
\newblock {\em IEEE Transactions on Knowledge and Data Engineering\/}~{\em
  28,\/}~9, 2296--2319.

\bibitem[\protect\citeauthoryear{Li, Wang, Yin, Li, Qian, and Hu}{Li
  et~al\mbox{.}}{2015}]{fi7030329}
{\sc Li, H.}, {\sc Wang, H.}, {\sc Yin, W.}, {\sc Li, Y.}, {\sc Qian, Y.}, {\sc
  and} {\sc Hu, F.} 2015.
\newblock Development of a remote monitoring system for henhouse environment
  based on iot technology.
\newblock {\em Future Internet\/}~{\em 7,\/}~3, 329--341.

\bibitem[\protect\citeauthoryear{Li, Wu, and Zhu}{Li
  et~al\mbox{.}}{2019}]{9068628}
{\sc Li, J.}, {\sc Wu, J.}, {\sc and} {\sc Zhu, Y.} 2019.
\newblock Selecting optimal mobile users for long-term environmental monitoring
  by crowdsourcing.
\newblock In {\em 2019 IEEE/ACM 27th International Symposium on Quality of
  Service (IWQoS)}. 1--10.

\bibitem[\protect\citeauthoryear{Liu, Shen, Narman, Chung, and Lin}{Liu
  et~al\mbox{.}}{2018}]{10.1145/3185504}
{\sc Liu, J.}, {\sc Shen, H.}, {\sc Narman, H.~S.}, {\sc Chung, W.}, {\sc and}
  {\sc Lin, Z.} 2018.
\newblock A survey of mobile crowdsensing techniques: A critical component for
  the internet of things.
\newblock {\em ACM Trans. Cyber-Phys. Syst.\/}~{\em 2,\/}~3.

\bibitem[\protect\citeauthoryear{Liu and Chen}{Liu and
  Chen}{2016}]{10.5555/3061053.3061148}
{\sc Liu, Y.} {\sc and} {\sc Chen, Y.} 2016.
\newblock Learning to incentivize: Eliciting effort via output agreement.
\newblock In {\em Proceedings of the Twenty-Fifth International Joint
  Conference on Artificial Intelligence}. IJCAI'16. AAAI Press, 3782–3788.

\bibitem[\protect\citeauthoryear{Manyika, Chui, and Bisson}{Manyika
  et~al\mbox{.}}{2015}]{Manyika_2015}
{\sc Manyika, J.}, {\sc Chui, M.}, {\sc and} {\sc Bisson, P.} 2015.
\newblock The internet of things: mapping the value beyond the hype.

\bibitem[\protect\citeauthoryear{McAfee}{McAfee}{1992}]{RPM_The_1992}
{\sc McAfee, R. .~P.} 1992.
\newblock A dominant strategy double auction.
\newblock {\em Journal of Economic Theory\/}~{\em 56,\/}~2, 434--450.

\bibitem[\protect\citeauthoryear{Mukhopadhyay, Singh, Pal, and
  Kumar}{Mukhopadhyay et~al\mbox{.}}{2022}]{Mukhopadhyay2021}
{\sc Mukhopadhyay, J.}, {\sc Singh, V.~K.}, {\sc Pal, A.}, {\sc and} {\sc
  Kumar, A.} 2022.
\newblock A truthful budget feasible mechanism for iot-based participatory
  sensing with incremental arrival of budget.
\newblock {\em Journal of Ambient Intelligence and Humanized Computing\/}~{\em
  13}, 1107--1124.

\bibitem[\protect\citeauthoryear{Myerson and Satterthwaite}{Myerson and
  Satterthwaite}{1983}]{MYERSON1983265}
{\sc Myerson, R.~B.} {\sc and} {\sc Satterthwaite, M.~A.} 1983.
\newblock Efficient mechanisms for bilateral trading.
\newblock {\em Journal of Economic Theory\/}~{\em 29,\/}~2, 265--281.

\bibitem[\protect\citeauthoryear{Nagatani, Kiribayashi, Okada, Otake, Yoshida,
  Tadokoro, Nishimura, Yoshida, Koyanagi, Fukushima, and Kawatsuma}{Nagatani
  et~al\mbox{.}}{2013}]{Nagatani:2013:ERN:2421033.2421037}
{\sc Nagatani, K.}, {\sc Kiribayashi, S.}, {\sc Okada, Y.}, {\sc Otake, K.},
  {\sc Yoshida, K.}, {\sc Tadokoro, S.}, {\sc Nishimura, T.}, {\sc Yoshida,
  T.}, {\sc Koyanagi, E.}, {\sc Fukushima, M.}, {\sc and} {\sc Kawatsuma, S.}
  2013.
\newblock Emergency response to the nuclear accident at the fukushima daiichi
  nuclear power plants using mobile rescue robots.
\newblock {\em Journal of Field Robotics\/}~{\em 30,\/}~1, 44--63.

\bibitem[\protect\citeauthoryear{Nisan, Roughgarden, Tardos, and
  Vazirani}{Nisan et~al\mbox{.}}{2007}]{NNis_Pre_2007}
{\sc Nisan, N.}, {\sc Roughgarden, T.}, {\sc Tardos, E.}, {\sc and} {\sc
  Vazirani, V.~V.} 2007.
\newblock {\em Algorithmic Game Theory}.
\newblock Cambridge University Press, New York, NY, USA.

\bibitem[\protect\citeauthoryear{Pan, Yu, Miao, and Leung}{Pan
  et~al\mbox{.}}{2017}]{Pan2017CrowdsensingAQ}
{\sc Pan, Z.}, {\sc Yu, H.}, {\sc Miao, C.}, {\sc and} {\sc Leung, C.} 2017.
\newblock Crowdsensing air quality with camera-enabled mobile devices.
\newblock In {\em AAAI}.

\bibitem[\protect\citeauthoryear{Peng, Wu, and Chen}{Peng
  et~al\mbox{.}}{2015}]{10.1145/2746285.2746306}
{\sc Peng, D.}, {\sc Wu, F.}, {\sc and} {\sc Chen, G.} 2015.
\newblock Pay as how well you do: A quality based incentive mechanism for
  crowdsensing.
\newblock In {\em Proceedings of the 16th ACM International Symposium on Mobile
  Ad Hoc Networking and Computing}. MobiHoc '15. Association for Computing
  Machinery, New York, NY, USA, 177–186.

\bibitem[\protect\citeauthoryear{Pereira and Aguiar}{Pereira and
  Aguiar}{2014}]{s141019582}
{\sc Pereira, C.} {\sc and} {\sc Aguiar, A.} 2014.
\newblock Towards efficient mobile m2m communications: Survey and open
  challenges.
\newblock {\em Sensors\/}~{\em 14,\/}~10, 19582--19608.

\bibitem[\protect\citeauthoryear{Phuttharak and Loke}{Phuttharak and
  Loke}{2019}]{Phuttharak2019ARO}
{\sc Phuttharak, J.} {\sc and} {\sc Loke, S.~W.} 2019.
\newblock A review of mobile crowdsourcing architectures and challenges: Toward
  crowd-empowered internet-of-things.
\newblock {\em IEEE Access\/}~{\em 7}, 304--324.

\bibitem[\protect\citeauthoryear{Plott and Gray}{Plott and
  Gray}{1990}]{PLOTT1990245}
{\sc Plott, C.~R.} {\sc and} {\sc Gray, P.} 1990.
\newblock The multiple unit double auction.
\newblock {\em Journal of Economic Behavior \& Organization\/}~{\em 13,\/}~2,
  245--258.

\bibitem[\protect\citeauthoryear{Poblet, Garc{\'i}a-Cuesta, and
  Casanovas}{Poblet et~al\mbox{.}}{2014}]{Poblet2014}
{\sc Poblet, M.}, {\sc Garc{\'i}a-Cuesta, E.}, {\sc and} {\sc Casanovas, P.}
  2014.
\newblock Crowdsourcing tools for disaster management: A review of platforms
  and methods.
\newblock In {\em AI Approaches to the Complexity of Legal Systems},
  {P.~Casanovas}, {U.~Pagallo}, {M.~Palmirani}, {and} {G.~Sartor}, Eds.
  Springer Berlin Heidelberg, Berlin, Heidelberg, 261--274.

\bibitem[\protect\citeauthoryear{Pődör and Szabó}{Pődör and
  Szabó}{2021}]{doi:10.1177/2399808320987567}
{\sc Pődör, A.} {\sc and} {\sc Szabó, S.} 2021.
\newblock Geo-tagged environmental noise measurement with smartphones: Accuracy
  and perspectives of crowdsourced mapping.
\newblock {\em Environment and Planning B: Urban Analytics and City
  Science\/}~{\em 48,\/}~9, 2710--2725.

\bibitem[\protect\citeauthoryear{Qiu, Xiao, and Zhou}{Qiu
  et~al\mbox{.}}{2013}]{Qiu2013FrameworkAC}
{\sc Qiu, T.}, {\sc Xiao, H.}, {\sc and} {\sc Zhou, P.} 2013.
\newblock Framework and case studies of intelligence monitoring platform in
  facility agriculture ecosystem.
\newblock {\em 2013 Second International Conference on Agro-Geoinformatics
  (Agro-Geoinformatics)\/}, 522--525.

\bibitem[\protect\citeauthoryear{Qolomany, Mohammed, Al-Fuqaha, Guizani, and
  Qadir}{Qolomany et~al\mbox{.}}{2021}]{9187421}
{\sc Qolomany, B.}, {\sc Mohammed, I.}, {\sc Al-Fuqaha, A.}, {\sc Guizani, M.},
  {\sc and} {\sc Qadir, J.} 2021.
\newblock Trust-based cloud machine learning model selection for industrial iot
  and smart city services.
\newblock {\em IEEE Internet of Things Journal\/}~{\em 8,\/}~4, 2943--2958.

\bibitem[\protect\citeauthoryear{Rahman, Rashid, Barnes, Hossain, Hassanain,
  and Guizani}{Rahman et~al\mbox{.}}{2019}]{8766496}
{\sc Rahman, M.~A.}, {\sc Rashid, M.}, {\sc Barnes, S.}, {\sc Hossain, M.~S.},
  {\sc Hassanain, E.}, {\sc and} {\sc Guizani, M.} 2019.
\newblock An iot and blockchain-based multi-sensory in-home quality of life
  framework for cancer patients.
\newblock In {\em 2019 15th International Wireless Communications Mobile
  Computing Conference (IWCMC)}. 2116--2121.

\bibitem[\protect\citeauthoryear{Rana, Chou, Kanhere, Bulusu, and Hu}{Rana
  et~al\mbox{.}}{2010}]{10.1145/1791212.1791226}
{\sc Rana, R.~K.}, {\sc Chou, C.~T.}, {\sc Kanhere, S.~S.}, {\sc Bulusu, N.},
  {\sc and} {\sc Hu, W.} 2010.
\newblock Ear-phone: An end-to-end participatory urban noise mapping system.
\newblock In {\em Proceedings of the 9th ACM/IEEE International Conference on
  Information Processing in Sensor Networks}. IPSN '10. Association for
  Computing Machinery, New York, NY, USA, 105–116.

\bibitem[\protect\citeauthoryear{Ray}{Ray}{2018}]{RAY2018291}
{\sc Ray, P.} 2018.
\newblock A survey on internet of things architectures.
\newblock {\em Journal of King Saud University - Computer and Information
  Sciences\/}~{\em 30,\/}~3, 291--319.

\bibitem[\protect\citeauthoryear{Ray, Adhikary, lana, Mitra, Halder, Paul,
  Mukherjee, Koshika, De, Chakravorty, Goswami, Kundu, and Sarkar}{Ray
  et~al\mbox{.}}{2018}]{8614931}
{\sc Ray, S.}, {\sc Adhikary, P.}, {\sc lana, S.}, {\sc Mitra, S.}, {\sc
  Halder, T.}, {\sc Paul, M.}, {\sc Mukherjee, A.}, {\sc Koshika}, {\sc De,
  A.}, {\sc Chakravorty, N.}, {\sc Goswami, R.}, {\sc Kundu, V.}, {\sc and}
  {\sc Sarkar, D.} 2018.
\newblock A survey paper on architecture of internet of things.
\newblock In {\em 2018 IEEE 9th Annual Information Technology, Electronics and
  Mobile Communication Conference (IEMCON)}. 908--913.

\bibitem[\protect\citeauthoryear{Roughgarden}{Roughgarden}{2014}]{T.roughgarden_20141}
{\sc Roughgarden, T.} 2014.
\newblock {CS}364{B}: Frontiers in mechanism design, ({S}tanford {U}niversity
  {C}ourse), {L}ecture \#4: The clinching auction.

\bibitem[\protect\citeauthoryear{Roughgarden}{Roughgarden}{2016}]{T.roughgarden_2016}
{\sc Roughgarden, T.} 2016.
\newblock {CS}269{I}: Incentives in computer science ({S}tanford {U}niversity
  course).
\newblock Lecture 3: Strategic Voting.

\bibitem[\protect\citeauthoryear{Samulowska, Chmielewski, Raczko, Lupa,
  Myszkowska, and Zagajewski}{Samulowska et~al\mbox{.}}{2021}]{ijgi10020046}
{\sc Samulowska, M.}, {\sc Chmielewski, S.}, {\sc Raczko, E.}, {\sc Lupa, M.},
  {\sc Myszkowska, D.}, {\sc and} {\sc Zagajewski, B.} 2021.
\newblock Crowdsourcing without data bias: Building a quality assurance system
  for air pollution symptom mapping.
\newblock {\em ISPRS International Journal of Geo-Information\/}~{\em 10,\/}~2,
  1--28.

\bibitem[\protect\citeauthoryear{Segal-Halevi, Hassidim, and
  Aumann}{Segal-Halevi et~al\mbox{.}}{2018a}]{ijcai2018-68}
{\sc Segal-Halevi, E.}, {\sc Hassidim, A.}, {\sc and} {\sc Aumann, Y.} 2018a.
\newblock Double auctions in markets for multiple kinds of goods.
\newblock In {\em Proceedings of the Twenty-Seventh International Joint
  Conference on Artificial Intelligence, {IJCAI-18}}. International Joint
  Conferences on Artificial Intelligence Organization, 489--497.

\bibitem[\protect\citeauthoryear{Segal-Halevi, Hassidim, and
  Aumann}{Segal-Halevi et~al\mbox{.}}{2018b}]{SegalHalevi2018MUDAAT}
{\sc Segal-Halevi, E.}, {\sc Hassidim, A.}, {\sc and} {\sc Aumann, Y.} 2018b.
\newblock {MUDA}: A truthful multi-unit double-auction mechanism.
\newblock In {\em Proceedings of the Thirty-Second AAAI Conference on
  Artificial Intelligence and Thirtieth Innovative Applications of Artificial
  Intelligence Conference and Eighth AAAI Symposium on Educational Advances in
  Artificial Intelligence}. AAAI'18/IAAI'18/EAAI'18. AAAI Press.

\bibitem[\protect\citeauthoryear{Siarry, Jabbar, Aluvalu, Abraham, and
  Madureira}{Siarry et~al\mbox{.}}{2021}]{siarry2021fusion}
{\sc Siarry, P.}, {\sc Jabbar, M.}, {\sc Aluvalu, R.}, {\sc Abraham, A.}, {\sc
  and} {\sc Madureira, A.} 2021.
\newblock {\em The Fusion of Internet of Things, Artificial Intelligence, and
  Cloud Computing in Health Care}.
\newblock Internet of Things, Technology, Communications and Computing.
  Springer International Publishing.

\bibitem[\protect\citeauthoryear{Singh and Mishra}{Singh and
  Mishra}{2022}]{Singh_2022}
{\sc Singh, V.~K.} {\sc and} {\sc Mishra, S.} 2022.
\newblock A truthful mechanism for time-bound tasks in {I}o{T}-based
  crowdsourcing with zero budget.

\bibitem[\protect\citeauthoryear{Singh, Mukhopadhyay, Xhafa, and Krause}{Singh
  et~al\mbox{.}}{2020a}]{Singh_2020}
{\sc Singh, V.~K.}, {\sc Mukhopadhyay, S.}, {\sc Xhafa, F.}, {\sc and} {\sc
  Krause, P.} 2020a.
\newblock A quality-assuring, combinatorial auction based mechanism for
  {I}o{T}-based crowdsourcing.
\newblock In {\em Advances in Edge Computing: Massive Parallel Processing and
  Applications}. Vol.~35. IOS Press, 148--177.

\bibitem[\protect\citeauthoryear{Singh, Mukhopadhyay, Xhafa, and Sharma}{Singh
  et~al\mbox{.}}{2020b}]{Singh2019}
{\sc Singh, V.~K.}, {\sc Mukhopadhyay, S.}, {\sc Xhafa, F.}, {\sc and} {\sc
  Sharma, A.} 2020b.
\newblock A budget feasible peer graded mechanism for iot-based crowdsourcing.
\newblock {\em Journal of Ambient Intelligence and Humanized Computing\/}~{\em
  11,\/}~4, 1531--1551.

\bibitem[\protect\citeauthoryear{Staniek}{Staniek}{2021}]{STANIEK2021554}
{\sc Staniek, M.} 2021.
\newblock Road pavement condition diagnostics using smartphone-based data
  crowdsourcing in smart cities.
\newblock {\em Journal of Traffic and Transportation Engineering (English
  Edition)\/}~{\em 8,\/}~4, 554--567.

\bibitem[\protect\citeauthoryear{Sun, Sun, Yu, and Guizani}{Sun
  et~al\mbox{.}}{2020}]{8891679}
{\sc Sun, G.}, {\sc Sun, S.}, {\sc Yu, H.}, {\sc and} {\sc Guizani, M.} 2020.
\newblock Toward incentivizing fog-based privacy-preserving mobile crowdsensing
  in the internet of vehicles.
\newblock {\em IEEE Internet of Things Journal\/}~{\em 7,\/}~5, 4128--4142.

\bibitem[\protect\citeauthoryear{Sun and Li}{Sun and
  Li}{2020}]{10.1007/s10479-018-2826-y}
{\sc Sun, J.} {\sc and} {\sc Li, G.} 2020.
\newblock Designing a double auction mechanism for the re-allocation of
  emission permits.
\newblock {\em Annals of Operations Research\/}~{\em 291}, 847–874.

\bibitem[\protect\citeauthoryear{Talavera, Tobón, Gómez, Culman, Aranda,
  Parra, Quiroz, Hoyos, and Garreta}{Talavera
  et~al\mbox{.}}{2017}]{TALAVERA2017283}
{\sc Talavera, J.~M.}, {\sc Tobón, L.~E.}, {\sc Gómez, J.~A.}, {\sc Culman,
  M.~A.}, {\sc Aranda, J.~M.}, {\sc Parra, D.~T.}, {\sc Quiroz, L.~A.}, {\sc
  Hoyos, A.}, {\sc and} {\sc Garreta, L.~E.} 2017.
\newblock Review of iot applications in agro-industrial and environmental
  fields.
\newblock {\em Computers and Electronics in Agriculture\/}~{\em 142}, 283--297.

\bibitem[\protect\citeauthoryear{Tan and Jiang}{Tan and
  Jiang}{2019}]{10.1145/3371425.3371459}
{\sc Tan, W.} {\sc and} {\sc Jiang, Z.} 2019.
\newblock A novel experience-based incentive mechanism for mobile crowdsensing
  system.
\newblock In {\em Proceedings of the International Conference on Artificial
  Intelligence, Information Processing and Cloud Computing}. AIIPCC '19.
  Association for Computing Machinery, New York, NY, USA.

\bibitem[\protect\citeauthoryear{Wang, Guo, Cao, and Guo}{Wang
  et~al\mbox{.}}{2018}]{DBLP:journals/tpds/WangGCG18}
{\sc Wang, H.}, {\sc Guo, S.}, {\sc Cao, J.}, {\sc and} {\sc Guo, M.} 2018.
\newblock Melody: {A} long-term dynamic quality-aware incentive mechanism for
  crowdsourcing.
\newblock {\em {IEEE} Trans. Parallel Distrib. Syst.\/}~{\em 29,\/}~4,
  901--914.

\bibitem[\protect\citeauthoryear{Wei, Zhu, Zhu, Zhang, and Xue}{Wei
  et~al\mbox{.}}{2015}]{7218592}
{\sc Wei, Y.}, {\sc Zhu, Y.}, {\sc Zhu, H.}, {\sc Zhang, Q.}, {\sc and} {\sc
  Xue, G.} 2015.
\newblock Truthful online double auctions for dynamic mobile crowdsourcing.
\newblock In {\em 2015 IEEE Conference on Computer Communications (INFOCOM)}.
  IEEE, Hong Kong, China, 2074--2082.

\bibitem[\protect\citeauthoryear{Wu, Wang, Hu, Lepine, Na, Ainalis, and
  Stettler}{Wu et~al\mbox{.}}{2020}]{s20195564}
{\sc Wu, C.}, {\sc Wang, Z.}, {\sc Hu, S.}, {\sc Lepine, J.}, {\sc Na, X.},
  {\sc Ainalis, D.}, {\sc and} {\sc Stettler, M.} 2020.
\newblock An automated machine-learning approach for road pothole detection
  using smartphone sensor data.
\newblock {\em Sensors\/}~{\em 20,\/}~19, 1--23.

\bibitem[\protect\citeauthoryear{W.Vickery}{W.Vickery}{1961}]{WVic_Fin_1961}
{\sc W.Vickery}. 1961.
\newblock Counterspeculation,auctions and competitive sealed tenders.
\newblock {\em Journal of Finance.\/}~{\em 16,\/}~1, 837.

\bibitem[\protect\citeauthoryear{Yan, Zhang, and Vasilakos}{Yan
  et~al\mbox{.}}{2014}]{YAN2014120}
{\sc Yan, Z.}, {\sc Zhang, P.}, {\sc and} {\sc Vasilakos, A.~V.} 2014.
\newblock A survey on trust management for internet of things.
\newblock {\em Journal of Network and Computer Applications\/}~{\em 42},
  120--134.

\bibitem[\protect\citeauthoryear{Yang, Xue, Fang, and Tang}{Yang
  et~al\mbox{.}}{2012}]{Yang2012MCN}
{\sc Yang, D.}, {\sc Xue, G.}, {\sc Fang, X.}, {\sc and} {\sc Tang, J.} 2012.
\newblock Crowdsourcing to smart-phones: Incentive mechanism design for mobile
  phone sensing.
\newblock In {\em Proceeding of 18th Annual International Conference on Mobile
  Computing and Networking}. ACM, Istanbul, Turkey, 173--184.

\bibitem[\protect\citeauthoryear{Yin and Chen}{Yin and
  Chen}{2015}]{10.5555/2832249.2832277}
{\sc Yin, M.} {\sc and} {\sc Chen, Y.} 2015.
\newblock Bonus or not? learn to reward in crowdsourcing.
\newblock In {\em Proceedings of the 24th International Conference on
  Artificial Intelligence}. IJCAI'15. AAAI Press, Buenos Aires, Argentina,
  201–207.

\bibitem[\protect\citeauthoryear{Zanella, Bui, Castellani, Vangelista, and
  Zorzi}{Zanella et~al\mbox{.}}{2014}]{6740844}
{\sc Zanella, A.}, {\sc Bui, N.}, {\sc Castellani, A.}, {\sc Vangelista, L.},
  {\sc and} {\sc Zorzi, M.} 2014.
\newblock Internet of things for smart cities.
\newblock {\em IEEE Internet of Things Journal\/}~{\em 1,\/}~1, 22--32.

\bibitem[\protect\citeauthoryear{Zhang}{Zhang}{2011}]{6158307}
{\sc Zhang, L.} 2011.
\newblock An iot system for environmental monitoring and protecting with
  heterogeneous communication networks.
\newblock In {\em 2011 6th International ICST Conference on Communications and
  Networking in China (CHINACOM)}. IEEE, Harbin, China, 1026--1031.

\bibitem[\protect\citeauthoryear{Zhou, Cao, Dong, and Vasilakos}{Zhou
  et~al\mbox{.}}{2017}]{7823334}
{\sc Zhou, J.}, {\sc Cao, Z.}, {\sc Dong, X.}, {\sc and} {\sc Vasilakos, A.~V.}
  2017.
\newblock Security and privacy for cloud-based iot: Challenges.
\newblock {\em IEEE Communications Magazine\/}~{\em 55,\/}~1, 26--33.

\bibitem[\protect\citeauthoryear{Zhou and Zheng}{Zhou and
  Zheng}{2009}]{5062011}
{\sc Zhou, X.} {\sc and} {\sc Zheng, H.} 2009.
\newblock Trust: A general framework for truthful double spectrum auctions.
\newblock In {\em IEEE INFOCOM 2009}. IEEE, Rio de Janeiro, Brazil, 999--1007.

\end{thebibliography}

\end{document}